\PassOptionsToPackage{table}{xcolor}
\documentclass[a4paper, 11pt]{article}

\title{Descriptive complexity for neural networks via Boolean networks} 

\date{} 					

\usepackage{authblk}

\author[ \hspace{-1ex}]{Veeti Ahvonen}
\author[ \hspace{-1ex}]{Damian Heiman}
\author[ \hspace{-1ex}]{Antti Kuusisto}

\affil[ \hspace{-1ex}]{Mathematics Research Centre, Tampere University, Finland}

\usepackage{graphicx} 

\usepackage[textwidth=170mm, hcentering, vmargin={21mm,27.7mm}, foot=15mm]{geometry}

\usepackage{hyperref}
\usepackage{verbatim}

\usepackage[utf8]{inputenc}
\usepackage[T1]{fontenc}
\usepackage{amsmath, amssymb, amsthm}
\usepackage{enumitem}
\setitemize{noitemsep,topsep=0pt,parsep=3pt,partopsep=0pt}
\setenumerate{noitemsep,topsep=0pt,parsep=3pt,partopsep=0pt}
\usepackage{tikz}
\usepackage{lipsum}
\usepackage{pdfpages}
\usepackage{multicol} 
\usepackage{stmaryrd}
\usepackage{caption}
\usepackage{soul}
\usepackage{makecell}
\usepackage{booktabs}

\usepackage{arydshln}
\usepackage{makecell}
\usepackage{mathtools}

\usepackage{marginnote}

\usepackage{colonequals}

\usepackage{multirow}

\theoremstyle{plain}
\newtheorem{theorem}{Theorem}[section]
\newtheorem{lemma}[theorem]{Lemma}
\newtheorem{corollary}[theorem]{Corollary}

\theoremstyle{definition}
\newtheorem{remark}[theorem]{Remark}
\newtheorem{example}[theorem]{Example}
\newtheorem{definition}[theorem]{Definition}

\newcommand{\N}{\mathbb N}
\newcommand{\Z}{\mathbb Z}

\newcommand{\R}{\mathbb R}

\newcommand{\ba}{\mathbf{a}}
\newcommand{\bb}{\mathbf{b}}
\newcommand{\bc}{\mathbf{c}}

\newcommand{\be}{\mathbf{e}}
\newcommand{\bbf}{\mathbf{f}}
\newcommand{\bg}{\mathbf{g}}

\newcommand{\bi}{\mathbf{i}}
\newcommand{\bj}{\mathbf{j}}

\newcommand{\bm}{\mathbf{m}}

\newcommand{\bp}{\mathbf{p}}

\newcommand{\bs}{\mathbf{s}}
\newcommand{\bt}{\mathbf{t}}

\newcommand{\bx}{\mathbf{x}}
\newcommand{\by}{\mathbf{y}}
\newcommand{\bz}{\mathbf{z}}

\newcommand{\cA}{\mathcal{A}}
\newcommand{\cB}{\mathcal{B}}
\newcommand{\cC}{\mathcal{C}}

\newcommand{\cF}{\mathcal{F}}

\newcommand{\cI}{\mathcal{I}}

\newcommand{\cN}{\mathcal{N}}

\newcommand{\cP}{\mathcal{P}}

\newcommand{\cS}{\mathcal{S}}
\newcommand{\cT}{\mathcal{T}}

\newcommand{\cX}{\mathcal{X}}

\newcommand{\cZ}{\mathcal{Z}}

\newcommand{\fF}{\mathfrak{F}}

\newcommand{\fS}{\mathfrak{S}}

\newcommand{\fa}{\mathfrak{a}}
\newcommand{\fb}{\mathfrak{b}}

\newcommand{\fw}{\mathfrak{w}}

\newcommand*{\abs}[1]{\lvert#1\rvert}   

\newcommand{\msc}{\mathrm{MSC}}

\newcommand{\SC}{\mathrm{SC}}
\newcommand{\GMSC}{\mathrm{GMSC}}

\newcommand{\prop}{\mathrm{PROP}}
\newcommand{\var}{\mathrm{VAR}}

\newcommand{\GNN}{\mathrm{GNN}}

\newcommand{\ordo}{\mathcal{O}}

\newcommand{\BNL}{\mathrm{BNL}}
\newcommand{\ReLU}{\mathrm{ReLU}}

\date{}

\begin{document}

\maketitle

\begin{abstract}
    \noindent
    We investigate the expressive power of neural networks from the point of view of descriptive complexity. 
    We study neural networks that use floating-point numbers and piecewise polynomial activation functions from two perspectives: 1) the general scenario where neural networks run for an unlimited number of rounds and have unrestricted topologies, and 2) classical feedforward neural networks that have the topology of layered acyclic graphs and run for only a constant number of rounds. We characterize these neural networks via Boolean networks formalized via a recursive rule-based logic. In particular, we show that the sizes of the neural networks and the corresponding Boolean rule formulae are polynomially related. In fact, in the direction from Boolean rules to neural networks, the blow-up is only linear. Our translations result in a time delay, which is the number of rounds that it takes to simulate a single computation step. In the translation from neural networks to Boolean rules, the time delay of the resulting formula is polylogarithmic in the size of the neural network. In the converse translation, the time delay of the neural network is linear in the formula size. Ultimately, we obtain translations between neural networks, Boolean networks, the diamond-free fragment of modal substitution calculus, and a class of recursive Boolean circuits. Our translations offer a method, for almost any activation function F, of translating any neural network in our setting into an equivalent neural network that uses F at each node. This even includes linear activation functions, which is possible due to using floats rather than actual reals!
\end{abstract}

\section{Introduction}

This article investigates the expressive power of neural networks from the perspective of descriptive complexity, giving a logical characterization for general neural networks with unlimited topologies and unbounded running time. By unlimited topologies, we mean that any neural network built over a directed graph is allowed, enabling recurrence within computations.
We also obtain a logical characterization for classical feedforward neural networks.
The characterizations are based on \emph{Boolean networks} \cite{cheng2010linear, schwab2020concepts}. Boolean networks have a long history, originating from the work of Kauffman in the 1960s \cite{kauffman1969homeostasis}. Current applications include a wide variety of research relating to topics varying from biology and medicine to telecommunications and beyond, see, e.g., \cite{zhang2008network, zanin2011boolean, schwab2020concepts}.

Boolean networks are usually not defined via a strict logical syntax, but it is easy to give them one as follows. Consider a set $\cT = \{X_1,\dots , X_k\}$ of Boolean variables. A \emph{Boolean rule} over $\cT$ is an expression of the form $X_i \colonminus \varphi$ consisting of a \emph{head predicate} $X_i\in \cT$ and a \emph{body} $\varphi$, which is a Boolean formula over the syntax $\varphi \coloncolonequals \top \mid X_j \mid \neg \varphi \mid \varphi\wedge \varphi$, where $X_j\in \cT$. A \emph{Boolean program} over $\cT$ is then a finite set of Boolean rules over $\cT$, one rule for each $X_i \in \cT$. 
Given an input $f \colon \cT \rightarrow \{0,1\}$, and executing the rules in parallel, the program then produces a time series of $k$-bit strings in a natural way (see the preliminaries section for the details). 
It is also possible to extend Boolean networks as follows. An extended Boolean program over~$\cT$ is a Boolean program over some finite $\cS \supseteq \cT$ together with a \emph{base rule} 
$X_j(0) \colonminus b$ for each $X_j\in \cS \setminus \cT$, where $b\in \{\top, \bot\}$. The predicates in $\cT$ are called \emph{input predicates}, and the predicates in $\cS \setminus \cT$ are \emph{auxiliary predicates}.
Extended programs produce time series just like regular programs, 
but the initial values of auxiliary predicates are not part of the input, but are instead given via the base rules (see the preliminaries section). The logic used in this paper, \textbf{Boolean network logic} BNL, consists of extended Boolean programs.
Using extended programs is necessary, as ordinary ones would not suffice for the characterization.

As already mentioned, the neural network (NN) architecture that we consider is very general.
We allow unrestricted topologies for the underlying directed graphs, including cycles and even reflexive loops, thereby considering the recurrent setting.
The edges and nodes of the underlying graph are weighted, and each node has an activation function that is piecewise polynomial. 
Reals are modeled via floating-point numbers, and the running times are unlimited.
Each NN is associated with a set of input nodes and an initialization function $\pi$ which gives initial values to non-input nodes (called \emph{auxiliary nodes}).
In each round, an NN computes a floating-point value for each node as follows.
In round $0$, each input node receives a value as an external input, and each auxiliary node is given a value by $\pi$.
In every subsequent round, the nodes update their activation values via the standard forward propagation scheme (also known as perceptron update) applied repeatedly (see Section \ref{section: general neural networks} for the details).

We prove that each NN translates to a Boolean program that simulates the time series of the NN with each input, and vice versa, Boolean programs can---likewise---be simulated by NNs. As a central part of this result, we also establish that the size blow-ups are mild in both translations.
In more detail, let $\cF(p,q,\beta)$ denote a floating-point format with \emph{fraction precision} $p$, \emph{exponent precision} $q$ and \emph{base} $\beta$
(see Section \ref{fpa} for the details). 
The \emph{degree} of an NN is the maximum degree of the underlying graph and the \emph{order} (resp. \emph{piece-size}) of an NN is the maximum order (resp. maximum number of pieces) of the NN's piecewise polynomial activation functions.
The below theorem establishes the first direction of our translations, showing that only a polynomial blow-up (in all involved parameters) is required when translating NNs to Boolean programs.

\smallskip
\textbf{Theorem\ \ref{NN_to_BNL}.}\ 
\emph{Let $\cF(p, q, \beta)$ be a floating-point format, let $r = \max\{p,q\}$, and let $\cN$ be a general neural network for $\cF$ with $N$ nodes, $L$ edges, degree $\Delta$ and piece-size $P$. Let $\Omega$ be the order of $\cN$ (setting $\Omega := 2$ if the order is $1$ or $0$). We can construct an equivalent{} $\BNL$-program $\Lambda$ such that
    \begin{itemize}
        \item the size of $\Lambda$ is 
        $\ordo((N P \Omega + L) (r^4 + r^3\beta + r^{2}\beta^2 + r\beta^3))$
        \item there are 
        $\ordo((N P\Omega + L)(r^{2} \beta + r \beta^{2}))$ 
        head predicates in $\Lambda$, and
        \item the computation delay of $\Lambda$ is $\ordo(\log(\Omega)(\log(r) + \log(\beta)) + \log(\Delta))$.
    \end{itemize}}
\smallskip 

Here, and also in the following theorems, \emph{equivalence{}} between a Boolean program and a neural network (or between two neural networks) means that the time series they produce are essentially identical. They may use different symbols, and there may be a constant computation time delay between the recorded steps; see Section \ref{subsection: run sequence equivalence} for the details.

In the converse direction, from Boolean programs to NNs, the choice of the activation function in the target NNs is only very slightly restricted. For example, it suffices that all target NNs use $\ReLU$ in all their nodes.
However, to give a more general result, 
we provide sufficient conditions that activation functions must satisfy to enable our translation; these come in the form of \emph{Axiomatization~\ref{A}}, which covers a very comprehensive list of activation functions including, e.g., the rectified linear unit $\ReLU$, sigmoid function, hyperbolic tangent and various others. We say that such activation functions are \emph{regular}.
The result below covers the second direction of our translation, showing that Boolean programs translate---with only a very mild (linear) blow-up---to NNs. The activation function can be chosen freely, as long as it is regular.
To formulate the result, let the depth of a program refer to the maximum nesting depth of Boolean operations appearing in the Boolean rules of the program.

\smallskip

\noindent
\textbf{Theorem\ \ref{BNL_to_NN}.}\
\emph{For any regular activation function $\alpha$, a $\BNL$-program of size $s$ and depth $d$ translates to an equivalent general neural network with $\ordo(s)$ nodes, computation delay $\ordo(d)$ and $\alpha$ at each node.}

Typically, floating-point formats are fixed. This also fixes an upper bound for the number of possible activation functions, so it is reasonable to assume that also the piece-size and order of the activation functions are fixed.\footnote{There are, however, infinitely many syntactic representations of such functions; see Section \ref{section: float polynomials} for more details.} In any case, in typical real-life scenarios, function approximations based on polynomials improve fast when the number of polynomial terms is increased.
We obtain the following rather simple fixed-parameter linearity result from Theorems \ref{NN_to_BNL} and \ref{BNL_to_NN} (where the direction obtained from Theorem~\ref{BNL_to_NN} in fact remains unchanged).

\textbf{Theorem \ref{theorem:NN_to_BNL_fixed}.}\ \emph{Given a fixed floating-point format and fixed parameters for the activation functions, the following hold.
\begin{itemize}
    \item A general neural network with $N$ nodes,
    $L$ edges
    and degree $\Delta$ translates to an equivalent{} $\BNL$-program of size 
    $\ordo(N + L)$ 
    and computation delay $\ordo(\log(\Delta))$.
    \item For any regular activation function $\alpha$, a $\BNL$-program of size $s$ and depth $d$ translates to an equivalent general neural network with $\ordo(s)$ nodes, computation delay $\ordo(d)$ and $\alpha$ at each~node.
\end{itemize}}

We also consider classical feedforward neural networks (FFNNs) \cite{rumelhart1986learning}. Informally, FFNNs
have the topology of acyclic graphs, in which the nodes can be partitioned into sequential non-empty layers. Synchronous FFNNs are FFNNs that do not allow connections that skip layers.
We also study feedforward Boolean programs (FFBNL) which are defined analogously to FFNNs.
Similarly to Theorems \ref{NN_to_BNL} and \ref{BNL_to_NN}, we obtain a characterization for the feedforward scenario (see Theorems \ref{theorem:FFNN_to_FFBNL} and \ref{theorem:FFBNL_to_FFNN}).

Combining the above-described results, we establish that one can replace activation functions in NNs by any regular activation function. The result, given below, covers both the general and feedforward scenario.

\noindent
\textbf{Theorem\ \ref{theorem:NN_to_NN}.}\ 
\emph{Let $\cN$ be a general neural network. For any regular activation function $\alpha$, we can translate $\cN$ to an equivalent general neural network $\cN'$ which uses $\alpha$ (at each node) and:
\begin{itemize}
    \item the number of nodes in $\cN'$ is polynomial in the size parameters of $\cN$,
    \item the computation delay of $\cN'$ is polylogarithmic in the size parameters of $\cN$,
    \item if $\cN$ is feedforward, then $\cN'$ is feedforward (and even synchronous, i.e., no layer-skipping).
\end{itemize}
}

In particular, Theorem\ \ref{theorem:NN_to_NN} holds when the regular activation function is the 
\emph{identity function} (cf. Case~$2$ of 
Axiomatization \ref{A}). This is striking, as this appears to translate NNs with non-linear activations to linear functions, but the explanation lies in the properties of floating-point arithmetic, where significant non-linear phenomena arise easily.
One of the primary purposes of activations is to non-linearize NNs, but our result shows that they nevertheless can be simulated by using floats in a suitable way.

It is worth noting that in our setting, while we allow for general topologies and unlimited running times, our systems have---strictly speaking---finite input spaces because our NNs use floating-point numbers. Note that, trivially, a single Boolean function suffices to model any NN with a finite input space when limiting to \emph{single outputs only} and 
\emph{not caring about size and time blow-ups} in translations.
However, our results stress the close relations between the size and time resources of general NNs and $\BNL$-programs, and as outputs, we consider time series rather than a single-output framework. Therefore, our equivalence notion is based on a similarity between time-series, which makes the related study more involved than simply considering the case where both inputs and outputs are finite strings.

It turns out that $\BNL$ is also closely related to two other frameworks: 
\emph{substitution calculus} (SC) which is 
the diamond-free fragment of \emph{modal substitution calculus} MSC (cf. \cite{Kuusisto13} where MSC was introduced and used to characterize message passing automata), and \emph{self-feeding circuits} (SFCs), which are formally introduced here and related to sequential circuits \cite{CavanaghJoseph2006SLAa}. 
Informally, SC is otherwise similar to $\BNL$, except that \emph{every} head predicate in an $\SC$-program has a base rule, the bodies of the base rules are allowed to be any formulae of propositional logic, Boolean rules may contain proposition symbols in addition to Boolean variables, and the inputs for $\SC$-programs are simply assignments of truth values to proposition symbols (see Section \ref{MSC} for the details). Due to the technical choices in its definition, the logic $\SC$ does not need separate notions of auxiliary and input predicates required in the transition from Boolean programs to extended Boolean programs, i.e., BNL.
An $\SC$-program produces a time series similarly to $\BNL$.

Self-feeding circuits are Boolean circuits, where the number of input and output gates match. Moreover, each SFC is associated with an initializing function that fixes the input for some subset $A$ of the input gates; the gates in $A$ are \emph{auxiliary positions}, while the input gates not in $A$ are \emph{input positions}. 
An SFC produces a time series in a natural way as follows. 
In round zero, the input positions receive their initial values as an external input, whereas the initialization function fixes the initial values for auxiliary positions.
In each subsequent round, the SFC feeds its output string from the previous round to itself as input.
See Section \ref{section:self-feeding_circuits} for the technical details on SFCs.
We observe that $\BNL$, $\SC$ and SFCs translate to each other with only a linear size increase in every direction.
As a corollary, we obtain a characterization for neural networks via SC and SFCs.
For a brief overview comparing the studied computation models--$\BNL$, FFBNL, $\SC$, SFCs, NNs and FFNNs--see Table \ref{table: all the stuff}.
\begin{table}[h]
\caption{A summary of the finer differences of the studied computation architectures.
Here ``auxiliary entry receiver'' refers to
a method of assigning fixed initial values to variables, nodes and gates.
}
\resizebox{\columnwidth}{!}{%
\begin{tabular}{|c|c|c|c|c|c|}
    \hline
    &\textbf{Recursiveness} &\textbf{Domain base} &\textbf{Input receiver} & \textbf{Auxiliary entry receiver} &\textbf{Update scheme}  \\
    \hline
    $\BNL$ &Recurrent &Bits &Input predicates &Auxiliary predicates &Boolean rules \\
    \hline
    FF$\BNL$ &Feedforward &Bits &Input predicates &None &Boolean rules \\
    \hline
    $\SC$ &Recurrent &Bits &Proposition symbols &Inherent &Boolean rules \\
    \hline
    SFCs &Recurrent &Bits &Input positions &Auxiliary positions &Boolean circuit \\
    \hline
    NNs &Recurrent &Floats &Input nodes &Auxiliary nodes &Perceptron update \\
    \hline
    FFNNs &Feedforward &Floats &Input nodes &None &Perceptron update \\
    \hline
\end{tabular}%
}
\label{table: all the stuff}
\end{table}

\textbf{Related work.}
The present paper is a revised and extended version of the conference article~\cite{Ahvonen_CSL}. We have generalized the translation from Boolean programs to neural networks such that it now applies for a much wider range of activation functions than the two---the rectified linear unit ReLU and Heaviside---studied in \cite{Ahvonen_CSL}. We have added Section~\ref{sec: fnn}, where we provide a new characterization for \emph{feedforward} neural networks via feedforward Boolean programs. In Section \ref{sec: corollaries}, we give further new results on translating general NNs to NNs having particular regular activation functions.

We will next discuss characterizations and translations between neural networks, logics and Boolean circuits in the notable articles \cite{mcculloch1943logical, shawe1992classes, BEIU19961155, choi2017compiling, bertossi2023compiling, watta1997recurrent} where---unlike in our work---the neural networks had binary inputs and, when translating from logic or circuits to neural networks, they only had threshold activation functions. The majority of these characterizations considered only feedforward neural networks, and the sizes of the translations were not thoroughly analyzed. The last three of these also only contained a translation in one direction.

An early logical characterization of (feedforward) neural networks was given by McCulloch and Pitts in \cite{mcculloch1943logical} which characterized a class of Boolean-valued neural networks (neural networks with activation values $0$ or~$1$, not to be confused with Boolean networks) with threshold activation functions via \emph{temporal propositional expressions} (TPEs). They also extended this characterization to include recurrent neural networks with the same restrictions via an extension of TPEs. 
In \cite{shawe1992classes}, Shawe-Taylor et al. gave a characterization for neural networks via Boolean circuits. 
The networks had binary strings as inputs and single-bit outputs.
They also had bounded fan-in, bounded depth, monotonic activation functions and a single output node with a threshold activation function (analogous to threshold gates in Boolean circuits).
The authors of that paper defined classes $\mathrm{NN}^k$ of such neural networks and placed them into the existing hierarchy of circuit complexity classes, showing that $\mathrm{NC}^k \subseteq \mathrm{NN}^k \subseteq \mathrm{AC}^k$. 
Recall here that $\mathrm{AC}^k$ contains the families of Boolean circuits of depth $\ordo(\log^k n)$ and polynomial size with unbounded fan-in, while $\mathrm{NC}^k$ 
restricts $\mathrm{AC}^k$ to circuits with bounded fan-in.
They also obtained translations in the scenario where the neural networks only used threshold activation functions, in which case the accuracy of the neural networks was allowed to be unlimited, i.e., the networks used exact reals instead of finite representations. In~\cite{BEIU19961155}, Beiu and Taylor extended the results of \cite{shawe1992classes} to cover less restrictive classes of neural networks, showing a wide range of inclusions with other established circuit complexity classes; for instance, it was shown that $\mathrm{NC}^k \subseteq \mathrm{NN}^k \subseteq \mathrm{NC}^{k+1}$.

Choi et al. in \cite{choi2017compiling} and Bertossi and León in \cite{bertossi2023compiling} gave translations (one direction only) from Boolean-valued neural networks with threshold activation functions into ``explainable Boolean circuits''.
In \cite{watta1997recurrent}, Watta et al. 
showed that recurrent neural networks with threshold activation functions and binary inputs can be translated into Boolean networks. They used this to study the stability of such networks, i.e., which of them converge to a fixed-point.

Concerning further related work, the descriptive complexity of graph neural networks (or GNNs) has been studied in the following papers.
Barceló et al. in \cite{barcelo} established a match between constant-time GNNs and graded modal logic, and also proved that each property expressible in the two-variable fragment of first-order logic with counting is expressible as a constant-time GNN with global readout. 
Moreover, Grohe studied GNNs in the papers \cite{grohe} and \cite{grohe2023descriptive}. In \cite{grohe}, a characterization for GNNs was given via the Weisfeiler-Leman algorithm \cite{weisfeiler1968reduction}.
In \cite{grohe2023descriptive},
constant-time GNNs with \emph{dyadic rationals} and \emph{rational piecewise linear approximable} activation functions were 
characterized with the guarded fragment of first-order logic with counting and built-in relations. Furthermore, 
such GNNs augmented with random initialization were characterized with the circuit complexity class $\mathsf{TC}^0$.
Ahvonen et al. in \cite{gnn_neurips} used graded modal substitution calculus $\GMSC$ (an extension of the logics $\SC$ and $\BNL$ used here) to characterize recurrent $\GNN$s with floating-point numbers (see the preliminaries for the definition of $\GMSC$).

Neural networks are special kinds of distributed systems, and the descriptive complexity of distributed computing was initiated in Hella et al. \cite{hella2012weak}, Kuusisto \cite{Kuusisto13} and Hella et al. \cite{weak_models} by relating distributed complexity classes to modal logics. While \cite{hella2012weak} and \cite{weak_models} gave characterizations to constant-time distributed complexity classes via modal logics, \cite{Kuusisto13} lifted the work to general non-constant-time algorithms by characterizing finite distributed message passing automata via the modal substitution calculus MSC which---as discussed above---is an extension of the logic $\SC$ studied in this article. This was recently lifted to circuit-based networks with identifiers in Ahvonen et al. \cite{dist_circ_mfcs, ahvonen_journal}, also utilizing modal substitution calculus. The logic MSC has also been linked to a range of other logics. 
For example, MSC is contained in partial fixed-point logic with choice \cite{DBLP:conf/csl/Richerby04}, and in \cite{Kuusisto13} MSC was shown to contain the $\mu$-fragment of the modal $\mu$-calculus.
Building on the results in \cite{Kuusisto13}, Reiter showed in \cite{reiter} that this fragment of the modal $\mu$-calculus captures finite message passing automata in the asynchronous setting. Finally, the logics BNL and SC used in this article as well as the logic $\msc$ are all rule-based systems. Rule-based logics are used widely in various applications, involving systems such as Datalog, answer-set programming (ASP) formalisms, and many others.

\section{Preliminaries}

First, we introduce some basic concepts.
For any set $S$, we let $\wp(S)$ denote the power set of $S$ and we let $\abs{S}$ denote the size (or cardinality) of $S$. 
We let $\N$ denote the set of non-negative integers, $\Z$ the set of integers, $\Z_+$ the set of positive integers and $\R$ the set of real numbers.
For every $n \in \Z_+$, we let $[n] \colonequals \{1, \ldots, n\}$, and for every $n \in \N$, we let $[0;n] \colonequals \{0, \ldots, n\}$.
We let bold lower-case letters $\ba, \bb, \bc, \dots$ denote strings. The letters of a string are written directly next to each other, i.e. $abc$, or with dots in-between, i.e. $a \cdot b \cdot c$, or a mix of both, i.e. $abc \cdot def$. Omitted segments of strings are represented with three dots, i.e. $abcd \cdots wxyz$. If $\bs = s_1 \cdots s_{k}$ is a string of length $k$, then for any $j \in [k]$, we let $\bs(j)$ denote the letter $s_j$. The alphabet for the strings will depend on the context. 

We let $\var = \{\, V_i \mid i \in \N\,\}$ denote the (countably infinite) set of all \textbf{schema variables}. Mostly, we will use meta variables, $X$, $Y$, $Z$ and so on, to denote symbols in $\var$. 
We assume a linear order $<^{\var}$ over the set $\var$. Moreover, for any set $\cT \subseteq \var$, a linear order $<^{\cT}$ over $\cT$ is induced by $<^{\var}$.
We let $\prop = \{\,p_{i} \mid i \in \N\,\}$ denote the (countably infinite) set of \textbf{proposition symbols} that is associated with the linear order $<^{\prop}$, inducing a linear order $<^{P}$ over any subset $P \subseteq \prop$. We denote finite subsets of proposition symbols by $\Pi \subseteq \prop$. When we talk about \textbf{rounds} in any context, we refer to non-negative integers that are interpreted as discrete steps of a computation.

\subsection{Run sequences}\label{subsec: run sequences}

Next, we consider \emph{discrete time series} of strings over any alphabet $\Sigma$, i.e., infinite string sequences $S = (\bs_i)_{i \in \N}$. These sequences will later be used to describe the operation of logic programs and neural networks in the sense that both logic programs and neural networks generate infinite string sequences. 
To separate important strings of a sequence from less important ones, we need to define when a time series produces an output. Importantly, we allow an arbitrary number of outputs for any time series and informally an output in round $i \in \N$ is just a substring of $\bs_i$. 
We will define two separate general output conditions for time series. First, we give an informal description of the two output conditions. 
\begin{enumerate}
    \item 
    In the first approach, a signal for when to output is given internally as follows. 
    The time series $S$ is associated with a set of distinguished substrings over $\Sigma$, and $S$ outputs in round $i$ if $\bs_i$ contains a distinguished substring.
    \item 
    In the second approach, the output is given externally as follows. The output rounds are fixed, i.e., if $i$ is an output round then $S$ produces an output in round $i$. 
\end{enumerate}
We study the two approaches for the sake of generality. For instance, it is natural to indicate output conditions within a program if it is part of the program's design. Retroactively, it might be more natural to augment a program to draw attention to rounds the original design does not account for, and a different mechanism could be used to compute the output rounds, e.g., a Turing machine.

Now, we formally define the output conditions in the first approach.
Let $\Sigma$ be a finite alphabet and $k \in \Z_+$. An infinite sequence $S = (\bs_j)_{j\in \mathbb{N}}$ of $k$-length strings $\bs_j \in \Sigma^k$, associated with sets $A \subseteq [k]$ of \textbf{attention} positions, $P \subseteq [k]$ of \textbf{print} positions and $\cS \subseteq \Sigma^{\abs{A}}$ of \textbf{attention strings}, is called the \textbf{run sequence of $S$ w.r.t. $(A, P, \cS)$}. 
Thus, formally a run sequence is a tuple $(S, A, P, \cS)$.
We may omit $(A, P, \cS)$ when $A$, $P$ and $\cS$ are clear from the context.
In the case where $\Sigma = \{0,1\}$, we may refer to attention positions as \textbf{attention bits}, and print positions as \textbf{print bits}.
The sets $A$ and $P$ induce corresponding sequences of substrings of the strings in $S$.
More formally, $(\ba_j)_{j\in \mathbb{N}}$ records the substrings with positions in $A$, and $(\bp_j)_{j\in \mathbb{N}}$ records the substrings with positions in $P$.

If $\ba_n \in \cS$ (for some $n \in \N$), then we say that $S$ \textbf{outputs} $\bp_n$ in round $n$ and that $n$ is an \textbf{output round}.
More precisely, $S$ \textbf{outputs in round $n$ with respect to $(A,P,\cS)$}, $\bp_n$ is the \textbf{output of $S$ in round $n$ with respect to $(A,P,\cS)$} and $n$ is an \textbf{output round of $S$ w.r.t. $(A, P, \cS)$}. 
Let $O \subseteq \N$ be the set of output rounds of $S$ w.r.t. $(A, P,\cS)$; it induces a subsequence $(\bs_{j})_{j \in O}$ of $S$.
We call the sequence $(\bp_j)_{j \in O}$ the \textbf{output sequence} of $S$ (w.r.t. $(A, P,\cS)$). 
In the case where $\Sigma = \{0,1\}$ and $\cS = \{0,1\}^{\abs{A}} \setminus \{0\}^{\abs{A}}$, we may simply write $(A, P)$ in place of $(A, P,\cS)$. If $m \in O$ is the $i$th number in $O$ w.r.t. the natural ordering of integers, then we simply call $\bp_m$ the $i$th output of $S$ and $m$ the $i$th output round of $S$.

In the second approach, output rounds are fixed by a set $O \subseteq \N$, and the attention positions and attention strings are excluded. Thus, in this case a run sequence is a tuple $(S, O, P)$. We say that $S$ \textbf{outputs} in rounds $O$. Outputs and output sequences w.r.t. $(O,P)$ are defined analogously to the first approach.

\subsection{Equivalences between run sequences}\label{subsection: run sequence equivalence}

Next, we define highly general notions of equivalence between run sequences. Later these definitions will be applied to our logics and neural networks.

First, we define similarity between two strings w.r.t. 
an embedding from strings to strings.
Let $\Sigma$ and $\Gamma$ be two alphabets and for each alphabet $A$ let $A^* \colonequals \bigcup_{n \in \Z_{+}} A^{n}$. 
Let 
$\sigma \colon \Sigma^* \to \Gamma^*$ 
be an embedding, i.e., an injective function. 
We call 
$\sigma$
a \textbf{similarity} from $\Sigma$ to $\Gamma$, and if 
$\sigma(\bb) = \bc$,
we say that $\bb$ and $\bc$ are \textbf{similar}.
If $\Sigma = \Gamma$, we may call 
$\sigma$
a similarity over $\Sigma$, and then if 
$\sigma$
is the identity function, we call 
$\sigma$
the \textbf{canonical similarity} over $\Sigma$. 
Instead of using embeddings one could define similarities by using equivalence relations between the strings.

Now, we shall define equivalence between run sequences.
We let $B= (\bb_j)_{j \in \N}$ be an infinite sequence of $k$-length strings over $\Sigma$ and $C = (\bc_j)_{j \in \N}$ an infinite sequence of $\ell$-length strings over $\Gamma$. Consider a run sequence $\cB$ of $B$ and a run sequence $\cC$ of $C$. 
Let $\sigma$ be a similarity from $\Sigma$ to $\Gamma$.
\begin{itemize}
    \item We say that $C$ is \textbf{similar} to $B$ (w.r.t. $\sigma$) if for each $j \in \N$, $\bb_j$ and $\bc_j$ are similar w.r.t. $\sigma$. 
    This definition extends for finite sequences of strings in the natural way.
    \item We say that $\cC$ is \textbf{equivalent{} to} $\cB$ (w.r.t. $\sigma$) if the output sequence of $C$ is similar to the output sequence of $B$ (w.r.t. $\sigma$).
    \item We say that $\cC$ is \textbf{strongly equivalent to} $\cB$ (w.r.t. $\sigma$) if $\cC$ is equivalent{} to $\cB$ (w.r.t. $\sigma$), the sequence $C$ is similar to $B$ (w.r.t. $\sigma$), and the output rounds of $B$ and $C$ are the same.
\end{itemize}
We may omit $\sigma$ if it is the canonical similarity over some alphabet. Furthermore, in the statements of later sections, our equivalences are strong in the following sense: we assume that if $\cC$ is equivalent{} or strongly equivalent to $\cB$, then both include attention positions and attention strings, or both use fixed output rounds instead.

Note that given a similarity $\sigma$ from $\Sigma$ to $\Gamma$ and another similarity $\sigma'$ from $\Gamma$ to $\Theta$, the composition $\sigma''$ of $\sigma$ and $\sigma'$ is a similarity from $\Sigma$ to $\Theta$. Thus, if we have run sequences $\cA$, $\cB$ and $\cC$ such that $\cB$ is equivalent{} to $\cA$ w.r.t. $\sigma$ and $\cC$ is equivalent{} to $\cB$ w.r.t. $\sigma'$, then $\cC$ is equivalent{} to $\cA$ w.r.t. $\sigma''$. Moreover, if both equivalences are strong, then $\cC$ is strongly equivalent to $\cA$.

Also, note that equivalence{} places no limitations on the output rounds of the sequences. Thus, 
we shall define a concept of time delay between two equivalent{} run sequences which tells that the output rounds of one run sequence can be computed with a linear function from the output rounds of the other run sequence. 
Assume that the run sequence $\cC$ is equivalent{} to the run sequence $\cB$ (w.r.t. $\sigma$). Let $b_1, b_2, \ldots$ and $c_1, c_2, \ldots$ enumerate the output rounds of $\cB$ and $\cC$ respectively in ascending order. Moreover, assume that the cardinalities of the sets of output rounds are the same and $b_n \leq c_n$ for every $n \in \N$. If $T \in \Z_+$, $t \in \N$ satisfy the equation $T \cdot b_{n} + t = c_{n}$ for every $n \in \N$, then we say that the \textbf{computation delay} of the run sequence $\cC$ (w.r.t. $\cB$) is $T$ with $t$ \textbf{precomputation rounds}. When we do not mention computation delay, the computation delay is $1$, and when we do not mention the number of precomputation rounds, there are $0$ precomputation rounds.

\begin{remark}\label{remark: outputs and attention}
We make the following note about the case where the sets of output rounds are fixed externally: if a run sequence $\cB$ with the set $O$ of output rounds has an equivalent{} run sequence $\cC$ where the computation delay of $\cC$ (w.r.t. $\cB$) is $T$ with $t$ precomputation rounds, then the set $O'$ of output rounds of $\cC$ can be constructed from the set $O$ by $O' \colonequals \{\, Tn + t \mid n \in O \,\}$.
\end{remark}

\subsection{Modal substitution calculus MSC and its variants}\label{MSC}

In this section, we introduce the logics we use for our characterizations. We start by defining graded modal substitution calculus $\GMSC$, which is a generalization of modal substitution calculus $\msc$ introduced
in \cite{Kuusisto13}. We do this in order to connect the paper to surrounding work, as $\GMSC$ was recently linked to graph neural networks in \cite{gnn_neurips}, and $\msc$ was also recently connected to message passing circuits in \cite{dist_circ_mfcs}. Then we introduce the logics relevant for this particular paper: the diamond-free fragment of $\GMSC$ called substitution calculus $\SC$, as well as an equivalent logic for Boolean networks called Boolean network logic $\BNL$.

Let $\Pi \subseteq \prop$ be a finite set of proposition symbols and $\cT \subseteq \var$. A \textbf{base rule} (over $(\Pi, \cT)$) is a string of the form $V_i(0) \colonminus \varphi$, where $V_i \in \cT$ and $\varphi$ is a formula of \textbf{graded modal logic} over $\Pi$ defined over the language 
\[
\varphi \coloncolonequals \top \mid p_i \mid \neg \varphi \mid \varphi \land \varphi \mid \Diamond_{\geq k} \varphi
\]
where $p_i \in \Pi$ and $k \in \N$ (see e.g. \cite{blackburn2001modal} for details about graded modal logic). An \textbf{induction rule} (over $(\Pi, \cT)$) is a string of the form $V_{i} \colonminus \psi$ where $\psi$ is a \textbf{$(\Pi, \cT)$-schema of graded modal substitution calculus} (or $\GMSC$) defined over the language 
\[
\psi \coloncolonequals \top \mid p_{i} \mid V_i \mid \neg \psi \mid \psi \land \psi \mid \Diamond_{\geq k} \psi
\]
where $p_{i} \in \Pi$, $k \in \N$ and $V_i \in \cT$. We also use symbols $\bot, \lor, \rightarrow$ and $\leftrightarrow$ as shorthand in the usual way. 
In a base rule $V_i(0) \colonminus \varphi$ or an induction rule $V_i \colonminus \psi$, the schema variable $V_i$ is the \textbf{head predicate} and $\varphi$ or $\psi$ is the \textbf{body} of the rule.

Let $\cT = \{X_1, \ldots, X_n\}$ be a finite nonempty set of $n \in \Z_{+}$ distinct schema variables. A $(\Pi, \cT)$-\textbf{program} $\Lambda$ of $\GMSC$ is a list of base and induction rules over $(\Pi, \cT)$:
\[
\begin{aligned}
    &X_{1}(0) \colonminus \varphi_{1}\qquad\qquad & &X_{1} \colonminus \psi_{1}, \\
    &\vdots & &\vdots \\
    &X_{n}(0) \colonminus \varphi_{n} & &X_{n} \colonminus \psi_{n},
\end{aligned}
\]
where each schema variable in $\cT$ has precisely one base rule and one induction rule. 
\textbf{Modal substitution calculus} $\msc$ is obtained from $\GMSC$ by only allowing the use of diamonds $\Diamond_{\geq 1}$ instead of $\Diamond_{\geq k}$ for any $k \in \N$.
The diamond-free fragment of $\GMSC$ and $\msc$ called \textbf{substitution calculus} $\SC$ simply restricts the base and induction rules, not allowing any diamonds $\Diamond_{\geq k}$. Note that in $\SC$ the bodies of base rules are thereby formulae of propositional logic.

Moreover, the programs of $\GMSC$ (and also $\msc$ and $\SC$) are associated with a set $\cP \subseteq \cT$ of \textbf{print predicates} and optionally a set $\cA \subseteq \cT$ of \textbf{attention predicates} that determine the output rounds of the program. 
If there are no attention predicates, the output rounds can be given externally with an \textbf{attention function} $a \colon \{0,1\}^k \to \wp(\N)$, where $k$ is the number of distinct proposition symbols that appear in the program. 
Informally, the attention predicates and the attention function are analogous to the two output conditions discussed for run sequences. We will later discuss how either can be used to determine a set of output rounds for the program. We use attention predicates by default, and only discuss the attention function when specified.
More formally, a program $\Lambda$ with an attention function $a$ and no attention predicates is a pair $(\Lambda, a)$, but we may simply refer to such a pair as a program $\Lambda$. However, whenever necessary and important for the context, we will tell whether $\Lambda$ is associated with an attention function or not.

The semantics of a $(\Pi, \cT)$-program $\Lambda$ of $\GMSC$ is defined over Kripke-models. A \textbf{Kripke-model} over $\Pi$ (or simply \textbf{$\Pi$-model}) is a tuple $M = (W, R, V)$, where $W$ is a set of \textbf{nodes} (sometimes referred to as ``worlds'' in literature), $R \subseteq W \times W$ is an \textbf{accessibility relation} and $V \colon \Pi \to \wp(W)$ is a \textbf{valuation}. If $w \in W$, we call $(M,w)$ a \textbf{pointed} Kripke model (over $\Pi$).
The set of \textbf{out-neighbours} of a node $v \in W$ is $\{\, u \in W \mid (v, u) \in R\,\}$.
We first have to define semantics for $(\Pi, \cT)$-schemata.
The \textbf{truth} of a $(\Pi, \cT)$-schema $\psi$ in round $n \in \N$ w.r.t. $\Lambda$ (written $M, w \models \psi^{n}$) is defined as follows:
\begin{itemize}
    \item $M, w \models \top^{n}$ (i.e. $\top$ is always true regardless of $M$, $w$ or $n$). 
    \item $M, w \models p_{i}^{n}$ iff $w \in V(p_{i})$
    where $p_i \in \Pi$ (i.e. the truth of $p_{i}$ does not depend on $n$).
    \item $M, w \models (\neg \theta)^{n}$ iff $M, w \not \models \theta^{n}$.
    \item $M, w \models (\chi \land \theta)^{n}$ iff $M, w \models \chi^{n}$ and $M, w \models \theta^{n}$.
    \item $M, w \models (\Diamond_{\geq k} \psi)^{n}$ iff there are at least $k$ distinct out-neighbours $v_1, \ldots, v_k$ of the node $w$ such that $M, v_i \models \psi^{n}$ for every $i \in [k]$.
    \item We define the truth of $X_i$ depending on $n$ as follows.
    \begin{itemize}
        \item If $n = 0$, then $M, w \models X_{i}^{0}$ iff $M, w \models \varphi_{i}^{0}$ where $\varphi_{i}$ is the body of the base rule of $X_{i}$ in $\Lambda$. 
        \item Assume we have defined the truth of all $(\Pi, \cT)$-schemata in round $n$. Then we define $M, w \models X_{i}^{n+1}$ iff $M, w \models \psi_{i}^{n}$ where $\psi_{i}$ is the body of the induction rule of $X_{i}$.
    \end{itemize}
\end{itemize}
Note that if $\psi$ is a formula of graded modal logic then for all $n \in \N$ we have that $M,w \models \psi^n$ iff $M, w \models \psi^0$. Now, we say that $(M, w)$ is \textbf{accepted} by $\Lambda$ if $M, w \models X^n$ for some $n \in \N$ and some attention predicate $X$. A Kripke model is \textbf{suitable} for $\Lambda$ if it is a $\Pi$-model.

If $\Pi' \subseteq \Pi$ is the set of proposition symbols that appear in a program, the linear order $<^{\Pi'}$ and the model $M$ induce a binary string $\bi \in \{0, 1\}^{\abs{\Pi'}}$ for each node $w \in M$ that serves as input, i.e. the $j$th bit of $\bi$ is $1$ iff $w$ is in the valuation of the $j$th proposition symbol in $\Pi'$. 

Since formulae of propositional logic and programs of $\SC$ contain no diamonds, their semantics are essentially defined over a Kripke model where $\abs{W} = 1$ and $R = \emptyset$. Such a model corresponds to a valuation $\Pi \to \{0, 1\}$ assigning a truth value to each proposition symbol in $\Pi$. Informally, $\GMSC$ can be seen as a separate $\SC$-program being run in each node of a Kripke model, where the programs are able to communicate with each other by passing messages using the diamond $\Diamond_{\geq k}$.

\begin{example}
    The centre-point property is the node property (i.e. a class of pointed Kripke models) stating that there exists an $n \in \N$ such that each directed path starting from the node leads to a node with no successors in exactly $n$ steps. Consider the program $Y (0) \colonminus \Box \bot$, $Y \colonminus \Diamond Y \land \Box Y$. It is clear that the program accepts precisely the pointed models which have the centre-point property.
\end{example}

\subsection{Boolean network logic}

We then define Boolean network logic (or $\BNL$) which we will later show to be equivalent to the fragment $\SC$.
Boolean network logic gets its name from Boolean networks, which are discrete dynamical systems commonly used in various fields, e.g., biology, telecommunications and various others. For example, they are used to describe genetic regulatory networks (e.g. \cite{kauffman1969homeostasis}). 
Informally, a (classical) Boolean network consists of a finite nonempty set $\cX$ of Boolean variables and each variable $x$ is associated with a Boolean function $f_x$ which takes as input the values of the variables in $\cX$. 
The Boolean values of all variables in $\cX$ are updated in discrete steps starting with a given initial value. 
In each subsequent round, 
the value of $x$ is the output of $f_x$ applied to the values of the variables in the previous round.
It is possible to extend Boolean networks by allowing the initial values of some variables to be fixed.
There is no general syntax for extended Boolean networks, but $\BNL$ will give us a suitable one.

Let $\cT \subseteq \var$. A \textbf{$\cT$-schema of Boolean network logic} (or $\BNL$) is defined over the language $\psi \coloncolonequals \top \mid V_i \mid \neg \psi \mid \psi \land \psi$, where $V_i \in \cT$ (in contrast to $\SC$, we do not include proposition symbols). 
Assume now that $\cT$ is finite and nonempty. 
A \textbf{$\cT$-program of $\BNL$} is defined analogously to a $(\Pi, \cT)$-program of $\SC$, with the following three differences:
\begin{enumerate}
    \item The base rules of $\BNL$ are either of the form $X(0) \colonminus \top$ or $X(0) \colonminus \bot$. 
    \item The bodies of the induction rules of $\BNL$ are $\cT$-schemata of $\BNL$ (see above). 
    \item Each schema variable in a $\BNL$-program has exactly one induction rule and either one or zero base rules. In a $\BNL$-program, we let $\cI$ denote the set of predicates that do not have a base rule, which we call \textbf{input predicates}.
\end{enumerate}
For example, consider a $\BNL$-program with the base rule $X(0) \colonminus \top$ and the induction rules $X \colonminus \neg X$ and $Y \colonminus Y \land X$. Here $Y$ is the sole input predicate and $X$ acts as an auxiliary predicate.
Like programs of $\SC$, each $\BNL$-program also includes print predicates and optionally attention predicates that determine the output rounds of the program. If there are no attention predicates, the output rounds can be given externally with an attention function $a \colon \{0,1\}^k \to \wp(\N)$, where $k = \abs{\cI}$.
Analogously to $\SC$-programs, whenever necessary for the context, we will tell whether a $\BNL$-program is associated with an attention function or not.

The semantics of a program of $\BNL$ is defined over a valuation $M \colon \cI \to \{0,1\}$, which we simply call $\cI$-\textbf{models} (or models or inputs if $\cI$ is clear from the context). 
A model is \textbf{suitable} for a $\BNL$-program $\Lambda$ if it is a valuation over the input predicates of $\Lambda$. 
Assume that $I_{1}, \dots, I_{\abs{\cI}}$ enumerates the set $\cI$ in the order $<^{\var}$.
Analogously to a model of $\SC$, any model $M$ for $\BNL$ 
induces a binary string $\bi_{M} \in \{0,1\}^{\abs{\cI}}$ that serves as input where $\bi_{M}(j) = M(I_{j})$. Note that vice versa each string $\bi \in \{0,1\}^{\abs{\cI}}$ induces a model defined by 
$M_\bi \colon \cI \to \{0,1\}$ such that $I_{j} \mapsto \bi(j)$.
The truth of a $\cT$-schema is defined analogously to $\SC$ except for the truth value of head predicates in round $0$. If $X \in \cI$, we define that $M \models X^0$ iff 
$M(X) = 1$.
If $X \notin \cI$, then $M \models X^0$ iff the body of the base rule of $X$ is $\top$. 
Later in Example \ref{example: BNL with and without auxiliary}, after we have defined how a $\BNL$-program produces a run sequence, we show how powerful auxiliary predicates $X \notin \cI$ are.

\subsection{Definitions for time series of BNL and SC}\label{subsection: notions_SC_BNL}

Next, we introduce useful concepts relating to Boolean network logic $\BNL$ and substitution calculus $\SC$ that later will be used to establish notions of equivalence among programs, circuits and neural networks.

Let $\psi$ be a schema of $\SC$ or $\BNL$.
The \textbf{depth} $d(\psi)$ of $\psi$ is the maximum number of nested operators $\neg$ and $\land$ in $\psi$.
The \textbf{subschemata} of $\psi$ are defined recursively in the obvious way.
Let $\Lambda$ be a program of $\SC$ or $\BNL$.
The \textbf{size} (or the required \textbf{space}) of $\Lambda$ is defined as the number of appearances of $\top$,
proposition symbols $p_{i}$, head predicates $V_{i}$ and operators $\neg$ and $\land$ in its base and induction rules. 
The \textbf{depth} of $\Lambda$ is the maximum depth of the bodies of the induction rules and base rules of~$\Lambda$.

Let $X_1, \ldots, X_n$ enumerate a set $\cT$ of schema variables (in the order $<^{\var}$). Let $\Lambda$ be a $(\Pi, \cT)$-program of $\SC$ or a $\cT$-program of $\BNL$, and $M$ a suitable model for $\Lambda$ that induces an input $\bi \in \{0, 1\}^{k}$.
Each round $t \in \N$ defines a \textbf{global configuration} (or \textbf{state}) $g_t \colon \cT \to \{0,1\}$, where $g_t(X_i) = 1$ iff $M \models X_i^t$, for each $X_i \in \cT$. 
Thus, $\Lambda$ with input $\bi$ also induces an infinite sequence $B_\Lambda = (\bb_t)_{t \in \N}$ called the \textbf{global configuration sequence}, where $\bb_t = g_{t}(X_1) \cdots g_{t}(X_n)$.

Now we shall define the run sequence of the global configuration sequence $B_\Lambda$ of $\Lambda$ with input $\bi$. 
If $\cP$ and $\cA$ are the sets of print and attention predicates of $\Lambda$ respectively, then the corresponding sets of print and attention positions are $P_\Lambda \colonequals \{\, i \mid X_i \in \cP\,\}$ and $A_\Lambda \colonequals \{\,i \mid X_{i} \in \cA\,\}$.
The corresponding set of attention strings is $\cS_\Lambda \colonequals \{0,1\}^{\abs{\cA}} \setminus \{0\}^{\abs{\cA}}$.
Thus, with input $\bi$, $\Lambda$ induces the \textbf{run sequence} of $B_\Lambda$ w.r.t. $(A_\Lambda, P_\Lambda, \cS_\Lambda)$. 
We may also simply say that with input $\bi$, $\Lambda$ induces the run sequence w.r.t. $(\cA, \cP)$.
If an attention function $a$ is used instead of attention predicates, then the output rounds are given by $a(\bi)$.
Similarly to the general output conditions defined in Section \ref{subsec: run sequences}, we say that $\Lambda$ with input $\bi$ induces \textbf{output rounds} and an \textbf{output sequence} w.r.t. $(\cA, \cP)$ (or resp. w.r.t. $(a(\bi), \cP)$).

$\BNL$-programs inherit a number of properties from Boolean networks. 
Each $\BNL$-program will eventually reach a single stable state (called \textbf{point attractor}, \textbf{fixed point attractor} or simply \textbf{fixed point}) or begin looping through a sequence of multiple states (called \textbf{cycle attractor}) since each program can only define a finite number of states. 
The smallest amount of time it takes to reach an attractor from a given state is called the \textbf{transient time} of that state. 
The \textbf{transient time} of a $\BNL$-program is the maximum transient time of a state in its state space \cite{cheng2010linear}. 
Analogously, these concepts are applicable to $\SC$.

Consider the fragment $\mathrm{BNL}_0$ of $\BNL$ where no head predicate of a program is allowed to have a base rule. The programs of this logic $\mathrm{BNL}_0$ are an exact match with \emph{classical Boolean networks}; each program encodes a Boolean network, and vice versa. 
\begin{example}\label{example: BNL with and without auxiliary}
    We demonstrate the power of auxiliary predicates in $\BNL$-programs. 
    The logic $\mathrm{BNL}$ extends $\BNL_0$ by allowing base rules.  
    Consider the following sequences:
    \[
    \begin{aligned}
        &1.\,\, 01, 10, 01, 10, 01, 10, 01, 10, \ldots\\
        &2.\,\, 01, 01, 10, 10, 01, 01, 10, 10, \ldots
    \end{aligned}
    \]
    Suppose that we want to build programs which induce the given output sequences with input $01$. The first can be induced with the following $\BNL_0$-program: 
    $X \colonminus Y$, $Y \colonminus X$, where $X <^{\mathrm{VAR}} Y$ and $\{X,Y\}$ is the set of attention and print predicates.
    However, the latter cannot be induced by any $\BNL_0$-program with attention predicates, because programs of $\BNL_0$ lack auxiliary predicates and the size of the input restricts the $\BNL_0$-program to two head predicates which are both attention and print predicates.
    This means that the program must have an attractor starting from the state $01$ (resp. $10$) where the state $01$ is possible to reach both with and without visiting $10$, which is impossible.
    However, the following program with auxiliary predicates induces the latter sequence: $A (0) \colonminus \bot$, $A \colonminus \neg A$, $X \colonminus (A \land Y) \lor (\neg A \land X)$ and $Y \colonminus (A \land X) \lor (\neg A \land Y)$, where again $\{X, Y\}$ is the set of attention and print predicates.
\end{example}

Finally, we define a special class of $\BNL$-programs.
A $\BNL$-program (with attention predicates) that only has fixed points (i.e. no input leads to a cycle attractor) and outputs precisely at fixed points, is called a \textbf{halting $\BNL$-program}. 
For a halting $\BNL$-program $\Lambda$ with input predicates $\cI$ and print predicates $\cP$, each input $\bi \in \{0,1\}^{\abs{\cI}}$ results in a single (repeating) \textbf{output} denoted by $\Lambda(\bi)$, which is the output string determined by the fixed-point values of the print predicates. 
In this sense, a halting $\BNL$-program is like a function $\Lambda \colon \{0, 1\}^{\abs{\cI}} \to \{0, 1\}^{\abs{\cP}}$. We say that $\Lambda$ \textbf{specifies} a function $f \colon \{0, 1\}^{\ell} \to \{0, 1\}^{k}$ if $\abs{\cI} = \ell$, $\abs{\cP} = k$ and $\Lambda(\bi) = f(\bi)$ for all $\bi \in \{0, 1\}^{\ell}$. 
The \textbf{computation time} of a halting $\BNL$-program is its transient time.
In the case where a $\BNL$-program $\Lambda$ with only fixed points has no attention predicates, we say that $\Lambda$ is a \textbf{halting $\BNL$-program w.r.t. $a$}, where $a$ is the attention function that gives for each input the corresponding rounds where the fixed point has been reached.

Naturally the notions of equivalence for run sequences from Section \ref{subsection: run sequence equivalence} generalize for programs. 
Let $\sigma$ be a similarity over $\{0,1\}$.
Let $\Lambda$ and $\Lambda'$ each be a program of $\SC$ or a program of $\BNL$ such that either they both have attention predicates or they both have an attention function.
We say that $\Lambda'$ is \textbf{equivalent{}} to $\Lambda$ w.r.t. $\sigma$ if given any input $\bb$ of $\Lambda$, 
the run sequence of $\Lambda'$ with input $\sigma(\bb)$ is 
equivalent{} (w.r.t. $\sigma$) to the run sequence of $\Lambda$ with input $\bb$.
The same generalization naturally applies for the concepts of \textbf{strong equivalence} w.r.t. $\sigma$, \textbf{time delay} and \textbf{precomputation}. Moreover, we require that the time delay and the number of precomputation rounds between the programs are the same regardless of input.
If $\sigma$ is the canonical similarity over $\{0,1\}$ then we may omit $\sigma$ (and this is often the case).
Note that if $\Lambda$ is a program with an attention function $a$ and
we want to construct a program $\Lambda'$ that is equivalent{} to $\Lambda$ w.r.t. a similarity $\sigma$ with computation delay $T$ and $t$ precomputation rounds, then it is necessary, sufficient and always possible to define the attention function $a'$ of $\Lambda'$ such that for all inputs $\bi$ of $\Lambda$ 
the set $a'(\sigma(\bi))$ is constructed from the set $a(\bi)$ by Remark \ref{remark: outputs and attention}.
Thus, we do not need to separately construct the attention functions when constructing equivalent{} programs.

Note that any given halting $\BNL$-program $\Lambda$ (with attention predicates) can be translated into a strongly equivalent halting $\BNL$-program $\Lambda'$ w.r.t. an attention function $a$, where $\Lambda'$ and $a$ are defined as follows. The program $\Lambda'$ is obtained from $\Lambda$ by treating the attention predicates as normal non-attention predicates and $a$ is defined by $a(\bi) = \{\, n \in \N \mid n \geq t_{\bi} \,\}$, where $t_\bi$ is the least round where $\Lambda$ reaches a fixed point with the input $\bi$.
Henceforth, when constructing halting $\BNL$-programs, we will consider the case with attention predicates, as the case with attention functions follows trivially by the above.

\section{Translations between BNL, SC and self-feeding circuits}

In this section, we present two alternative ways of representing $\BNL$ as well as auxiliary tools and results which we will use later. We start by showing how two fundamental programming techniques---if-else statements and while loops---can be implemented in $\BNL$. 
Then we show equivalence between $\BNL$ and $\SC$ in Theorem \ref{SC_BNL}.
After that, we show in Theorem \ref{BNL_to_FBNL} how each $\BNL$-program can be translated into an equivalent $\BNL$-program with very simple induction rules.
Finally, we introduce a self-feeding circuit model that is also equivalent to $\BNL$, shown in Theorems \ref{thrm:BNL_to_SF_circuit} and \ref{thrm:SF_circuit_to_BNL}.
By using these connections to $\BNL$, we later obtain additional characterizations for neural networks via $\SC$ and self-feeding circuits.

\subsection{If-else statements and while loops}\label{section: tools}

We start by introducing some standard programming tools that are most relevant in our proofs when constructing $\BNL$-programs. We especially consider ``if-else statements'' and ``while loops''. Later, these tools are used to combine subprograms that together simulate arithmetic operations and neural networks in $\BNL$; we simply construct the relevant subprograms and refer to this section for how those subprograms may be combined. 

\textbf{If-else statements}

Informally, an if-else statement in a computer program executes an action if a condition is met and instead executes another action if the condition fails.
We demonstrate the implementation of this in $\BNL$.

We start with a warm-up where we define a conditional rule for a single head predicate. Let $X$ be a schema variable and let $\chi$, $\psi$ and $\varphi$ be $\BNL$-schemata. Informally, we define an induction rule for $X$ which executes $\psi$ if $\chi$ is true (i.e. the condition is met) and else~$\varphi$ (i.e. the condition fails). 
Formally, the \textbf{conditional rule for $X$ w.r.t. $(\chi, \psi, \varphi)$} is $X \colonminus (\psi \land \chi) \lor (\varphi \land \neg \chi)$. 
While this technique can be used for singular head predicates, it can also be used for whole subprograms. Moreover, two separate programs can be combined as follows.
Given two separate programs $\Lambda$ and $\Lambda'$, we can replace the bodies of their induction rules with conditional rules to create a combined program that carries out $\Lambda$ if some condition is met and $\Lambda'$ if not. The chosen program can then be run for as long as desired, even to the point of stopping for halting programs, as we will see with the next technique.

\textbf{While loops}

Informally, a while loop means that an action of a computer program is executed repeatedly while a certain condition is met and stopped once the condition fails. We demonstrate the implementation of this in $\BNL$.

It is easy to define a while loop in $\BNL$ by using if-else statements. 
Say we want to execute a program $\Gamma$ while a $\BNL$-schema $\chi$ is true, and else execute alternate rule bodies $\varphi$; then we simply replace each rule $X \colonminus \psi$ in $\Gamma$ with the conditional rule for $X$ w.r.t. $(\chi, \psi, \varphi)$. Now, each $X$ applies its original body $\psi$ (i.e. executes) while $\chi$ is true and else 
applies~$\varphi$.
Often $\varphi = X$, meaning that $X$ does not update its truth value (i.e. stops) when $\chi$ is false. 
It is only a small step to apply this for multiple rounds in a row, as we shall see next.

Often the condition for a while loop is given by a \textbf{clock} or \textbf{one-hot counter}, i.e., a program of type
\[
\begin{aligned}
    &T_0 (0) \colonminus \top\qquad &&T_0 \colonminus T_\ell \\
    &T_1 (0) \colonminus \bot &&T_1 \colonminus T_0 \\
    &\vdots &&\vdots \\
    &T_\ell (0) \colonminus \bot &&T_\ell \colonminus T_{\ell-1}.
\end{aligned}
\]
The predicates are updated in a periodic fashion as follows.
In round $0$ the only true head predicate is $T_0$. Then in round $1$ the only true predicate is $T_1$, and so on. Moreover, in round 
$\ell+1$
the only true predicate is $T_0$ again, and the process continues in a loop. 
Now we can, for example, define a rule $C \colonminus \bigvee_{i = 0}^{k} T_{i}$ and construct conditional rules w.r.t. $(C, \psi, \varphi)$ such that the program executes a subprogram defined by formulae $\psi$ for $k+1$ rounds and then another subprogram with the rules $\varphi$ for the next $\ell-k$ rounds.

\subsection{SC and BNL are equivalent}

In the following theorem, we show that $\BNL$-programs translate to $\SC$-programs and vice versa. Both translations are linear.

\begin{theorem}\label{SC_BNL}
    For each $\SC$-program of size $m$, we can construct an equivalent{} $\BNL$-program of size $\ordo(m)$ with $1$ precomputation round. Moreover, for each $\BNL$-program of size $n$, we can construct a strongly equivalent $\SC$-program of size $\ordo(n)$.
\end{theorem}

\begin{proof}
We start with the direction from $\SC$ to $\BNL$.
We prove the case where the $\SC$-program includes attention predicates since the case where output rounds are given by an attention function follows from Remark \ref{remark: outputs and attention}.
Given an $\SC$-program $\Lambda$, we construct an equivalent{} $\BNL$-program $\Lambda'$ as follows. Informally, $\Lambda'$ uses one round to compute the base rules of $\Lambda$; the base rules of $\Lambda$ are embedded into the induction rules of $\Lambda'$ using conditional rules.

More formally, we first copy each head predicate of $\Lambda$ to $\Lambda'$ and their status as attention/print predicates, and define their rules as follows. We add a fresh auxiliary predicate $T$ with the rules $T (0) \colonminus \bot$, $T \colonminus \top$. Then from a base rule $X (0) \colonminus \varphi$ and an induction rule $X \colonminus \psi$ of $\Lambda$, we form a new induction rule $X \colonminus (T \land \psi) \lor ( \neg T \land \varphi)$ for $\Lambda'$, i.e., we create a conditional rule for $X$ w.r.t. $(T, \psi, \varphi)$.
Intuitively, $T$ makes it so that the induction rule of $X$ uses the body of the original base rule in round $1$ and the body of the original induction rule in every subsequent round. For each head predicate originally from $\Lambda$ we set $X(0) \colonminus \bot$; this prevents attention predicates from triggering in round $0$. 
Finally, for each proposition symbol $p_i$ that appears in $\Lambda$, we introduce an input predicate $P_{i}$ with the induction rule $P_i \colonminus P_i$, and we assume that variables $P_i$ are ordered w.r.t. the order of proposition symbols $p_i$. Then we replace each instance of a proposition symbol $p_i$ that appears in the bodies of base and induction rules originally from $\Lambda$ with $P_i$. 

Now, the program $\Lambda'$ is ready, and it is easy to show that $\Lambda$ and $\Lambda'$ are equivalent{}. The size of $\Lambda'$ is clearly linear in the size of $\Lambda$ and the computation delay is $0$ with $1$ precomputation round.

From $\BNL$ to $\SC$, we amend the $\BNL$-program with the missing base rules by using proposition symbols. That is, if $X_i$ is an input predicate of the $\BNL$-program, then we add a base rule $X_i (0) \colonminus p_i$ for $X_i$ and assume that the proposition symbols are in the same order as the input predicates. Otherwise, the $\SC$-program is identical to the $\BNL$-program. 
\end{proof}

\subsection{Fully-open BNL-programs}

The induction rules in a $\BNL$-program can vary greatly in terms of complexity from elementary rules such as $X \colonminus Y$ to potentially having to explicitly state the accepted configurations of all head predicates. 
In a realistic setting, the time required to calculate such schemata also varies, e.g., depending on the depth of the schema. Thus, we introduce a so-called fully-open form where each induction rule is elementary, i.e., built by applying at most a single negation or conjunction to schema variables or the logical constant $\top$. 
More formally, we say that a $\cT$-program $\Lambda$ of $\BNL$ is \textbf{fully-open}, if each body $\psi$ of an induction rule $X \colonminus \psi$ of $\Lambda$ fulfills one of the following conditions:
1) $\psi \in \cT \cup \{\top\}$ or 2) $\psi = \neg \varphi$ for some $\varphi \in \cT \cup \{\top\}$ or 3) $\psi = \varphi \land \theta$ for some $\varphi, \theta \in \cT \cup \{\top\}$.
We obtain the following auxiliary result, which we will use in Section \ref{BNL to NN}.

\begin{theorem}\label{BNL_to_FBNL}
    Given a $\BNL$-program $\Lambda$ of size $s$ and depth $d$, we can build an equivalent{} fully-open $\BNL$-program $\Lambda'$ 
    of size $\ordo(s)$ 
    where the computation delay of $\Lambda'$ is $\ordo(d)$.
\end{theorem}
\begin{proof}
    The rough idea is to break the induction rules of $\Lambda$ down into components consisting of at most one operator $\neg$ or $\land$  
    by replacing subschemata $\varphi$ with auxiliary schema variables $X_{\varphi}$. For example, the induction rule $X \colonminus \neg Y \land \neg Z$ could be broken down to the induction rules $X \colonminus X_{\neg Y} \land X_{\neg Z}$, $X_{\neg Y} \colonminus \neg Y$ and $X_{\neg Z} \colonminus \neg Z$, which would simulate each iteration of $X$ in $\Lambda$ in two iterations of~$\Lambda'$. 
    However, this is not quite enough as the depth of the subschemata can vary. Thus, we have to synchronize these auxiliary variables to make sure that subschemata are calculated in the correct round in accordance with their depth; we do this by using a one-hot counter to build conditional rules (see Section \ref{section: tools}). 
    Although conditional rules are not in a fully-open form,
    we find that it is possible to break them down into auxiliary variables as well without further complications.

    We construct an equivalent{} fully-open $\BNL$-program $\Lambda'$ (w.r.t. the canonical similarity over $\{0,1\}$) in the following way. We assume that $\Lambda$ uses attention predicates, since the case for attention functions is easy to obtain by Remark \ref{remark: outputs and attention}.
    First, we copy all head predicates from $\Lambda$, as well as their base rules and their status as input/print/attention predicates. We define a one-hot counter $T_{0}, \dots, T_{5(d+1)}$ exactly as we did in Section \ref{section: tools}.
    We also define a complementing counter $T_{0}', \dots, T_{5(d+1)}'$ where in each round we have that $T_{i}'$ is true if and only if $T_{i}$ is false (i.e., $T_{i}' \equiv \neg T_{i}$, meaning that in any given round exactly one variable $T_{i}'$ is false). This is simple to construct by modifying the one-hot counter in Section \ref{section: tools} by changing each base rule from $\top$ to $\bot$ and vice versa.
    
    Intuitively, we construct a module for each subschema of $\Lambda$, where each module contains eight head predicates in five layers that operate in a loop, each loop consisting of five iteration rounds, one per layer. Each of these loops essentially calculates a conditional rule of depth five for the given subschema, one iteration round per depth.
    More formally, for each occurrence of a subschema $\varphi$, we define corresponding head predicates $X_{\varphi}^{1}, \dots, X_{\varphi}^{5}$ and $Y_{\varphi}^{1}, \dots, Y_{\varphi}^{3}$, where the upper index signifies the layer of the loop the predicate belongs to. Before defining the induction rules for these predicates, we define the schema $\theta_{\varphi}$ as follows. Let $\cT$ denote the set of head predicates of $\Lambda$. If $\varphi \in \cT \cup \{\top\}$, then $\theta_{\varphi} \colonequals \varphi$. If $\varphi \colonequals \neg \psi$, then $\theta_{\varphi} \colonequals \neg X_{\psi}^{5}$. If $\varphi \colonequals \psi_{1} \land \psi_{2}$, then $\theta_{\varphi} \colonequals X_{\psi_{1}}^{5} \land X_{\psi_{2}}^{5}$. Assume the depth of $\varphi$ is $d'$. We define the induction rules of $X_{\varphi}^{1}, \dots, X_{\varphi}^{5}$ and $Y_{\varphi}^{1}, \dots, Y_{\varphi}^{3}$ as follows:\footnote{In the second bullet, the variable $T_{5d'+1}$ makes sure that the truth values of the immediate subschemata of $\varphi$ are witnessed at the right moment; note that $T_{5d'+1}$ is true iff $T_{5d'+1}'$ is false}
    \begin{itemize}
        \item $X_{\varphi}^{1} \colonminus \theta_{\varphi}$ and $Y_{\varphi}^{1} \colonminus X_{\varphi}^{5}$,
        \item $X_{\varphi}^{2} \colonminus T_{5d'+1} \land X_{\varphi}^{1}$ and $Y_{\varphi}^{2} \colonminus T_{5d'+1}' \land Y_{\varphi}^{1}$,
        \item $X_{\varphi}^{3} \colonminus \neg X_{\varphi}^{2}$ and $Y_{\varphi}^{3} \colonminus \neg Y_{\varphi}^{2}$,
        \item $X_{\varphi}^{4} \colonminus X_{\varphi}^{3} \land Y_{\varphi}^{3}$,
        \item $X_{\varphi}^{5} \colonminus \neg X_{\varphi}^{4}$.
    \end{itemize}
    
    By substituting the induction rules of $X_{\varphi}^{1}, \dots, X_{\varphi}^{4}$ and $Y_{\varphi}^{1}, \dots, Y_{\varphi}^{3}$ into the induction rule of $X_{\varphi}^{5}$, we obtain the schema on the left below which by De Morgan's law is logically equivalent to the conditional rule on the right:
    \[
        \neg (\neg (T_{5d'+1} \land \theta_{\varphi}) \land \neg (T_{5d'+1}' \land X_{\varphi}^{5})) \equiv (T_{5d'+1} \land \theta_{\varphi}) \lor (\neg T_{5d'+1} \land X_{\varphi}^{5}).    
    \]
    This illustrates how the head predicates together calculate a conditional rule where $T_{5d'+1}$ is the condition.
    Finally, for the head predicates that we copied from $\Lambda$, we define the induction rule $X \colonminus T_{5(d+1)} \land X_{\psi}^{5}$, where $\psi$ is the induction rule of $X$ in $\Lambda$.
    Clearly, the size of $\Lambda'$ is linear in the size of $\Lambda$, since we only add a constant number of constant-size induction rules for each subschema of $\Lambda$, and the one-hot counters are linear in the size of $\Lambda$. 

    Let $M$ be a suitable model for $\Lambda$ and $\Lambda'$. Note that the programs have the same input predicates $\cI$. It is easy to show by induction on $n \in \N$ that for each $\cI$-model and for each head predicate $X$ that appears in $\Lambda'$ and $\Lambda$ that $M \models X^n$ w.r.t. $\Lambda$ iff $M \models X^{5(d+1)n}$ w.r.t. $\Lambda'$.
    Thus, $\Lambda'$ is equivalent{} to $\Lambda$ and the computation delay of $\Lambda'$ is $\ordo(d)$.
\end{proof}

\subsection{Link to self-feeding circuits}\label{section:self-feeding_circuits}

In this section, we introduce self-feeding circuits which are more or less just a syntactic variation of $\BNL$-programs and related to sequential circuits \cite{CavanaghJoseph2006SLAa}. We will show that for every $\BNL$-program, we can construct an equivalent self-feeding circuit and vice versa. We also pay special attention to the size and time complexities in the translations.

\subsubsection{Circuits}

We first recall some basics related to Boolean circuits. 
A \textbf{Boolean circuit} is a directed, acyclic graph where each node is either labeled with one of the symbols $\land , \lor , \neg$ or unlabeled; only nodes with in-degree $1$ can be labeled $\neg$, and only nodes with in-degree $0$ can be unlabeled.
The nodes of a circuit are called \textbf{gates}.
The in-degree of a gate $u$ is called the \textbf{fan-in} of $u$, and the out-degree of $u$ is the \textbf{fan-out}. The \textbf{input gates} of a circuit are precisely the gates that have zero fan-in and no label $\land , \lor$ or $\neg$. The \textbf{output gates} are the ones with fan-out zero; we allow multiple output gates in a circuit.

The \textbf{size} $\abs{C} $ of a circuit $C$ is the number of gates in $C$.
The \textbf{height} $\mathsf{height}(G)$ of a gate $G$ in $C$ is the length of the longest path from a gate with fan-in zero to the gate $G$. 
The \textbf{depth} $\mathsf{depth}(C)$ (or the \textbf{computation/evaluation time}) of $C$ is the maximum height of a gate in $C$. 
The notions of depth and height can also be used analogously for any directed acyclic graphs.
The sets of input and output gates of a
circuit are both linearly ordered. A circuit with $n$ input gates and $k$ output gates then \textbf{computes} (or specifies) a 
function of type $\{0,1\}^n \rightarrow \{0,1\}^k$. This is done in the natural way, analogously to the Boolean operators corresponding to $\wedge,\vee,\neg$; see, for example, \cite{libkin} for the formal definition. The output of the circuit is the \emph{binary string} determined by the bits of the linearly ordered output gates. Note that 
gates with the labels $\wedge,\vee$ can have any fan-in (also $0$), meaning that by default, circuits have \textbf{unbounded fan-in}. In the elaborations below, we say that a circuit has \textbf{bounded fan-in} if
the fan-in of every $\land$-gate and $\lor$-gate of the circuit is at most $2$. The $\land$-gates that have fan-in zero always output $1$, and therefore correspond to the symbol $\top$. Respectively, the $\lor$-gates that have fan-in zero correspond to the symbol $\bot$.

A \textbf{Boolean formula} $\varphi$ with $n$ variables is a circuit with bounded fan-in specifying a function of type $\{0,1\}^n \to \{0,1\}$ and where the fan-out of every non-input gate is at most~$1$.
Let $\cC$ denote the set of all circuits. Given $x$ and $y$ in the set 
$\cC$,
we say that $x$ and $y$ are \textbf{equivalent} if they specify the same function.

\subsubsection{Self-feeding circuits}

Given an integer $k \in \Z_+$, a \textbf{self-feeding circuit for $k$} is a circuit $C$ that specifies a function 
$
f \colon \{0,1\}^{k} \rightarrow \{0,1\}^{k}.
$
The circuit $C$ is associated with a set of \textbf{input positions} $I \subseteq [k]$ and an \textbf{initializing function} $\pi \colon [k] \setminus I \to \{0,1\}$. The elements of $[k]\setminus I$ are called \textbf{auxiliary positions}. 
Moreover, $C$ is also associated with a set $P\subseteq [k]$ of \textbf{print positions} and optionally a set $A \subseteq [k]$ of \textbf{attention positions} that determine the output rounds of the circuit. If there are no attention positions, the output rounds can be given by an external \textbf{attention function} $a \colon \{0,1\}^{\abs{I}} \to \wp(\N)$. Analogously to programs, whenever necessary for the context, we will tell whether a self-feeding circuit is associated with an attention function or not.

Informally, a self-feeding circuit computes as follows. In round $0$, the input gates in input positions are fed with the input and the input gates in auxiliary positions are fed with the values given by the initializing function; then the circuit produces an output in the ordinary way.
After that in each round $n$ the output from the previous round $n-1$ is fed back to the circuit itself to produce a new output. We then define the computation of self-feeding circuits formally. Let $C$ be a self-feeding circuit for $k$ with the set $I$ of input positions and input $i \colon I \to \{0,1\}$ (or the corresponding bit string $\bi \in \{0,1\}^{\abs{I}}$). Respectively, the function $\pi \cup i$ corresponds to the binary string $\bs_{\pi \cup i} \in \{0,1\}^{k}$, where for each $\ell \in [k]$, $\bs_{\pi \cup i}(\ell) = i(\ell)$ if $\ell \in I$ and $\bs_{\pi \cup i}(\ell) = \pi(\ell)$ if $\ell \notin I$. Each round $n \in \N$ defines a \textbf{global configuration} $\bg_n \in \{0,1\}^{k}$. The configuration of round $0$ is the $k$-bit binary string $\bg_0=\bs_{\pi \cup i}$. 
Recursively, assume we have defined $\bg_{n}$. 
Then $\bg_{n+ 1}$ is the 
output string of $C$ when it is fed with the string $\bg_n$.
The set of attention strings of $C$ is $\cS = \{0,1\}^{\abs{A}} \setminus \{0\}^{\abs{A}}$.
Now, consider the sequence $B_{C} = (\bg_n)_{n\in\mathbb{N}}$ of $k$-bit strings that $C$ produces. The circuit $C$ with input $i$ (or $\bi$) also induces the run sequence of $B_{C}$ w.r.t. $(A, P, \cS)$ (or the run sequence w.r.t. $(A, P)$).
It thus also induces a set of \textbf{output rounds} and an \textbf{output sequence} w.r.t. $(A,P)$ (or if $C$ does not contain attention positions, then together with an attention function $a$ it induces a set of output rounds and an output sequence w.r.t. $(a(\bi), P)$). 

Analogously to $\SC$ and $\BNL$, we define \textbf{equivalence{}}, \textbf{strong equivalence}, \textbf{computation delay} and \textbf{precomputation rounds} between two self-feeding circuits or between a self-feeding circuit and a program. 
Also, analogously to $\BNL$, we define \textbf{halting self-feeding circuits} and related notions, i.e., the \textbf{computation time} and the \textbf{output} of a halting self-feeding circuit, and how a halting self-feeding circuit specifies a function.

Self-feeding circuits differ from Boolean circuits; they can express some Boolean functions more compactly (in terms of size or depth). 
A \textbf{family} of Boolean functions is a sequence $(f_n)_{n \in \N}$, where $f_n \colon \{0,1\}^n \to \{0,1\}^m$ for some $m \in \N$. We say that a sequence $(C_n)_{n \in \N}$ of self-feeding circuits (or Boolean circuits) \textbf{computes} a family $(f_n)_{n \in \N}$ of Boolean functions if $C_n$ specifies $f_n$ for every $n \in \N$.
Given $n \in \N$, a \textbf{parity function} $\mathrm{PAR}_n \colon \{0,1\}^n \to \{0,1\}$ is defined by $\mathrm{PAR}_n (\bs) = 1$ if and only if $\bs$ has an odd number of $1$s. It is a well-known fact that $\mathrm{PARITY} = (\mathrm{PAR}_n)_{n \in \N}$ is not in the circuit complexity class $\mathsf{AC}^0$, i.e. there does not exist a sequence $(C_n)_{n\in \N}$ of Boolean circuits such that $C_n$ computes $\mathrm{PAR}_n$ and the depth of $C_n$ is $\ordo(1)$ and the size of $C_n$ is polynomial in $n$. For self-feeding circuits this is not true, as we will show in the following theorem.

\begin{theorem}
    $\mathrm{PARITY}$ is computable by a sequence $(C_n)_{n \in \N}$ of halting self-feeding circuits such that the depth of $C_n$ is constant, the size of $C_n$ is linear in $n$ and $C_n$ computes $\mathrm{PAR}_n$ 
    with computation time $\ordo(\log n)$.
    Moreover, each $C_n$ has only one auxiliary position.
\end{theorem}

\begin{proof}
    We start with an informal description. For a bit string of length $n_{0} \in \Z_{+}$, each iteration of the circuit pairs up the bits of the string and calculates the parity of each pair, resulting in a string of length $n_{k+1} \colonequals \lceil \frac{n_{k}}{2} \rceil$, which is filled back up to length $n_{0}$ by adding zeros to the right in order to feed it back to the circuit. 
    After at most $\lceil \log(n_{0}) \rceil$ rounds the remaining string has length $1$, meaning that the parity of the whole input string has been calculated. One extra bit is used in an attention position to determine when to output.

    More formally, the circuit $C_{n}$ is defined as follows. The input gates are $x_{1}, \dots, x_{n}, a$, where $a$ is in an auxiliary attention position, meaning that the input string has an extra zero added to it. The output gates of $C_{n}$ are $o_{1}, \dots, o_{n}, o$, that operate as follows.
    \begin{itemize}
        \item For $1 \leq i \leq \lceil \frac{n}{2} \rceil$, the gate $o_{i}$ is the output gate of a subcircuit calculating the ``exclusive or'': $\mathrm{XOR}(x_{2i-1}, x_{2i}) \colonequals (x_{2i-1} \land \neg x_{2i}) \lor (\neg x_{2i-1} \land x_{2i})$, i.e., it outputs $1$ if exactly one of the input bits is $1$. If $n$ is odd, then $o_{\lceil \frac{n}{2} \rceil}$ is instead an $\lor$-gate with fan-in $1$ for $x_{n}$.
        \item For $\lceil \frac{n}{2} \rceil < i \leq n$, the gate $o_{i}$ is an $\lor$-gate with fan-in zero, i.e., it always outputs $0$.
        \item The gate $o$ is the output gate of a subcircuit checking when $\mathrm{PAR}_n$ has been computed, which happens when $x_2, \dots, x_n$ are all $0$s. In other words, $o$ computes $\bigwedge_{2 \leq i \leq n} \neg x_{i}$. 
    \end{itemize}

    The input and output gates are ordered in the natural way, i.e., $x_{1} < \dots < x_{n} < a$ for input gates and $o_{1} < \dots < o_{n} < o$ for output gates. The only print position is $1$, the only attention position is $n+1$, and the set of input positions is $[n]$. The initializing function is $\pi \colon \{n+1\} \to \{0,1\}, \pi(n+1) = 0$.

    It is easy to verify that $C_{n}$ constructed this way computes $\mathrm{PAR}_{n}$ 
    with computation time $\ordo(\log(n))$.
    Below are two tables showing how $C_{5}$ handles the inputs $01100$ and $01011$. The underlined elements of each row are essentially the length of the string in that round, not taking into account the greyed-out elements which are zeros regardless of the input. 
    \begin{table}[!htb]
        \begin{minipage}{0.5\linewidth}
        \caption{Simulation with input $01100$.}    
        \centering
        \begin{tabular}{c|cccccc}
            Round & $x_1$ & $x_2$ & $x_3$ & $x_4$ & $x_5$ & $a$\\
            \hline
            0. & \underline{0} & \underline{1} & \underline{1} & \underline{0} & \underline{0} & 0\\
            \hline
            1. & \underline{1} & \underline{1} & \underline{0} & \textcolor{gray}{0} & \textcolor{gray}{0} & 0\\
            \hline
            2. & \underline{0} & \underline{0} & \textcolor{gray}{0} & \textcolor{gray}{0} & \textcolor{gray}{0} & 0\\
            \hline
            3. & \underline{0} & \textcolor{gray}{0} & \textcolor{gray}{0} & \textcolor{gray}{0} & \textcolor{gray}{0} & 1\\
            \hline
            4. & \underline{0} & \textcolor{gray}{0} & \textcolor{gray}{0} & \textcolor{gray}{0} & \textcolor{gray}{0} & 1\\
            \hline
            \vdots & \vdots & \vdots & \vdots & \vdots & \vdots & \vdots\\
        \end{tabular}
        \end{minipage}
        \begin{minipage}{0.5\linewidth}
        \centering
        \caption{Simulation with input $01011$.}
        \begin{tabular}{c|cccccc}
            Round & $x_1$ & $x_2$ & $x_3$ & $x_4$ & $x_5$ & $a$\\
            \hline
            0. & \underline{0} & \underline{1} & \underline{0} & \underline{1} & \underline{1} & 0\\
            \hline
            1. & \underline{1} & \underline{1} & \underline{1} & \textcolor{gray}{0} & \textcolor{gray}{0} & 0\\
            \hline
            2. & \underline{0} & \underline{1} & \textcolor{gray}{0} & \textcolor{gray}{0} & \textcolor{gray}{0} & 0\\
            \hline
            3. & \underline{1} & \textcolor{gray}{0} & \textcolor{gray}{0} & \textcolor{gray}{0} & \textcolor{gray}{0} & 0\\
            \hline
            4. & \underline{1} & \textcolor{gray}{0} & \textcolor{gray}{0} & \textcolor{gray}{0} & \textcolor{gray}{0} & 1\\
            \hline
            \vdots & \vdots & \vdots & \vdots & \vdots & \vdots & \vdots\\
        \end{tabular}
        \end{minipage}
        \label{label}
    \end{table}
\end{proof}

We recall a well-known fact for Boolean formulae (that also applies to $\BNL$-schemata).
\begin{lemma}[\cite{size-depth-tradeoff}]\label{lem:BL_to_circuit}
    Given a Boolean formula of size $n$, there exists an equivalent Boolean formula of size $\ordo(n^2)$ and depth $\ordo(\log n)$.
\end{lemma}

The next two theorems give our translation from $\BNL$ to self-feeding circuits and back. Theorem~\ref{thrm:BNL_to_SF_circuit} gives two alternative translations with a trade-off between depth and size. 
\begin{theorem}\label{thrm:BNL_to_SF_circuit}
    Given a $\BNL$-program $\Lambda$ of size $n$ and depth $d$, we can construct a strongly equivalent self-feeding circuit with bounded fan-in either of size $\ordo(n)$ and depth $d$ or 
    of size~$\ordo(n^2)$ and depth $\ordo(\log n)$.
\end{theorem}

\begin{proof}    
We prove both claims at once. We consider the case without attention functions, since the case with attention functions is easy to obtain by Remark \ref{remark: outputs and attention}. Let $X_1 \colonminus \psi_1, \ldots, X_m \colonminus \psi_m$ enumerate the induction rules of $\Lambda$ w.r.t. $<^{\var}$ and let $\cI$ denote its set of input predicates. 
We construct a strongly equivalent self-feeding circuit $C_{\Lambda}$ as follows.
First, for each $\psi_i$, we obtain a corresponding circuit $C_i$ with bounded fan-in as follows. 
For the first claim, it is straightforward to represent each $\psi_i$ as a corresponding circuit $C_i$ which contains a node $v_{\psi}$ for each instance of a subschema $\psi$ of $\psi_i$ (including $\psi_i$ itself) and for all subschemata $\psi$ and $\chi$ an edge from $v_{\chi}$ to $v_{\psi}$ iff $\chi$ is an immediate subschema of $\psi$ (e.g. $\chi$ is a conjunct of $\psi$, etc.).
Nodes are labeled with the main connective of the corresponding schema (i.e. the label of $v_{\neg \chi}$ is $\neg$, etc.).
Note that $\top$ corresponds to $\land$ with fan-in~$0$. 
Nodes corresponding to variables $X_j$ are not labeled; they are input nodes and ordered according to the indices.
Note that each circuit $C_i$ is a Boolean formula.
To prove the second claim, each $C_i$ is obtained by applying Lemma \ref{lem:BL_to_circuit} to the corresponding Boolean formula obtained for the first claim. 
Then, we combine the circuits $C_i$ into a single circuit $C_{\Lambda}$ such that they share common input gates and order the output gates based on the indices $i$.
In both cases, the initializing function $\pi$ is defined as follows. If $X_i \notin \cI$, then $\pi(i) = 1$ iff the body of the base rule of $X_i$ is $\top$.

The depth of the obtained circuit $C_{\Lambda}$ is $\ordo(\log n)$ if each circuit $C_i$ is obtained by applying Lemma \ref{lem:BL_to_circuit} and otherwise the depth is $d$, since combining circuits does not affect the depth. 
For each head predicate $X_i$, let $n_i$ denote the size of the body of its induction rule.
The size of $C_{\Lambda}$ is $\ordo(n^2)$ if Lemma \ref{lem:BL_to_circuit} was applied to circuits $C_i$, since the size of each $C_i$ is $\ordo(n_i^2)$ and thus the size of $C_\Lambda$ is $\ordo(n^2)$ (since $n_1^2 + \cdots + n_m^2 \leq n^2$).
Otherwise the size of $C_{\Lambda}$ is $\ordo(n)$ since each $C_i$ was linear in the size of the corresponding body $\psi_i$. The corresponding input positions, print positions and attention positions (or attention function) are straightforward to define. Clearly $C_{\Lambda}$ is strongly equivalent to the program $\Lambda$. 
\end{proof}

Below, the \textbf{fan-in of a circuit} refers to the maximum fan-in of its gates.
\begin{theorem}\label{thrm:SF_circuit_to_BNL}
    Given a self-feeding circuit $C$ with size $n$, depth $d$, fan-in $k$ and $m$ edges, we can construct an equivalent{} $\BNL$-program of size $\ordo(n+m)$ and depth $\ordo(k)$ where the computation delay of the program is $\ordo(d)$. Moreover, if $C$ has bounded fan-in, then the size of the program is $\ordo(n)$ and the depth of the program is $\ordo(1)$.
\end{theorem}
\begin{proof}
    We prove the case without attention functions, since the case with attention functions is easy to obtain by Remark \ref{remark: outputs and attention}.
    The proof is heavily based on the proofs of Lemma $4.1$ and Theorem $4.2$ in \cite{ahvonen_journal}, but we give a proof here. We assume that $d > 0$, since the case for $d = 0$ is trivial. First, we modify $C$ so that we obtain a strongly equivalent circuit $C'$ of size $\ordo(n)$ such that each output gate has the same height $\ordo(d)$, see, e.g., Lemma $4.1$ in \cite{ahvonen_journal}. We define a one-hot counter $T_0, \ldots, T_{\mathsf{depth}(C')}$ as defined in Section \ref{section: tools}. We define a head predicate $X_G$ for each gate $G$ in $C'$ as follows. 
    If $G$ is labeled by $\land$ at height $h$, its fan-in is non-zero and $Y_1, \ldots, Y_k$ are the corresponding head predicates of gates that connect to $G$, then 
    $X_G \colonminus (T_{h} \land Y_1 \land \cdots \land Y_k) \lor (\neg T_{h} \land \psi_G)$, where $\psi_G$ is $X_G$ if $G$ is not an output gate and otherwise $\psi_G$ is $\bot$. On the other hand, if the fan-in of $G$ is zero, then $X_G \colonminus (T_{0} \land \top) \lor (\neg T_{0} \land \psi_G)$.
    The cases for $\lor$-gates and $\neg$\,-gates are analogous. Intuitively, $T_h$ acts as a condition to ensure that each $X_G$ is evaluated in the correct time step.
    The input, print and attention predicates are the predicates corresponding to the output gates in input, print and attention positions, respectively. Let $I_i$ and $O_i$ be the input gate and the output gate in the $i$th positions, respectively.
    We define $X_{I_i} \colonminus (T_{0} \land X_{O_i}) \lor (\neg T_{0} \land X_{I_i})$. 
    If $i$ is in an auxiliary position and $\pi$ is the initializing function of $C$, then we define $X_{O_i}(0) \colonminus \top$ if $\pi(i) = 1$ and $X_{O_i}(0) \colonminus \bot$ if $\pi(i) = 0$. For all predicates corresponding to non-output gates, we may define the base rules as we please, e.g. $\bot$ for all.

    The constructed program is clearly equivalent{} to $C$. Moreover, the computation delay is $\ordo(d)$ (or more precisely $\mathsf{depth}(C') + 1$) since it takes $\mathsf{depth}(C') + 1$ rounds to ``simulate'' each round of $C$. 
    The depth of the program is clearly $\ordo(k)$ if $C$ does not have bounded fan-in and $\ordo(1)$ if $C$ does have bounded fan-in.
    The size is clearly $\ordo(n + m)$ if $C$ does not have bounded fan-in and $\ordo(n)$ if $C$ has bounded fan-in. This is because for each gate $G$, the size of its corresponding base rule is $\ordo(1)$ and the size of its corresponding induction rule is either constant or linear in the fan-in of $G$. Thus, the number of edges adds to the size of the program, but a bound for the fan-in ensures that the number of edges stays linear in the size of the circuit.
\end{proof}

\section{Arithmetic with BNL}\label{Arithmetic with BL}

In this section, we show how to carry out integer addition (see Lemma \ref{lem:add_BL}) and multiplication (see Lemma \ref{lem:mul_BL}) in Boolean network logic in parallel. We then extend this demonstration to floating-point arithmetic, including floating-point polynomials and piecewise polynomial functions, see Theorems \ref{fp-addition}, \ref{BL multiplication} and \ref{piecewise-polynomial_simulation}.

Informally, the idea is to split both addition and multiplication into simple steps that are executed in parallel.
We will show that we can simulate integer arithmetic (respectively, floating-point arithmetic) by programs whose size is polynomial in the size of the integer format where the integers are represented (respectively, in the size of the floating-point format).
We also analyze the computation times of the constructed programs. The computation time is polylogarithmic in the size of the integers (and resp. in the size of the floating-point format) and sometimes even a constant. Ultimately, the same applies to floating-point piecewise polynomial functions.

\subsection{Integer arithmetic}\label{Integer arithmetic}

Let $\beta, p \in \Z_+$ and $\beta \geq 2$.
An \textbf{integer format} over $p$ and $\beta$ is the set $\cZ(p, \beta)$ of strings of the form $\pm d_{1} \cdots d_{p} \in \{+, -\} \times [0; \beta - 1]^{p}$ (excluding $-0 \cdots 0$), where $\beta$ is called the \textbf{base} and $p$ the \textbf{precision}.
We exclude the string $-0 \cdots 0$ from integer formats for simplicity, but we could also include it and obtain essentially the same results.
When $p$ and $\beta$ are clear from the context, we may omit them and simply write $\cZ$ instead of $\cZ(p, \beta)$. 
The semantics of a string of the above form is naturally the integer that the string represents in base $\beta$. In later parts, we may also identify a string of the above form with its corresponding integer for brevity when the format is clear from the context. Note that the ordinary linear ordering of integers induces a linear order in each integer format $\cZ$.

Given an integer $\pm d_1\cdots d_p$ in a format $\cZ(p, \beta)$, by \textbf{leading zeros} we mean a prefix $d_1 \cdots d_{m} = 0^m$ such that either $d_{m+1} \neq 0$ or $m = p$; we call $d_i$ a \textbf{leading zero} for each $i = 1, \dots, m$. 
For each $\bz \in \cZ(p, \beta)$ and each $p' \geq p$ there is trivially some $\bz' \in \cZ(p', \beta)$ such that $\bz$ and $\bz'$ represent the same integer; $\bz'$ is obtained simply by adding leading zeros to $\bz$.
On the other hand, if $\bz \in \cZ(p,\beta)$ has $m$ leading zeros, then for each $p'$ such that $p-m \leq p' < p$ there exists some $\bz' \in \cZ(p', \beta)$ such that $\bz$ and $\bz'$ represent the same integer, and $\bz'$ obtained simply by omitting enough leading zeros from $\bz$.

We next define how a \emph{halting} $\BNL$-program simulates integer functions with arbitrary precision $p$ and base $\beta \in \Z_+$, $\beta\geq 2$. 
Informally, we define an embedding from $\cZ(p, \beta)$ to bit strings of length $1+p\beta$; the strings are split into substrings of length $\beta$, where exactly one bit in each substring is $1$ and the others are $0$.
Formally, let $\cZ(p,\beta)$ be an integer format, let $\bz = \pm d_1 \cdots d_p \in \cZ$, let $\bs_{0} \in \{0, 1\}$ be a bit, and let $\bs_1, \ldots, \bs_p \in \{0,1\}^{\beta}$ be \textbf{one-hots}, i.e. bit strings with exactly one $1$. We say that $\bs \colonequals \bs_0 \bs_1 \cdots \bs_p$ \textbf{corresponds} to 
$\bz$ (w.r.t. $\cZ$)
if $\bs_{0} = 1$ exactly when the sign of $\bz$ is $+$ and for every $i \in [p]$, we have $\bs_i(d_i+1) = 1$
(and other values in $\bs_i$ are zero). We may also call $\bs$ a \textbf{one-hot representation} of 
the integer $\bz$. For example, if $\beta = 5$, then $1 \cdot 00100 \cdot 01000 \cdot 00001 \in \{0,1\}^{1 +  3 \cdot 5}$ corresponds to $+ 2 1 4 \in \cZ(3, 5)$. 
We likewise say that a bit string $\bs$ \textbf{corresponds} to (or is a \textbf{one-hot representation} of) a sequence $(\bz_{1}, \dots, \bz_{k}) \in \cZ^k$ of integers (w.r.t. $\cZ$) if $\bs$ is the concatenation of the bit strings that correspond to $\bz_{1}, \dots, \bz_{k}$ (w.r.t. $\cZ$) from left to right.

Using binary one-hot representations, we can represent integers in $\BNL$ by assigning each bit with a head predicate that is true if and only if the bit is $1$.

\begin{definition}\label{definition: simulation}
    Let $\cZ_1 \colonequals \cZ(p, \beta)$ and $\cZ_2 \colonequals \cZ(p', \beta)$ be integer formats. We say that a halting $\BNL$-program $\Lambda$ \textbf{simulates} (or computes) a function $f \colon \cZ_{1}^{\ell} \to \cZ_{2}^{k}$ if whenever a bit string $\bi$ corresponds to a sequence $(\bz_{1}, \ldots, \bz_{\ell}) \in \cZ^\ell_{1}$, the output $\Lambda(\bi)$ also corresponds to $f(\bz_{1}, \dots, \bz_{\ell}) \in \cZ_{2}^{k}$.
\end{definition}
Note that in the definition of simulation, the ``incorrect inputs'', i.e., the bit strings that do not correspond to a sequence of integers, are disregarded. However, we could define that the simulating program does not produce an output with an incorrect input, which in a sense corresponds to ``error handling''.

We start with an auxiliary result for comparing two $p$-length integers in an arbitrary base $\beta$. A comparison function is a function $f \colon \cZ^{2} \to \cZ$ over an integer format $\cZ$ such that for all $\bx, \by \in \cZ$: $f(\bx, \by) = + 0 \cdots 01$ if $\bx > \by$, and $f(\bx, \by) = + 0 \cdots 0$ otherwise.

\begin{lemma}\label{Integer comparison}
    Let $\cZ(p, \beta)$ be an integer format.
    Comparing two numbers in $\cZ$ can be simulated with a (halting) $\BNL$-program of size $\ordo(p \beta^{2} + p^{2})$ and computation time $2$ (where the program has $\ordo(p\beta)$ head predicates).
\end{lemma}
\begin{proof}
    Let $(\bx, \by) \in \cZ^2$ be an input for the comparison function $\cZ^{2} \to \cZ$. We assume that both numbers $\bx$ and $\by$ are positive since we can easily cover each case using conditional rules. We assign a head predicate $X_{i, d}$ for the digit $d \in [0;\beta-1]$ in position $i \in [p]$ (counting from left to right) of the first number $\bx$, and likewise predicates $Y_{i, d}$ for the second number $\by$. These are the input predicates with induction rules $X_{i, d} \colonminus X_{i, d}$ and $Y_{i, d} \colonminus Y_{i, d}$. This encoding means that if the $i$th digit (from left to right) of $\bx$ is $d$, then $X_{i, d}$ is true and $X_{i, d'}$ is false for each $d' \neq d$, and likewise for $\by$ and predicates~$Y_{i, d}$.

    To simulate the comparison function, we define a subprogram with auxiliary predicates $H_{i}$ and $L_{i}$ where $i \in [p]$. The intuition is that a predicate $H_{i}$ (resp. $L_i$) becomes true if $\bx$ has a higher (resp. lower) digit in position $i$ than $\by$. Each of these predicates has a base rule with the body $\bot$ and the induction rules are
    \[
    H_{i} \colonminus \bigvee\limits_{\substack{d, d' \in [0;\beta-1],\, d > d'}} X_{i, d} \land Y_{i, d'} \quad \text{and} \quad L_{i} \colonminus \bigvee\limits_{\substack{d, d' \in [0;\beta-1],\, d < d'}} X_{i, d} \land Y_{i, d'}.
    \]
    These head predicates can be computed in a single round and the total size of these rules is $\ordo(p\beta^2)$. 
    With these head predicates we can easily define a subprogram consisting of a single head predicate $C$ which is true iff $\bx > \by$ and otherwise false as follows: 
    \[
    C (0) \colonminus \bot, \qquad
    C \colonminus \bigwedge\limits_{i \in [p]} \Big( L_{i} \rightarrow \bigvee\limits_{\substack{j \in [p],\, j < i}} H_{j} \Big) \land \bigvee_{i \in [p]} H_{i}.
    \]
    The total size of the rules for $C$ is $\ordo(p^{2})$ and computing it from predicates $H_i$, $L_i$ requires a single iteration round.

    We also define the print predicates $P_0$ and $P_{i, k}$, where $i \in [p]$ and $k \in [0;\beta-1]$ with the following induction rules as follows. 
    The predicate $P_0$ encodes the sign which is always $+$, thus we define $P_0 \colonminus \top$. 
    Next we define auxiliary formulae: for each $k \in [0; \beta-1]$, we let $\theta_k$ denote $\bot$ if $k \in [\beta-1]$ and $\top$ for $k = 0$. 
    Now, we can define the rest of the print predicates: for each $i \in [p-1]$ and $k \in [0; \beta-1]$ we define $P_{i,k} \colonminus \theta_k$, and for $i = p$ we define $P_{p,0} \colonminus \neg C$, $P_{p, 1} \colonminus C$ and $P_{p,k} \colonminus \bot$ for all $k = 2, \dots, \beta-1$.
    By combining these constructed subprograms (using methods in Section \ref{section: tools}), we obtain a complete program of size $\ordo(p\beta^2 + p^2)$ and computation time $2$,
    which contains $\ordo(p \beta)$ distinct schema variables.
\end{proof}

\begin{remark}\label{remark:integer_comparison}
    Note that we can always compare two integers $\bx \in \cZ(p, \beta)$ and $\by \in \cZ(p', \beta)$ by letting $r = \max\{p, p'\}$, finding $\bx', \by' \in \cZ(r, \beta)$ such that $\bx$ and $\bx'$ (resp. $\by$ and $\by'$) represent the same integer and then comparing $\bx'$ and $\by'$ in $\cZ(r, \beta)$.
    We may thus always compare two integers in separate formats without specifying the details.
\end{remark}

\subsubsection{Parallel addition}\label{paral_add}

In this section, we construct a $\BNL$-program that simulates integer addition via a parallel algorithm. The algorithm is mostly well known and is based on how integer addition is computed in the circuit complexity class $\mathsf{NC}^1$ called Nick's class, i.e. we parallelize the textbook method (sometimes called the long addition algorithm or the carry-look ahead addition). Here, the main difference to integer addition in Nick's class (which is carried out in base $2$) is that we generalize the algorithm for arbitrary bases.

To illustrate our integer addition simulation method, consider the following example where we add $\bx = + 614$ and $\by = + 187$ in base $10$. Let $c_1, c_2$ and $c_3$ denote the \textbf{carry-over digits} and let $s_0$, $s_1$, $s_2$ and $s_3$ denote the digits of the sum $\bx + \by$ from left to right. We have
\[
c_3 = \left\lfloor \frac{4 + 7}{10} \right\rfloor = 1, \quad c_2 = \left\lfloor \frac{1 + 8 + c_3}{10} \right\rfloor = 1, \quad c_1 = \left\lfloor \frac{6 + 1 + c_2}{10} \right\rfloor = 0.
\]
Thus $s_3 = 1$ because $1 \equiv 4 + 7(\, \mathrm{mod}\, 10)$, $s_2 = 0$ since $0 \equiv 1+ 8 + 1(\, \mathrm{mod}\, 10)$, $s_1 = 8$ because $8 \equiv 6 + 1 + 1 (\, \mathrm{mod}\, 10)$ and $s_0 = c_3 = 0$.
Therefore $\bx + \by = + 0801$, as wanted. As we can see, in order to know that $c_2 = 1$ we have to first check if $c_3 = 1$. 
Generally, let $\bx = + x_1 \cdots x_p$ and $\by = + y_1 \cdots y_p$ be two positive numbers in a format $\cZ(p,\beta)$. The carry-over digits of $\bx + \by$ can then be computed in a similar fashion: $c_p = \left\lfloor \frac{x_p + y_p }{\beta} \right\rfloor$ and
$c_{i} = \left\lfloor \frac{x_{i} + y_{i} + c_{i+1} }{\beta} \right\rfloor$ for all $i \in [p-1]$. 
Now,
$\bx + \by = (x_1 \cdots x_p) + (y_1 \cdots y_p) = s_{0} \cdots s_{p} \eqqcolon \bs$,
where for each $j \in [p]$, $c_j$ is the carry-over digit and $s_j \equiv (x_{j} + y_{j} + c_{j+1})(\, \mathrm{mod}\, \beta)$ (if $j = p$, then $c_{j+1}$ is omitted), and $s_{0} = c_1$.
In other words, we have to check if the carry from the previous position has been propagated forward. This method is easy to generalize for other signs of $\bx$ and $\by$.

Now we are ready to prove the following lemma. In the statement below the \textbf{addition of two integers} $\bx$ and $\by$ in the integer format $\cZ(p, \beta)$ refers to a saturating sum $+ \colon \cZ^2 \to \cZ$ where we take a precise sum of $\bx$ and $\by$ and then map it to the closest number in $\cZ$. 
More precisely, with a \emph{precise sum} we mean the regular sum of the integers that $\bx$ and $\by$ represent over the ring of all integers. 
Note that the precise sum of $\bx$ and $\by$ might not be representable in the finite format $\cZ$ since it might exceed the maximum or minimum integer representable in $\cZ$, in which case any accurate representation of the precise sum in base $\beta$ has more than $p$ digits. 
Such numbers are mapped to either the greatest or lowest integer in $\cZ$ depending on the sign. Note that depending on the context, the notation $+$ may have different meanings throughout the paper, but we trust the meaning to always be clear from the context.
\begin{lemma}\label{lem:add_BL}
Let $\cZ(p, \beta)$ be an integer format.
The addition of two numbers in $\cZ(p, \beta)$ can be simulated with a (halting) $\BNL$-program of size $\ordo(p^3 + p\beta^2)$ and computation time $\ordo(1)$ (where the program has $\ordo(p \beta)$ head predicates).
\end{lemma}

\begin{proof}
We consider the case where both numbers are positive; in the end of the proof we give a short sketch of how the construction can be generalized for other signs. 
Let $\bx = + x_1 \cdots x_p$ and $\by = + y_1 \cdots y_p$ be integers in $\cZ(p, \beta)$. 
We encode the input integers $\bx$ and $\by$ to input predicates $X_{j, m}$ and $Y_{j,m}$ respectively where $j \in [p]$ and $m \in [0;\beta-1]$ using one-hot representations (and ignoring the signs for convenience). We define the induction rules $X_{j, m} \colonminus X_{j,m}$ and $Y_{j,m} \colonminus Y_{j,m}$. 
The body of the base rule for each head predicate defined henceforth is $\bot$.
The precise sum of $\bx$ and $\by$ is encoded by the predicates $S_{i, k}$, where $i \in [0;p]$ and $k \in [0;\beta - 1]$, for which we define the induction rules later.

For each $i \in [p]$ and $b \in \{0,1\}$, we define the head predicates $O_{i, b}$ and $C_{i}$ with the following intuitions.
Each predicate $O_{i, b}$ becomes true if the sum of the digits in position $i$ results in the carry-over digit $1$ \emph{presuming that the sum of digits in position} $i + 1$ has resulted in the carry-over digit $b$. Each predicate $C_{i}$ becomes true if the carry-over digit created in position $i$ is $1$, taking into account the whole sum.
We thus define the following induction rules:
$O_{p,1} \colonminus \bot$ and otherwise we define
\[
O_{i,b} \colonminus 
\bigvee_{n+m +b \geq \beta} 
(X_{i,n} \land Y_{i,m}), 
\quad\quad\quad 
C_i \colonminus \bigvee_{i \leq j \leq p} \Big( O_{j,0} \land \bigwedge_{i \leq k < j} \Big( O_{k,1} \Big) \Big).
\]

Now we can write rules for variables $S_{i, k}$, where $i \in [0;p]$ and $k \in [0;\beta - 1]$. For $i = p$ and $k \in [0;\beta-1]$, we have $S_{p,k} \colonminus \bigvee_{n + m \equiv k (\mathrm{mod}\, \beta)} ( X_{p, n} \land Y_{p, m} )$ and for $i \in [p-1]$, we have
\[
S_{i,k} \colonminus \bigvee_{n + m \equiv k (\mathrm{mod}\, \beta)} ( X_{i, n} \land Y_{i, m} \land \neg C_{i+1}) \lor \bigvee_{n + m +1\equiv k (\mathrm{mod}\, \beta)} ( X_{i, n} \land Y_{i, m} \land C_{i+1}).
\]
For $i = 0$, we have
$S_{0,0} \colonminus \neg C_1$, $S_{0,1} \colonminus C_1$, and $S_{0, k} \colonminus \bot$
for all $k \in [0;\beta-1] \setminus \{0,1\}$.
After three iteration rounds, the values of predicates $S_{i,k}$ have been computed correctly. 

Now, we can easily define print predicates $P_{i, k}$, where $i \in [p]$, $k \in [0;\beta-1]$ (ignoring the sign for convenience), which in one step calculate the correct output integer depending on the truth value of $S_{0,1}$ (after the values of $S_{i,k}$ have been computed correctly) as follows.
If the predicate $S_{0,1}$ is true after three iteration rounds, then we output the greatest integer in $\cZ$. Otherwise, we output the integer encoded by $S_{i, k}$, where $i \in [p]$ and $k \in [0; \beta-1]$.
More precisely, the rules are defined as follows: 
$P_{i,k} \colonminus (S_{0,1} \land \theta_k) \lor (\neg S_{0,1} \land S_{i,k})$
where $\theta_k \colonequals \top$ if $k = \beta - 1$ and otherwise $\theta_k \colonequals \bot$.
These one-step subprograms (the first consisting of predicates $O_{i, b}$, the second consisting of predicates $C_{i}$, the third consisting of predicates $S_{i, k}$, and the fourth consisting of print predicates $P_{i,k}$) can be combined into the complete program correctly without increasing the size and time complexity by using if-else statements explained in Section \ref{section: tools}.
We have $\ordo(p)$ predicates $O_{i,b}$ of size $\ordo(\beta^{2})$, $\ordo(p)$ predicates $C_{i}$ of size $\ordo(p^{2})$, $\ordo(p \beta)$ predicates $S_{i, k}$ of size $\ordo(\beta)$ and $\ordo(p \beta)$ predicates $P_{i,k}$ of size $\ordo(1)$. Thus, the total size is $\ordo(p \beta^{2} + p^3 + p\beta^2 + p \beta) = \ordo(p^3 + p\beta^2)$ for the whole program.
The program also contains $\ordo(p \beta)$ distinct schema variables.

If $\bx$ and $\by$ both have negative signs, we can use the same addition algorithm; the output simply includes a negative sign. If $\bx$ and $\by$ have opposite signs, we need to use a subtraction algorithm instead. First, we need to compare $\bx$ and $\by$ with signs omitted; in other words, we compare which number has a greater absolute value (this also determines the sign of the output). As stated before in Lemma \ref{Integer comparison}, this requires $\ordo(p \beta^{2} + p^{2})$ space and $2$ iteration rounds. Then, we modify the addition algorithm above in the following way. Instead of adding digits of $\bx$ and $\by$ together, we subtract them modulo $\beta$; the digits of the number with the smaller absolute value are subtracted from the digits of the number with the greater absolute value. If the subtraction of two digits goes past $0$, it results in a negative carry $-1$ that is added to the next subtraction. Otherwise the algorithm works in the same exact way, and thus does not increase the size and time complexity of the program.
\end{proof}
\begin{remark}\label{remark: precise add}
It is always possible to simulate the sum (or resp. the precise sum) of integers from two (possibly different) formats $\cZ(p, \beta)$ and $\cZ(p', \beta)$ by letting $r = \max\{p, p'\}$, representing both numbers in $\cZ(r, \beta)$ (resp. in $\cZ(r+1, \beta)$) and then adding them. 
We may thus always add two integers in separate formats without specifying the details.
\end{remark}

\subsubsection{Parallel multiplication}\label{paral_mul}

In this section, we introduce a parallel multiplication algorithm and show that it can be simulated in $\BNL$. The parallelization method is mostly well known and is based on splitting the multiplication into simple addition tasks. 
As with addition, in the statement below the \textbf{product of two integers} $\bx$ and $\by$ in the integer format $\cZ(p, \beta)$ refers to a saturating product $\cdot \colon \cZ^2 \to \cZ$, where we take a \emph{precise product} of $\bx$ and $\by$ and then map it to the closest number in $\cZ$. Just like we explained in parallel addition, a precise product means the regular multiplication of the integers that $\bx$ and $\by$ represent over the ring of all integers. As we noted with addition, the precise product can lead to a number that is not representable in the format $\cZ$ and depending on the sign such a number is mapped to either the greatest or lowest integer in $\cZ$. 
Respectively as with the parallel sum, we note that the notation $\cdot$ may have different meanings throughout the paper, but the meaning should always be clear from the context.
\begin{lemma}\label{lem:mul_BL}
Let $\cZ(p, \beta)$ be an integer format.
The product of two numbers in the format $\cZ$ can be simulated with a (halting) $\BNL$-program of size 
$\ordo(p^4 + p^3\beta + p^{2}\beta^2 + p\beta^3)$
and computation time $\ordo(\log(p) + \log(\beta))$ (where the program has $\ordo(p^{2} \beta + p \beta^{2})$ head predicates).
\end{lemma}

\begin{proof}
Suppose that we have two integers in the format $\cZ(p, \beta)$: a multiplicand $\bx$ and a multiplier $\by = \pm y_1 \cdots y_p$. Informally, the product of $\bx$ and $\by$ is computed in the following two steps.
\begin{enumerate}
    \item 
    We run $p$ different precise multiplications in parallel, where the multiplicand $\bx$ is multiplied by $y_j0 \cdots 0$ with $p-j$ zeros on the right (for each $j \in [p]$). Moreover, every such multiplication is computed in parallel using the parallel addition algorithm. 
    As a result, we obtain $p$ different numbers of length $2p$.
    \item 
    We add the numbers from the first step by using the parallel addition algorithm. 
\end{enumerate} 

More formally, to compute in parallel, we use the ``binary tree strategy'', where the bottom of the tree consists of the multiplications in step 1, and the higher levels of the tree consist of adding these products together in parallel. This can be formally constructed as follows. Let $d = \lceil \log (p) \rceil$. For each $i \in [0;d]$ and $j \in [2^{d-i}]$ we compute integers $\bz_{i,j}$. Intuitively, picture a binary tree with $d+1$ levels where each level $i$ (starting with level $i=0$ at the bottom) consists of integers $\bz_{i,j}$ and the number of integers on a level decreases logarithmically as we move up the tree. Each integer $\bz_{i,j}$ is calculated based on the level $i-1$ as follows. The integer $\bz_{0,j}$ is obtained as the product of $x$ and $y_{j} 0^{p-j}$, and each integer $\bz_{i+1, j}$ is obtained as the sum of integers $\bz_{i,2j-1}$ and $\bz_{i, 2j}$. Note that if $p < 2^{d}$, then for $i > p$ we can simply set $\bz_{0,j}$ to $0$. The sign of the final product is easy to obtain from the signs of $\bx$~and $\by$.

We compute each product $\bz_{0, j}$ in parallel with a different method. First, we compute the integers $2 \bx, \dots, (\beta-1) \bx$ as follows. First we let $0 \bx$ denote the representation of $0$ in the format and we let $1 \bx$ denote $\bx$. Then for each $i = 2, \dots, \beta-1$, we define that $i \bx$ is the precise sum of $k \bx$ and $(i-k) \bx$, where $k$ is the greatest power of $2$ such that $k < i$. Now, we define that $\bz_{0, j}$ is the integer $(y_j \bx) \cdot 0^{p-j}$; we have to add $p-j$ zeros to the product to account for the position of $y_{j}$, since $\bz_{0, j}$ is the product not of $\bx$ and $y_{j}$ but of $\bx$ and $y_{j}0^{p-j}$.

The described algorithm is easy (but relatively notation-heavy) to implement with $\BNL$ so we omit the implementation. 
To obtain the complete program, we use one-hot counters and if-else statements (introduced in Section \ref{section: tools}) to ``time'' the program properly without affecting the size and time. Next, we analyze the time and space complexity of our parallel algorithm.
First, each $\bz_{0, j}$ can be computed with parallel addition as described above with integers $i \bx$, which means that we run $\beta-2$ additions in parallel in $\ordo(\log(\beta))$ levels. Since each addition requires $\ordo(1)$ time and runs in parallel, we need $\ordo(\log(\beta))$ time to compute every $\bz_{0, j}$. Likewise, since the number of integers $i \bx$ is $\ordo(\beta)$ and each addition needs size $\ordo(p^3 + p\beta^2)$, the total size required is $\ordo(p^3\beta + p\beta^3)$.

Next, for the other integers $\bz_{i, j}$, where $i > 0$, the total size and time are calculated as follows. We have $\lceil \log(p) \rceil$ layers, each performing an addition algorithm that lasts for $\ordo(1)$ rounds, which adds up to $\ordo(\log(p))$. On the other hand, we have $\ordo(p)$ additions of size $\ordo(p^3 + p\beta^2)$ each, which gives us a size increase of $\ordo(p^4 + p^{2}\beta^2)$.

The total amount of time required is thus $\ordo( \log(\beta)) + \ordo(\log(p)) = \ordo(\log(p) + \log(\beta))$ and the amount of space required is $\ordo(p^4 + p^3\beta + p^{2}\beta^2 + p\beta^3)$. The program requires $\ordo(p^{2} \beta + p \beta^{2})$ distinct schema variables, because we have $\ordo(\beta)$ additions required to obtain integers $\bz_{0, j}$ and another $\ordo(p)$ additions to obtain $\bz_{d, 1}$, and each addition requires $\ordo(p \beta)$ distinct schema variables.
\end{proof}
\begin{remark}\label{remark: precise mul}
It is always possible to simulate multiplication (resp. precise multiplication) of two numbers in two (possibly different) formats $\cZ(p, \beta)$ and $\cZ(p', \beta)$ by letting $r = \max\{p, p'\}$, representing both numbers in $\cZ(r, \beta)$ (resp. $\cZ(2r, \beta)$) and then multiplying them. We may thus always multiply two integers from any two formats without specifying the details.
\end{remark}

\subsection{Floating-point arithmetic}\label{fpa}

In this section, we consider floating-point arithmetic, including polynomials and piecewise polynomial functions. We show that $\BNL$-programs can simulate these in polynomial space and in 
polylogarithmic time, and some simple arithmetic operations even in constant time.

\subsubsection{Floating-point formats}\label{section: float formats}

Floating-point numbers are a syntactical way for computers to represent rational numbers over a large range with the same relative accuracy for numbers both close to and far from zero. This is intuitively achieved by defining numbers via two integers. The first integer defines the sign of the number and a sequence of so-called significant digits in the number's representation in the chosen base (e.g. the binary or decimal representation). The second integer scales the first integer; intuitively, it expresses how far the number is from zero.
Floating-point numbers also allow the approximation of real numbers with some obvious rounding errors.
For example, the decimal representation of $\pi$ is $3.14159 26535\ldots$, but it may be approximated as a floating-point number $0.314159 \times 10^1$; here the first integer expressing the significant digits is $314159$ and the second integer expressing the distance from zero is the exponent~$1$ (numbers with a negative exponent are even closer to zero).

Now, we formally define floating-point numbers. Let $p, q, \beta \in \Z_+$ where $\beta \geq 2$.
A \textbf{floating-point number} over $p$, $q$ and $\beta$ is a pair $(\bs, \be)$, where $\bs = \pm d_1\cdots d_p \in \cZ(p, \beta)$ and $\be = \pm e_1 \cdots e_q \in \cZ(q, \beta)$. 
We let $\bbf$ denote the string $\pm 0.d_1 \cdots d_p$, where the sign is the sign of $\bs$.
The semantics of a floating-point number $(\bs,\be)$ is the (rational) number represented by the string $\bbf \times \beta^{\be}$ in the natural way.
For such a number, we call $\bs$ the \textbf{significand}, $\bbf$ the \textbf{fraction}, $p$ the \textbf{significand precision} or \textbf{fraction precision}, the dot between $0$ and $d_1$ in $\bbf$ the \textbf{radix point}, $\be$ the \textbf{exponent}, $q$ the \textbf{exponent precision} and $\beta$ the \textbf{base} (or \textbf{radix}).
Multiple floating-point numbers over $p$, $q$ and $\beta$ may represent the same rational number; in practice, a one-to-one correspondence is obtained via a so-called normalized form.
We call $(\bs, \be)$ \textbf{normalized} (w.r.t. $\beta$) if $d_1 \neq 0$ or if $\bs = +0^p$ and $\be$ is the smallest integer representable in $\cZ(q, \beta)$. 
Some floating-point numbers over $p$, $q$ and $\beta$ may represent a number very close to zero that has no corresponding normalized floating-point number over $p$, $q$ and $\beta$; in practice, such numbers can be represented in subnormalized form.
We call $(\bs, \be)$ \textbf{subnormalized} (or \textbf{de-normalized})
(w.r.t. $\beta$) 
if $d_1 = 0$, $\bs\neq +0^p$ and $\be$ is the smallest integer representable in $\cZ(q, \beta)$.

A \textbf{floating-point format over $p$, $q$ and $\beta$} is the set $\cF(p, q, \beta)$ of all normalized and subnormalized floating-point numbers 
over $p$, $q$ and $\beta$. 
Moreover, a \textbf{normalized floating-point format over $p$, $q$ and $\beta$} is the set $\cF_{\mathsf{n}}(p,q,\beta)$ of all floating-point numbers in $\cF(p, q, \beta)$ that are normalized (w.r.t. $\beta$).
An \textbf{unrestricted floating-point format over $p$, $q$ and $\beta$} is the set $\cF_{\mathsf{u}}(p,q,\beta)$ of all floating-point numbers over $p$, $q$ and $\beta$.
If $p, q$ and $\beta$ are clear from the context, we may omit them and simply write $\cF$ (resp. $\cF_{\mathsf{n}}$ and $\cF_{\mathsf{u}}$).
Note that each number that is representable in $\cF$ or $\cF_{\mathsf{n}}$ has a unique corresponding representation in that system, i.e., in $\cF$ and $\cF_{\mathsf{n}}$ the semantics and syntax match one-to-one.
On the other hand, in $\cF_{\mathsf{u}}$ a number might have multiple corresponding representations. 
Unrestricted formats $\cF_{\mathsf{u}}$ will be mainly used as an implementation tool, while arithmetic operations for floats will be defined such that both inputs and outputs belong to a format $\cF$ (respectively, to a normalized format $\cF_{\mathsf{n}}$).
In the rest of the paper, if a format or normalized format is clear from the context, then we may identify a floating-point number in the format with the rational number that it represents.
Note that since floating-point numbers represent rational numbers, the ordinary linear ordering of rational numbers induces a linear order in each floating-point format $\cF$ and each normalized floating-point format $\cF_{\mathsf{n}}$.

In practice, it is often assumed that the inputs and outputs of floating-point operations are normalized, as this makes the algorithms faster, more uniform and offers simple bounds for errors, 
see \cite{Knuth_art, DBLP:books/sp/18/MBDJ2018}. 
Subnormalized numbers are used to obtain more accurate results near zero (e.g. IEEE 754 standard).
For example, consider the subnormalized number $0.01 \times 10^{-99}$ in $\cF(2,2,10)$; it has the normalized form $0.10 \times 10^{-100}$ in the format $\cF(2,3,10)$, but no normalized form in $\cF(2,2,10)$.
Normalized formats $\cF_{\mathsf{n}}$ do not include subnormalized numbers, resulting in a larger gap around zero.
Some floating-point formats in practice also include special numbers such as $+ \infty$, $- \infty$ and NaN (not a number); we do not include them in our formats and their inclusion would not change our results much, but we discuss this in more detail in Section \ref{section:remarks}.

\subsubsection{Representing floating-point numbers in binary}\label{one-hot}

Our way of representing floating-point numbers of arbitrary base in binary is based on international standards (e.g. IEEE 754). Informally, we define an embedding $\sigma$ from floating-point numbers in a format $\cF(p,q,\beta)$ to binary strings, where for a given floating-point number $F \in \cF$, $\sigma(F)$ represents $F$ in binary as follows. The
the first two bits encode the signs of the exponent and fraction.
The next $q \beta$ bits encode the exponent (minus the sign) in base $\beta$, and the last $p \beta$ bits encode the significand (minus the sign) in base $\beta$. The embedding naturally extends for sequences of floats.

More formally, let $\cF(p,q,\beta)$ be a floating-point format and let $\fF \in \{\cF_{\mathsf{u}}, \cF, \cF_{\mathsf{n}}\}$. Let
$(\bs, \be)\in \fF$ where $\bs = \pm d_1 \cdots d_p$ and $\be = \pm e_{1} \cdots e_{q}$. Let $\bp_1, \bp_2 \in \{0,1\}$ and $\bb_1, \ldots, \bb_q, \bb'_1, \ldots, \bb'_p \in \{0,1\}^{\beta}$. We say that $\bb \colonequals \bp_1 \bp_2 \bb_1 \cdots \bb_q \bb'_1 \cdots \bb'_p$ \textbf{corresponds} to $(\bs, \be)$ w.r.t. $\fF$ 
(or $\bs$ is a \textbf{one-hot representation} of $(\bs, \be)$ w.r.t. $\fF$) 
if $\bp_1 \bb_1 \cdots \bb_q$ corresponds to $\be$ and $\bp_2 \bb'_1 \cdots \bb'_p$ corresponds to $\bs$.
For example, the floating-point number $-0.2001 \times 3^{+120}$ in the format $\cF(4,3, 3)$ has the following corresponding bit string w.r.t. $\cF$ or $\cF_{\mathsf{n}}$:
\[
\underbrace{1}_{\bp_1}\cdot \underbrace{0}_{\bp_2} \cdot \underbrace{010 \cdot 001 \cdot 100}_{\bb_{1} \bb_{2} \bb_{3}} \cdot \underbrace{001 \cdot 100 \cdot 100 \cdot 010}_{\bb'_{1} \bb'_{2} \bb'_{3} \bb'_{4}}.
\] 
We say that a bit string $\bb$ \textbf{corresponds} to a sequence 
$(F_{1}, \dots, F_{k}) \in \fF^{k}$
w.r.t. $\fF$ if $\bb$ is a concatenation of bit strings that correspond to the floating-point numbers
$F_{1}, \dots, F_{k}$ 
w.r.t. $\fF$ in the same order from left to right. 
Each floating-point number in a format $\fF \in \{\cF, \cF_{\mathsf{n}}\}$ has a unique one-hot representation w.r.t. $\fF$, but possibly various one-hot representations w.r.t. $\cF_{\mathsf{u}}$.
We let $\sigma \colon \fF^* \to \{0,1\}^*$ denote the embedding where $\sigma(F_1, \ldots, F_n)$ is the one-hot representation of $(F_1, \dots, F_n)$ and we call $\sigma$ the \textbf{embedding of $\fF$ into binary}.

It is easy to generalize correspondence to mixed sequences of floating-point numbers and integers.
Let $k \in \Z_+$ and for each $i \in [k]$ let $\cS_i$ be either an integer format, a floating-point format, a normalized floating-point format or unrestricted floating-point format. Let $\fS = (\cS_1, \dots, \cS_k)$. We say that a bit string $\bb$ \textbf{corresponds} to a sequence $(x_1, \dots, x_k) \in \cS_1 \times \cdots \times \cS_k$ (w.r.t. $\fS$) if $\bb$ is a concatenation of bit strings $\bb_1, \dots, \bb_k$ such that for each $i \in [k]$ the bit string $\bb_i$ corresponds to $x_i$ (w.r.t. $\cS_i$).

\begin{definition}
    Let $\fS = (\cS_{1}, \ldots, \cS_{k})$ and $\fS' = (\cS_{1}', \ldots, \cS_{\ell}')$ be sequences of formats. We say that a halting $\BNL$-program $\Lambda$ \textbf{simulates} a given function $f \colon \fS \to \fS'$ if the following holds for all 
    $(x_1, \ldots, x_k) \in \fS$:
    if a bit string $\bi$ corresponds to $(x_1, \ldots, x_k)$ 
    (w.r.t. $\fS$), then
    the output $\Lambda(\bi)$ corresponds to 
    $f(x_1, \ldots, x_k)$ 
    (w.r.t.~$\fS'$).
\end{definition}
Analogously to the simulation of integer functions, we could define that the simulating program does not produce an output with inputs that do not correspond to a sequence of integers and floats, which in a sense would correspond to ``error handling''.

\subsubsection{Shifting}\label{shifting}

As an important tool to implement floating-point arithmetic in $\BNL$, we introduce \emph{shifting}, which informally means transforming a floating-point number to another floating-point number that represents the same rational number. In a bit more detail, this means adding leading zeros to the right of the radix point (or resp. removing leading zeros from the right of the radix point) and adding the number of added zeros to the exponent (or resp. subtracting the number of removed zeros form the exponent). 
Moreover, the program for shifting to the right is defined w.r.t. an upper bound for how much a floating-point number can be shifted to keep our programs finite and small in size.
Importantly, we do not perform rounding when shifting; rounding is a separate operation that will be described in the next section.

More formally, let $\cF(p,q,\beta)$ be a floating-point format, let $n \in \N$ and let $(\bs, \be) \in \cF_{\mathsf{u}}$. Assuming $\bs = \pm d_1 \cdots d_p$, \textbf{$(\bs, \be)$ shifted to the right by $n$ steps} refers to any float $(\pm 0^n d_1 \cdots d_p 0^k, \be')$, where $\be'$ is the precise sum of $\be$ and the number $n$ (represented in base $\beta$ with any precision), $k \in \N$ and the sign of the significand stays the same. Analogously, assuming $\bs = \pm 0^m d_1 \cdots d_{p-m}$ and $n \leq m$, \textbf{$(\bs, \be)$ shifted to the left by $n$ steps} refers to a float $(\pm 0^{m-n} d_1 \cdots d_{p-m} 0^n, \be')$, where $\be'$ is the precise subtraction of $n$ (represented in base $\beta$ with any precision) from $\be$ and the sign of the significand stays the same.  
For example, the floating-point number $+ 0.235 \times 10^{+9}$ shifted to the right by $2$ steps is $+ 0.00235 \times 10^{+11}$ in the format $\cF_{\mathsf{u}}(5,2,10)$. 

Now, we can define related shifting functions formally.
Let $\cF$ be a floating-point format in base $\beta$, let $n \in \N$, let $\cZ$ be the smallest integer format (in base $\beta$) where $n$ is representable and let $\cF'$ be the smallest floating-point format (in base $\beta$) such that $\cF_{\mathsf{u}}'$ contains each number in $\cF_{\mathsf{u}}$ shifted to the right by $n$ steps. 
\textbf{Shifting to the right in $\cF_{\mathsf{u}}$ by at most $n$ steps} is a function of type 
$
\cF_{\mathsf{u}} \times \cZ \to \cF_{\mathsf{u}}'
$ 
that maps each pair $((\bs, \be), \bz)$ to the floating-point number $(\bs', \be')$ as follows. If $\bz$ represents an integer $\ell \in [0;n]$, then $(\bs', \be')$ is the number $(\bs, \be)$ shifted to the right by $\ell$ steps; if $\bz$ represents a number greater than $n$ or less than zero, then $(\bs', \be')$ is the normalized representation of zero.

Let $\cZ'$ be the smallest integer format (in base $\beta$) where $p$ is representable, and let $\cF''$ be the smallest floating-point format (in base $\beta$) such that $\cF_{\mathsf{u}}''$ contains each number in $\cF_{\mathsf{u}}$ shifted to the left by any number of steps.
\textbf{Shifting to the left in $\cF_{\mathsf{u}}$} refers to a function of type 
$
\cF_{\mathsf{u}} \times \cZ' \to \cF_{\mathsf{u}}''
$
that maps each pair $((\bs, \be), \bz)$ to a floating-point number $(\bs', \be')$ as follows. Assume that $\bs = \pm 0^m d_1 \cdots d_{p-m}$. If $\bz$ represents a number $\ell \in [0;m]$, then $(\bs', \be')$ is the number $(\bs, \be)$ shifted to the left by $\ell$ steps. If $\bz$ represents a number greater than $m$ or less than zero, then $(\bs', \be')$ is the normalized representation of zero.

\begin{lemma}\label{lemma:shifting}
    Let $\cF(p,q,\beta)$ be a floating-point format and let $r= \max\{p,q\}$. Shifting to the right in $\cF_{\mathsf{u}}$ by at most $r$ steps (resp. to the left in $\cF_{\mathsf{u}}$) can be simulated with a (halting) $\BNL$-program of size $\ordo(r^{3} + r^2\beta + r \beta^{2})$ and computation time $\ordo(1)$ (where the program has $\ordo(r \beta)$ head predicates).
    Moreover, shifting to the right in $\cF_{\mathsf{u}}$ by at most one step can be simulated by a program of size $\ordo(r^3 + r\beta^2)$ and computation time $\ordo(1)$. 
\end{lemma}
\begin{proof}
It is easy but notation-heavy to implement shifting to the left and right in $\BNL$, but we go through the implementation informally and also consider the size and time complexity of the required $\BNL$-program. We consider the case where we shift to the right, as shifting to the left is analogous, and we do not consider the signs of the numbers as they are trivial to deal with.
Let the head predicates $F_{i, b}$, where $i \in [p]$ and $b \in [0; \beta-1]$, encode the fraction of the input float (without sign). 
Let $r' \colonequals \lfloor \log_{\beta}(r) \rfloor + 1$ denote the number of digits required to represent $r$ in base $\beta$.
Let the head predicates $Z_{j, b}$, where 
$j \in [r']$
and $b \in [0; \beta-1]$, encode the input integer (without sign). 
For each positive integer $x \in [0;r]$ represented by $\bx \colonequals x_1 \cdots x_{r'}$ in $\cZ(r', \beta)$, we define the induction rule 
\[
X_x \colonminus \bigwedge_{i \in [r']} \bigg{(}Z_{i, x_i} \land \bigwedge_{b \in [0;\beta-1]\setminus \{x_i\}} \neg Z_{i, b}\bigg{)},
\]
which encodes $\bx$ and the total required size is 
$\ordo(r \beta \log_{\beta}r)$.
Let us define for each $b \in [0;\beta-1]$ an auxiliary formula $\chi_b$ as follows: $\chi_0 \colonequals \top$ and for each $b \neq 0$ we define $\chi_b \colonequals \bot$.

Now, the shifted fraction can be encoded with the head predicates $F'_{i,b}$, where $i \in [p+r]$ and $b \in [0; \beta-1]$, where each $F_{i,b}'$ has the following induction rule: 
\[
F'_{i, b} \colonminus \bigwedge_{\substack{x \in [0;r],\, i -x \in [p] }} \bigg{(} X_{x} \rightarrow F_{i-x, b} \bigg{)} \land \bigwedge_{\substack{x \in [0;r],\, i -x \notin [p] }} \bigg{(} X_{x} \rightarrow \chi_b \bigg{)}.
\]
The total size of these induction rules is $\ordo(r^2\beta)$ and they can be computed in a single round. 
The addition of the exponent can be done by Lemma \ref{lem:add_BL} in $\ordo(r^{3} + r \beta^{2})$ space and $\ordo(1)$ time, since $\cF'_{\mathsf{u}}$ is the smallest format which contains each number in $\cF$ shifted to the right by at most $r$ steps, which means that the exponent precision is $q' = \max\{q+1, r'+1\} \leq r+1$.

Cases where the input integer is greater than $r$ are easy to handle by defining a head predicate $E$ that becomes true if $X_{x}$ is true for some $x \in [0;r]$ as follows: $E \colonminus \bigvee_{x \in [0;r]} X_{x}$. This can be used to create conditional rules for the print predicates: if $F_{i,b}' \colonminus \psi$ is the induction rule for $F_{i,b}'$ as defined above, then we replace it with the conditional rule for $F_{i,b}'$ w.r.t. $(E, \psi, \chi_b)$, i.e., $F_{i,b}' \colonminus (\psi \land E) \lor (\chi_b \land \neg E)$.
In the case of shifting to the left, there is an additional step that precedes this: we have to first encode the number of leading zeros and then use Lemma \ref{Integer comparison} to compare the encoding to the input integer. The exponent variables can be dealt with in a similar manner. Clearly this does not increase the size or time complexity of the program.

Using the tools in Section \ref{section: tools}, it is easy to combine the steps described into a halting $\BNL$-program which performs shifting and requires 
$\ordo(r^{3} + r^2\beta + r \beta^{2})$ space, $\ordo(r \beta)$ distinct head predicates, and $\ordo(1)$ computation time. Moreover, when we consider shifting to the right by at most $1$ step, the total size of head predicates $X_x$ is only $\ordo(\beta)$ and the total size of head predicates $F_{i,b}'$ is only $\ordo(r \beta)$, so the size of the program is that of addition, i.e. $\ordo(r^3 + r \beta^2)$.
\end{proof}

\subsubsection{Normalizing a floating-point number}

In this section, we consider the normalization of floating-point numbers. Given an unrestricted floating point format $\cF_{\mathsf{u}}$, \textbf{normalization} (in $\cF_{\mathsf{u}}$) refers to a function of type $\cF_{\mathsf{u}} \to \cF$ or $\cF_{\mathsf{u}} \to \cF_{\mathsf{n}}$ defined as follows. A normalization function of type $\cF_{\mathsf{u}} \to \cF$ maps each floating-point number to the normalized or subnormalized floating-point number that represents the same rational number.
A normalization function of type $\cF_{\mathsf{u}} \to \cF_{\mathsf{n}}$ does the same except that floating-point numbers that would be mapped to subnormalized numbers are instead mapped to zero.
The following lemma shows how the normalization of floating-point numbers can be simulated with $\BNL$ in 
cubic
space and constant time w.r.t. a given floating-point format.
\begin{lemma}\label{normalization}
    Let $\cF(p, q, \beta)$ be a floating-point format and let $r = \max\{p,q\}$. Normalization of a floating-point number in $\cF_{\mathsf{u}}$ can be simulated with a (halting) $\BNL$-program of size 
    $\ordo(r^{3} + r^2\beta + r \beta^{2})$
    and computation time $\ordo(1)$ (where the program has $\ordo(r \beta)$ head predicates).
\end{lemma}

\begin{proof}
We describe the $\BNL$-program consisting of four subprograms which are the computational steps required for normalization. 
Because the rules of the program are very notation-heavy, but not that complicated to write, we only describe individual and important rules of the subprograms during the description. We also analyze the space and time complexity of each subprogram as we go. 

Let $(\bs, \be)$ be a floating-point number in the format $\cF_{\mathsf{u}}$ where $\bs =  d_{1} \cdots d_{p}$ and $\be = e_{1} \cdots e_{q}$ (we omit the signs w.l.o.g.). We normalize $(\bs, \be)$ to the format $\cF$ in the following steps (the case for $\cF_{\mathsf{n}}$ is the same, but easier). Let $\be_{\max}$ (and resp. $\be_{\min}$) denote the maximum (resp. the minimum) exponent in $\cF$, i.e. the maximum and minimum integers in $\cZ(q, \beta)$, and let $r = \max\{p,q\}$.
Below, we assume that $\bs$ is encoded by the input predicates $F_{i, k}$, where $i \in [p]$ and $k \in [0; \beta-1]$.
\begin{enumerate}
    \item First, we test whether $\bs = 0 \cdots 0$.
    If $\bs \neq 0 \cdots 0$, we move to step 2.
    If $\bs = 0 \cdots 0$, then $(\bs, \be)$ represents the number $0$; we set the floating-point number to be the representation of $0$ in $\cF$ and output it. 
    This step is trivial to implement in $\BNL$ in a single computation step. For example, testing if $\bs = 0 \cdots 0$ can be done by an auxiliary rule $A_{1} \colonminus \bigwedge_{i \in [p]} F_{i,0}$ which only adds $\ordo(p)$ to the size.
    \item 
    Let $m$ be the greatest number for which $d_1 \cdots d_{m} = 0 \cdots 0$.
    Determining $m$ can be done with the following auxiliary rules:
    $A_n \colonminus \bigwedge_{i \in [n]} F_{i, 0} \land \neg F_{n+1, 0}$ 
    for each $n \in [p]$ (we omit $\neg F_{n+1, 0}$ if $n = p$); the total size of these rules is $\ordo(p^2)$.
    We encode $m$ as an integer $\bm_{\beta} \in \cZ(r, \beta)$ as follows.
    We define head predicates $M_{j, b}$ where $j \in [\lfloor \log_{\beta}(p) \rfloor+1]$ (i.e. $j$ is at most the number of digits required to represent $p$ in base $\beta$) and $b \in [0; \beta-1]$. The rule for each $M_{j,b}$ is simply a disjunction of exactly those predicates $A_n$ where the $j$th digit of the representation of $n$ in base $\beta$ is $b$; the total size of these rules is $\ordo(r \log_{\beta}(r) \beta)$.
    \item We compute the precise sum $\be - \bm_{\beta}$ by Remark \ref{remark: precise add}, which by Lemma \ref{lem:add_BL} takes $\ordo(r^3 + r \beta^2)$ space and $\ordo(1)$ computation steps.
    \item We compare $\be - \bm_{\beta}$ obtained in step 3 to $\be_{\min}$ which requires size $\ordo(r\beta^2 + r^2)$ and time $\ordo(1)$ by Lemma \ref{Integer comparison}, and shift the number to the left, where the number of steps shifted depends on which inequality holds:
    \begin{enumerate}
        \item If $\be_{\min} \leq \be - \bm_{\beta}$, then we shift $(\bs, \be)$ to the left by $m$ steps, i.e., we run a program for shifting to the left with $(\bs, \be)$ and $\bm_\beta$ as input. The resulting exponent is representable in $\cZ(q, \beta)$.
        \item If $\be - \bm_{\beta} < \be_{\min}$, then we shift $(\bs, \be)$ to the left by the number of steps represented by the precise sum $\be + \be_{\max}$, i.e., we run a program for shifting to the left with inputs $(\bs, \be)$ and the precise sum $\be + \be_{\max}$.
        Again the precise sum $\be + \be_{\max}$ takes $\ordo(r^3 + r\beta^2)$ and $\ordo(1)$ computation steps. 
        The shifting results in an exponent that represents the same number as $\be - (\be + \be_{\max}) = \be_{\min}$.
    \end{enumerate}
    By Lemma \ref{lemma:shifting}, shifting takes $\ordo(r^3 + r^2 \beta + r\beta^2)$ space and $\ordo(1)$ computation steps.
    In both cases the significand is expressible in $\cZ(p,\beta)$ and the exponent in $\cZ(q,\beta)$ and thus the whole floating-point number is representable in $\cF(p, q, \beta)$. Thus, we have obtained the desired number.
\end{enumerate}
With these steps completed, the floating-point number is either normalized or subnormalized, and each step is possible to execute with $\BNL$ in the described space and time. To combine the subprograms together, we use the tools described in Section \ref{section: tools}. As each of the subprograms requires only $\ordo(r \beta)$ distinct schema variables, so does the whole program.
\end{proof}

\subsubsection{Rounding a floating-point number}

In this section, we consider the implementation of rounding operations in $\BNL$. Given two floating-point formats $\cF$ and $\cF' \subseteq \cF$ (in the same base), \textbf{rounding} (from $\cF$ to $\cF'$) refers to a function $\cF \to \cF'$ mapping each number $(\bs, \be) \in \cF$ to a number $(\bs', \be') \in \cF'$ such that $(\bs', \be')$ is either the smallest number in $\cF'$ such that $(\bs, \be) \leq (\bs', \be')$ or the greatest such that $(\bs', \be') \leq (\bs, \be)$.
The easiest way to round a floating-point number is \textbf{truncation}, where the least significant digits of the number are simply omitted, rounding the number toward zero. 
With this technique, a rounding function $\cF(3,1,10) \to \cF(2,1,10)$ would round both $0.115 \times 10^3$ and $0.114 \times 10^3$ to $0.11 \times 10^3$.
On the other hand, the most common method is to round to the nearest number in the format, with ties rounding to the number with an even least significant digit (the rightmost digit in the significand). This is called \textbf{round-to-nearest ties-to-even}. With this technique, the number $0.115 \times 10^3$ would instead round to $0.12 \times 10^3$. We will use round-to-nearest ties-to-even as it is the most common and arguably the most precise method, but other rounding methods like truncation are equally possible in our framework.

The following theorem shows that round-to-nearest ties-to-even can be simulated by a $\BNL$-program in 
cubic
space and constant time w.r.t. a given floating-point format.
\begin{lemma}\label{lemma: rounding}
Let $\cF(p, q, \beta)$ be a floating-point format, let $p' \geq p$ and $q' \geq q$ and let $r = \max\{p', q'\}$. Rounding of a floating-point number in $\cF(p', q', \beta)$ to the floating-point format $\cF(p, q, \beta)$ 
using round-to-nearest ties-to-even can be simulated with a (halting) $\BNL$-program of size 
$\ordo(r^{3} + r \beta^{2})$,
and computation time $\ordo(1)$ (where the program has $\ordo(r \beta)$ head predicates).  
\end{lemma}
\begin{proof}
    Consider a float $(\bs', \be') \in \cF(p', q', \beta)$ where $\bs' = \pm  d_{1} \cdots d_{p'}$ and $\be' = \pm e_{1} \cdots e_{q'}$, which we seek to round to the format $\cF(p, q, \beta)$ (where $p' \geq p$ and $q' \geq q$) using round-to-nearest ties-to-even.

    Informally, round-to-nearest ties-to-even can be simulated and implemented in the following two steps in $\BNL$. 
    Without loss of generality, we assume that the signs of the fraction and exponent are both $+$, since the other cases are analogous.
    \begin{enumerate}
        \item \textbf{Description:} We check the value of $d_{p + 1}$ to round the fraction correctly. If $d_{p + 1} < \frac{\beta}{2}$, then we let $\bs = d_{1} \cdots d_{p}$. If $d_{p + 1} > \frac{\beta}{2}$, 
        then we let $\bs = d_{1} \cdots d_{p} + 0^{p - 1} 1$ denote the precise sum in $\cZ(p+1, \beta)$ obtained by Lemma \ref{lem:add_BL} and Remark \ref{remark: precise add}.
        (We round to the nearest number in both cases.) If $d_{p + 1} = \frac{\beta}{2}$, then we let $\bs = d_{1} \cdots d_{p}$ if $d_{p}$ is even and $\bs = d_{1} \cdots d_{p} + 0^{p - 1} 1$ if $d_{p}$ is odd (again the sum is precise). In other words, in the case of a tie we round to the nearest number whose rightmost digit is even. Then we move to the next step.
        
        \textbf{Implementation:} Assume that the significand $\bs'$
        is encoded in the variables $F_{i, b}$, where $i \in [p']$ and $b \in [0; \beta-1]$. Checking if
        $d_{p + 1} < \frac{\beta}{2}$ can be done with an auxiliary rule $A_{<} \colonminus \bigvee_{b < \beta / 2} F_{p+1, b}$. Similarly, we can define auxiliary predicates, which check if $d_{p + 1} > \frac{\beta}{2}$ or $d_{p + 1} = \frac{\beta}{2}$. Using these head predicates as conditions, it is easy to round the fraction. In the worst case we need to perform integer addition which requires size $\ordo(p^{3} + p\beta^{2})$ and time $\ordo(1)$ by Lemma~\ref{lem:add_BL}.
        \item \textbf{Description:}
        Let $\be_{\max}$ denote the greatest number in $\cZ(q, \beta)$.
        First, if $\bs \in \cZ(p, \beta)$, then we output $(\bs, \be)$ where $\be \in \cZ(q, \beta)$ represents the same integer as $\be'$. If $\bs \notin \cZ(p, \beta)$, then we check if $\be' \geq \be_{\max}$; if yes, we output the largest number in $\cF(p,q,\beta)$ and if no, then we shift to the right by one and run the program again starting from step 1 with the shifted floating-point number as input.

        \textbf{Implementation:} It is easy to define subprograms for this step with the subprograms for integer comparison in Lemma \ref{Integer comparison} and for shifting to the right by one in Lemma \ref{lemma:shifting}; together these require size 
        $\ordo(r^{3} + r \beta^{2})$
        and time $\ordo(1)$ where $r = \max\{p', q'\}$.
    \end{enumerate}
    For each step above it is easy but notation-heavy to write a complete simulating subprogram. Using if-else statements and while loops, it is possible to combine the subprograms into a single $\BNL$-program which simulates rounding based on round-to-nearest ties-to-even. Combining the size and time complexities of the steps above, the program requires space 
    $\ordo(r^{3} + r \beta^{2})$,
    time $\ordo(1)$ and $\ordo(r \beta)$ distinct schema variables.
\end{proof}

\subsubsection{Addition of floating-point numbers}

In this section, we show that we are able to simulate floating-point addition via $\BNL$-programs in constant computation time and 
cubic
space w.r.t. a given floating-point format.
\textbf{Floating-point addition} in a format $\cF$ refers to a saturating sum $+ \colon \cF^{2} \to \cF$ (respectively, floating-point addition in $\cF_{\mathsf{n}}$ is an operation $+_{\mathsf{n}} \colon \cF_{\mathsf{n}}^{2} \to \cF_{\mathsf{n}}$), where we first take the precise sum of two numbers $\bx$ and $\by$, and then round it to the closest float in the format by using round-to-nearest ties-to-even (in the case of normalized addition, rounded values between the two non-zero numbers closest to zero in $\cF_{\mathsf{n}}$ are taken to zero). 
By precise sum, we mean the floating-point number that represents the regular sum of the rational numbers represented by $\bx$ and $\by$ over the field of all rational numbers. By saturation, we refer to the fact that if the precise sum is greater (or lesser) than any number in the format $\cF$, then the result is simply the greatest (or least) number in $\cF$. We consider addition both in the case where we have subnormalized numbers and in the case where we do not.

\begin{lemma}\label{fp-addition}
Let $\cF(p, q, \beta)$ be a floating-point format and $r = \max\{p, q\}$.
Addition of two floating-point numbers in $\cF$ (resp. in $\cF_{\mathsf{n}}$) can be simulated with a (halting) $\BNL$-program where the program has size 
$\ordo(r^{3} + r^2\beta + r \beta^{2})$
and computation time $\ordo(1)$ (and $\ordo(r \beta)$ head predicates).
\end{lemma}

\begin{proof}
We describe the steps used when adding numbers together, as well as the size and time complexity they require in $\BNL$, rather than writing out the rules of the program. Adding two floating-point numbers together takes advantage of integer addition. 

\begin{enumerate}
    \item If either of the input numbers is zero, we skip the remaining steps and output the other input number
    (if both are zero, we output zero).
    \item We compare the exponents to see which number is greater; this is because we need to shift the lesser number before adding the fractions together. Comparing the exponents requires $\ordo(q \beta^{2} + q^{2})$ space and $\ordo(1)$ time steps by Lemma \ref{Integer comparison}. The greater exponent is stored in new predicates, which requires size $\ordo(q \beta)$ and only a single computation step.
    \item We perform the precise subtraction of the lesser exponent from the greater one by Remark \ref{remark: precise add} to determine the difference in magnitude of the two numbers; this requires $\ordo(q^{3} + q \beta^{2})$
    space and $\ordo(1)$ time by Lemma \ref{lem:add_BL}. 
    \item We compare the difference of the exponents to $p$, which requires $\ordo(r \beta^{2} + r^{2})$ space and $\ordo(1)$ time by Lemma \ref{Integer comparison}. If the difference is greater than $p$, i.e., the lesser number has no effect on the sum, then we skip the remaining steps and output the greater input number.
    Otherwise, we move to step $5$.
    \item We shift the floating-point number with the lower exponent to the right by the difference of the exponents (at most $p$) using the program for shifting to the right by at most $r$ steps which requires 
    $\ordo(r^{3} + r^2\beta + r \beta^{2})$
    space and $\ordo(1)$ time steps by Lemma \ref{lemma:shifting}. The number that is not shifted is then represented in the same format as the shifted number. Note that after the shifting, the exponents are the same.
    \item We perform precise integer addition on the significands obtained in the previous step by using
    Remark \ref{remark: precise add} which requires $\ordo(p^{3} + p \beta^{2})$
    space and $\ordo(1)$ time (note that the precision of the significands is $\ordo(p)$).
    Let $\be$ be the exponent of the greater input number.
    If $\bs'$ is the precise sum of the significands and 
    $\be' = \be+ 0 \cdots 01$,
    then we define $F = (\bs', \be')$. 
    Let $\cF_{\mathsf{u}}(p', q', \beta)$ be the unrestricted floating-point format where $F$ belongs to (note that it is the same for all non-zero inputs).
    \item We normalize the floating-point number $F$ to the format $\cF(p', q', \beta)$.
    We then round the number to the floating-point format $\cF(p,q,\beta)$.
    In the case of addition in $\cF_{\mathsf{n}}(p, q, \beta)$, 
    as an extra step we normalize to the format $\cF_{\mathsf{n}}(p, q, \beta)$ to flush the subnormalized numbers to zero.
    In either case, this requires 
    $\ordo(r^{3} + r^2\beta + r \beta^{2})$ size and $\ordo(1)$ time by Lemmas \ref{normalization} and \ref{lemma: rounding}.
\end{enumerate}

With the above steps completed, the sum has been computed, and each step is executable in $\BNL$ in the described space and time. The subprograms can be combined together in the desired space and time by using the tools described in Section \ref{section: tools}. As each of the subprograms requires only $\ordo(r \beta)$ distinct schema variables, so does the whole program.
\end{proof}

\subsubsection{Multiplication of floating-point numbers}

In this section, we show that we are able to simulate floating-point multiplication via $\BNL$-programs in logarithmic time and 
quartic
space w.r.t. a given floating-point format.
\textbf{Floating-point multiplication} in a format $\cF$ is the saturating operation $\cdot \colon \cF^{2} \to \cF$ (respectively, floating-point multiplication in $\cF_{\mathsf{n}}$ is an operation $\cdot_{\mathsf{n}} \colon \cF_{\mathsf{n}}^{2} \to \cF_{\mathsf{n}}$) where we first take the precise product of two numbers $\bx$ and $\by$  and then round it to the closest number in the format by using round-to-nearest ties-to-even (in normalized multiplication, rounded values between the two non-zero values in $\cF_{\mathsf{n}}$ closest to zero are taken to zero).
Analogously to floating-point addition, by the precise product we mean the floating-point number that represents the regular product of the rational numbers represented by $\bx$ and $\by$ over the field of all rational numbers.
The multiplication requires logarithmic time, since the proof applies the result obtained for integer multiplication in Lemma \ref{lem:mul_BL}.
 
\begin{lemma}\label{BL multiplication}
Let $\cF(p, q, \beta)$ be a floating-point format and $r = \max\{p, q\}$.
Multiplication of two floating-point numbers in $\cF$ (resp. in $\cF_{\mathsf{n}}$) can be simulated with a (halting) $\BNL$-program of size 
$\ordo(r^4 + r^3\beta + r^{2}\beta^2 + r\beta^3)$
and computation time $\ordo(\log(r) + \log(\beta))$ (where the program has $\ordo(r^{2} \beta + r \beta^{2})$ head predicates).
\end{lemma}
\begin{proof}
    As before, we describe how to calculate the product of two floating-point numbers and describe the size and time complexity each step requires in $\BNL$. Informally and simply, we add the exponents together and multiply the fractions. 
    Formally, multiplication of two floating-point numbers is done in the following steps.

\begin{enumerate}
    \item We take the precise sum of the exponents of the floating-point numbers by using Remark \ref{remark: precise add}. This adds $\ordo(q^{3} + q \beta^{2})$
    space and $\ordo(1)$ time steps by Lemma \ref{lem:add_BL}. 
    We let $\be$ denote the sum of the exponents. 
    \item We take the precise product of the significands by using Remark \ref{remark: precise mul}. This results in an integer $\bs$ of length $2p$. Multiplication requires 
    $\ordo(p^4 + p^3\beta + p^{2}\beta^2 + p\beta^3)$
    space and $\ordo(\log(p) + \log(\beta))$ time by Lemma \ref{lem:mul_BL}. 
    \item 
    Let $\cF_{\mathsf{u}}(p', q', \beta)$ denote the format that $(\bs, \be)$ belongs to (note that it is the same for all inputs, and the precision $q'$ of the exponent is $\ordo(q)$ and the precision $p'$ of the significand is $\ordo(p)$).
    We normalize the number $(\bs, \be)$ to the format $\cF(p',q', \beta)$.
    We then round the number to the format $\cF(p,q,\beta)$.
    Finally, as an extra step for multiplication in $\cF_{\mathsf{n}}$, we normalize the number to $\cF_{\mathsf{n}}(p,q,\beta)$, flushing all subnormalized numbers to zero.
    In either case, by Lemmas \ref{normalization} and \ref{lemma: rounding} this
    requires 
    $\ordo(r^{3} + r^2\beta + r \beta^{2})$
    space and $\ordo(1)$ time.
\end{enumerate}

With the above steps completed, we have calculated the product of the two numbers, and each step is executable in $\BNL$ in the described size and time.
These subprograms are combined using the tools described in Section \ref{section: tools}, giving the desired space and time. The number of distinct schema variables is $\ordo(r^{2} \beta + r \beta^{2})$ as integer multiplication requires that many of them, and the other subprograms only require $\ordo(r \beta)$ of them.
\end{proof}

\subsubsection{Floating-point polynomials and piecewise polynomial functions}\label{section: float polynomials}

In this subsection, we consider floating-point polynomials and piecewise polynomial functions over floats. 
For context, over real numbers, a piecewise polynomial function refers to a real function which is defined as separate polynomial functions over certain intervals of reals, e.g., the function ``$f(x) = x^{2}$ when $x \geq 0$ and $f(x) = -x$ when $x < 0$'' is piecewise polynomial. 
We begin by defining the syntax and the semantics for piecewise polynomial functions over floating-point numbers and then show in Theorem \ref{piecewise-polynomial_simulation} that such functions can be simulated in polynomial space and polylogarithmic time in $\BNL$.

Before we define floating-point polynomials and piecewise polynomial functions over floats, we define iterated addition and powering over floats.
Let $\fF$ be a floating-point format or normalized floating-point format.
For each $n \in \N$, the \textbf{$n$-ary iterated parallel floating-point sum over $\fF$} is the function $\mathrm{SUM}_n \colon \fF^n \to \fF$ defined as follows. For $n = 0$, $\mathrm{SUM}_0$ denotes the floating-point number that represents $0$ in $\fF$. Now assume $n \neq 0$ and let $j = \lceil \log(n) \rceil$. Given an input $(s_{1}, \dots, s_{n}) \in \fF^n$, for each $i \in [2^j]$ we let $s_{i,0} \colonequals s_{i}$ if $i \in [n]$ and otherwise $s_{i,0}$ is the floating-point number zero in $\fF$. We iteratively calculate for each $\ell \in [j]$ and each $i \in [2^{j-\ell}]$ the floating-point sum $s_{i, \ell+1} = s_{2i-1, \ell} + s_{2i, \ell}$ in $\fF$. Finally, we define that $\mathrm{SUM}_n(s_{1}, \dots, s_{n}) \colonequals s_{1, j}$.
Next for each $n \in \N$ and $F \in \fF$, we define $F^n$, i.e., \textbf{$F$ to the power of $n$} as follows. 
We define that $F^1 \colonequals F$ and for each $i \geq 2$ we define that $F^{i}$ is the floating-point product of $F^k$ and $F^{i-k}$ in $\fF$, where $k$ is the greatest power of $2$ such that $k < i$. Naturally, $F^0$ denotes the floating-point number that represents $1$ in $\fF$. Note that rounding is performed in each addition and product, but this is natural since it also happens in practical applications as well.

Given $n \in \N$, a (single-variable) \textbf{floating-point polynomial (of order $n$)} over a floating-point format $\fF$ 
is a sequence $\bp \colonequals (a_n, \dots, a_0)$ in $\fF^{n+1}$, where $a_n$ does not represent $0$
and $\bp$ represents an expression of the form $a_{n} x^{n} + a_{n-1} x^{n-1} + \cdots + a_{1} x + a_{0}$.
Moreover, the semantics for a floating-point polynomial $\bp$ is defined as follows. 
Each floating-point polynomial $(a_n, \dots, a_0)$ induces a function $p \colon \fF \to \fF$ where 
\[
p(x) = \mathrm{SUM}_{n+1}(a_n x^n, a_{n-1} x^{n-1}, \dots, a_1 x, a_0).
\]
The semantics of a floating-point polynomial is the function $p'$ obtained from $p$ by replacing floating-point numbers in $\fF$ with the rational numbers they represent.

Our method for evaluating polynomials is essentially the same as the method in \cite{MUNRO1973189}, p.~191. In that article, the method does not consider the space or time complexities of the evaluations of the powers $x^{i}$, the products $a_{i}x^{i}$ or the sum of terms $a_{i}x^{i}$. We believe that other similar evaluation methods for polynomials can be implemented efficiently in $\BNL$.

A (single-variable) \textbf{piecewise polynomial function over $\fF$ (with $\ell$ pieces)} is a pair $(s_1, s_2)$, where $s_1$ is a sequence $(\bp_1, \dots, \bp_{\ell})$ of floating-point polynomials $\bp_i$ over $\fF$ and $s_2$ is a sequence 
$((L_1, U_1), (L_2, U_2), \dots, (L_{\ell}, U_{\ell}))$ of a lower bound $L_i \in \fF$ and an upper bound $U_i \in \fF$ for each $i \in [\ell]$, 
where $L_i < U_i < L_{i+1}$ for each $i \in [\ell-1]$, and for each $F \in \fF$ there is an $i \in [\ell]$ such that $L_{i} \leq F \leq U_{i}$ (i.e. $U_{i}$ and $L_{i+1}$ are adjacent floats for each $i \in [\ell-1]$ with $L_1$ and $U_{\ell}$ being the least and greatest numbers in $\fF$ respectively).
A piecewise polynomial function induces a function $f \colon \fF \to \fF$  such that $f(x) = p_{i}(x)$ if
$L_{i} \leq x \leq U_{i}$.
We may use the syntactic pair $\rho \colonequals (s_1, s_2)$ to denote its induced function $f_\rho$ and for any floating-point number $F \in \fF$, we let $\rho(F)$ denote $f_\rho(F)$.
The semantics of a piecewise polynomial function is the function $g$ defined otherwise the same as $f$ except that we replace the input and output floating-point numbers with the rational numbers they represent.
The \textbf{order} of a piecewise polynomial function $(s_1, s_2)$ over $\fF$, where $s_1 = (\bp_1, \dots, \bp_{\ell})$, is the maximum order of a polynomial $\bp_1, \dots, \bp_{\ell}$.
We obtain the following theorem for simulating piecewise polynomial floating-point functions.

\begin{theorem}\label{piecewise-polynomial_simulation}
    Let $\alpha$ be a piecewise polynomial function over a floating-point format $\cF(p, q, \beta)$ (or $\cF_{\mathsf{n}}$) and let $r = \max\{p, q\}$. 
    Let $\Omega$ be the order of $\alpha$ (setting $\Omega := 2$ if the order is $1$ or $0$).
    and let $P$ be the number of pieces in $\alpha$. We can construct a $\BNL$-program $\Lambda$ that simulates $\alpha$ such that
    \begin{enumerate}
        \item the size of $\Lambda$ is $\ordo(P \Omega (r^4 + r^3\beta + r^{2}\beta^2 + r\beta^3))$,
        \item there are $\ordo(P \Omega(r^{2} \beta + r \beta^{2}))$ head predicates in $\Lambda$, and
        \item the computation time of $\Lambda$ is $\ordo(\log (\Omega) (\log(r) + \log(\beta)))$.
    \end{enumerate}
\end{theorem}
\begin{proof}
    The simulation of $\alpha$ starts with determining which of the $P$ polynomials is to be calculated by comparing the input $x$ simultaneously to each of the $P-1$ interval limits of $\alpha$. By Lemma~\ref{Integer comparison}, comparing $P$ fractions and exponents requires size $\ordo(P (r \beta^{2} + r^{2}))$ and time $\ordo(1)$. Then, we calculate the polynomial. 
    We first calculate the powers $x^2, \dots, x^{\Omega}$ in parallel, which means $\lceil \log(\Omega)\rceil$ sequential multiplications and $\Omega-1$ products in total to compute $x^{\Omega}$; by Lemma \ref{BL multiplication}, the powers thus require 
    $\ordo(\Omega(r^4 + r^3\beta + r^{2}\beta^2 + r\beta^3))$ 
    space and
    $\ordo(\log(\Omega)(\log(r) + \log(\beta)))$
    time steps.
    After this, for each of the $P$ polynomials we calculate at most $\Omega$ floating-point products $a_{i} x^{i}$, which require space 
    $\ordo(P \Omega (r^4 + r^3\beta + r^{2}\beta^2 + r\beta^3))$ 
    and time $\ordo(\log(r) + \log(\beta))$.
    
    Finally, for each of the $P$ polynomials there is an $(\Omega+1)$-ary iterated parallel floating-point addition which requires $\lceil \log(\Omega)\rceil$ sequential additions and $\Omega$ additions in total. Therefore, computing the $(\Omega+1)$-ary iterated parallel floating-point additions
    in total requires space 
    $\ordo(P \Omega(r^{3} + r^2\beta + r \beta^{2}))$
    and time $\ordo(\log(\Omega))$, since a single addition requires 
    $\ordo(r^{3} + r^2\beta + r \beta^{2})$
    space and $\ordo(1)$ time steps
    by Lemma~\ref{fp-addition}.
    Both the size and time complexity introduced by addition are thus dwarfed by the earlier multiplications. 
    The complete program can be obtained by using the tools described in Section~\ref{section: tools} without increasing size or time.
    Thus, the program contains $\ordo(P \Omega(r^{2} \beta + r \beta^{2}))$ distinct schema variables as there are $P \Omega$ multiplications which each contain $\ordo(r^{2} \beta + r \beta^{2})$ distinct schema variables.
    Finally, as $\alpha$ consists entirely of additions and multiplications, the size and time complexities are not affected by whether the operations are w.r.t. a floating-point format or a normalized format.
\end{proof}

\section{Descriptive complexity for general neural networks}\label{descsection}

In this section, we establish connections between Boolean network logic and neural networks, see Theorems \ref{NN_to_BNL} and \ref{BNL_to_NN}. Informally, we define a general neural network as an edge-weighted directed graph (with any topology) operating on floating-point numbers in some floating-point format. Each node receives either a fixed initial value or an input as its first activation value. In each communication round, each node sends its activation value to its out-neighbours and calculates a new activation value as follows. Each node multiplies the incoming activation values from its in-neighbours with associated weights, adds them all together along with a node-specific bias and feeds the result into a node-specific activation function. Note that floating-point formats are bounded, and the input space of a neural network is thus finite, but the obtained translations are not trivial to obtain as we care about the size and time complexities.

\subsection{General neural networks}\label{section: general neural networks}

Next, we define neural networks formally. We consider neural networks over any floating-point format  or normalized format $\fF$.
A \textbf{(directed) graph} is a tuple $(V, E)$, where $V$ is a finite set of \textbf{nodes} (also referred to as ``vertices'' in literature) and $E \subseteq V \times V$ is a set of \textbf{edges}. Note that we allow self-loops in graphs, i.e. edges $(v,v) \in E$.
A (heterogeneous) \textbf{general neural network} $\cN$ for a floating-point format $\fF$ is defined as a tuple $(G, \fa, \fb, \fw, \pi)$, where $G = (V, E, <^V)$ is a directed graph associated with a linear order $<^V$ for nodes in $V$, and $\fa$, $\fb$, $\fw$ and $\pi$ are defined below. The network $\cN$ contains sets $I, O \subseteq V$ of \textbf{input} and \textbf{output} nodes respectively, and a set $H = V \setminus (I \cup O)$ of \textbf{hidden nodes}. The tuples $\fa = (\alpha_{v})_{v \in V}$ and $\fb = (b_{v})_{v \in V}$ are assignments of a piecewise polynomial \textbf{activation function} $\alpha_{v}$ over $\fF$ and a \textbf{bias} $b_{v} \in \fF$ for each node $v \in V$. Likewise, $\fw = (w_{e})_{e \in E}$ is an assignment of a \textbf{weight} $w_{e} \in \fF$ for each edge $e \in E$. The function $\pi \colon (V \setminus I) \to \fF$ assigns an initial value to each non-input node.
If each node has the same activation function, then we say that $\cN$ is \textbf{homogeneous}.
For each $(u, v) \in E$, we say that $u$ is an \textbf{in-neighbour} of $v$ and $v$ is an \textbf{out-neighbour} of $u$.
The \textbf{in-degree} (resp., \textbf{out-degree}) of a node $v$ is the number of in-neighbours (resp., out-neighbours) of $v$. Note that we allow reflexive loops so a node might be its own ``neighbour''.

The computation of a general neural network is defined w.r.t. an input function of type $i \colon I \to~\fF$. Similar to $\BNL$-programs, an input function $i$ induces a floating-point string $\bi \in \fF^{\abs{I}}$, and respectively, a floating-point string $\bi$ induces an input function $i$. The \textbf{state of the network at time $t \in \N$} is a function $g_{t} \colon V \to \fF$, which is defined recursively as follows. For $t = 0$, we have $g_{0}(v) = i(v)$ if $v$ is an input node and $g_{0}(v) = \pi(v)$ if $v$ is not an input node. Now, assume we have defined the state at time~$t$.
The state of $v$ at time $t + 1$ is
\[
g_{t + 1}(v) = \alpha_{v} \Big( b_{v} + \sum_{(u,v) \in E} \left( g_{t}(u) \cdot w_{(u,v)} \right) \Big).
\]
In more detail, $\sum$ above denotes the iterated sum $\mathrm{SUM}_{n}$ where $n$ is the number of neighbours of $v$ and the inputs $g_{t}(u) \cdot w_{(u,v)}$ are given in increasing order w.r.t. $<^V$.
Given a node $v \in V$, its \textbf{activation value} at time $t$ is $g_t(v)$.
If we designate that $u_{1}, \dots, u_{k}$ enumerate the set $O$ of output nodes in the order $<^{V}$,
then the state of the system induces an output tuple $o_{t} = (g_{t}(u_{1}), \dots, g_{t}(u_{k}))$ at time $t$ for all $t \in \N$.

We once again define two frameworks for designating output rounds: machine-internal and machine-external. In the first framework, the set $V$ contains a set $A$ of \textbf{attention nodes} $u$, each of which is associated with a \emph{threshold} $s_{u} \in \fF$; the order of the nodes induces a \textbf{threshold string} $\bt \in \fF^{\abs{A}}$. Intuitively, an output round is triggered whenever the activation value of any attention node exceeds its threshold. In the second framework, attention nodes and thresholds are excluded and the neural network is instead given an \textbf{attention function} $a \colon \fF^{\abs{I}} \to \wp(\N)$ that gives the output rounds based on the input. 
Analogously to programs, whenever necessary for the context, we will tell whether a neural network is associated with an attention function or not.

Next, we define how a neural network generates a run sequence. Let $v_{1}, \dots, v_{n}$ enumerate the nodes of a neural network $\cN$ (in the order $<^{V}$). Let $i \colon I \to \fF$ be an input function that induces an input $\bi \in \fF^{\abs{I}}$. The neural network $\cN$ with input $\bi$ induces the sequence $C_{\cN} = (\bs_{t})_{t \in \N}$ called the \textbf{network state sequence}, where $\bs_{t} = g_{t}(v_{1}) \cdots g_{t}(v_{n})$. The set $O$ of output nodes corresponds to the set $P_{\cN} \colonequals \{\, i \mid v_{i} \in O \,\}$ of print positions. If the network has a set $A$ of attention nodes, then $A$ corresponds to the set $A_{\cN} \colonequals \{\, i\mid v_{i} \in A \,\}$ of attention positions. 
Moreover, the set of corresponding attention strings is defined as follows: $\cS_\cN \colonequals \{\, \bs \in \fF^{\abs{A}} \mid 
\bs(i) \geq \bt(i) \text{ for some } i \in \{1, \dots, \abs{A}\}
\,\}$, where $\bt$ is the threshold string of $\cN$. 
Thus, $\cN$ with input $\bi$ \textbf{induces} the run sequence of $C_{\cN}$ w.r.t. $(A_{\cN}, P_{\cN}, \cS_\cN)$ or we may simply say that $\cN$ with input $\bi$ induces the run sequence w.r.t. $(A, O, \bt)$. 
If the network has no attention nodes, and is instead given an attention function, then the output rounds are given by $a(\bi)$. Therefore, a neural network with an input $\bi$ induces \textbf{output rounds} and an \textbf{output sequence} w.r.t. $(A, O, \bt)$ (or resp. w.r.t. $(a(\bi), O)$).

We then define some parameters that will be important when describing how neural networks and $\BNL$-programs are related in terms of space and time complexity. The \textbf{degree} of a general neural network $\cN$ is the maximum in-degree of a node in the underlying graph. The \textbf{piece-size} of $\cN$ is the maximum number of ``pieces'' across all its piecewise polynomial activation functions. 
The \textbf{order} of $\cN$ is the highest order 
of its piecewise polynomial activation functions. 

A general neural network can easily emulate typical \emph{feedforward neural networks} where we require that the graph of the general neural network is, inter alia, connected and acyclic and the sets $I$, $O$ and $H$ are chosen correctly. We study feedforward neural networks in Section~\ref{sec: fnn}, where we prove that any feedforward neural network can be translated into an equivalent feedforward-style $\BNL$-program and vice versa.

In general, our neural network models are \emph{recurrent} in the sense that they allow loops. 
They are \emph{one-to-many} networks, i.e., they can map any and all inputs to a sequence of outputs, unlike feedforward neural networks, which always map each input to a single output.

\subsection{Notions on equivalence between programs and neural networks}\label{subsection:equivalence_of_NNs}

In this section, we define equivalence between $\BNL$-programs and general neural networks (analogously to programs). We start by defining a general notion of equivalence and then we define two special cases of equivalence that will be used in theorems.
Let $\fF$ be an arbitrary floating-point format or normalized float format, let $\sigma_1$ be a similarity from $\fF$ to $\{0,1\}$ and let $\sigma_2$ be a similarity from $\{0,1\}$ to $\fF$. The below definitions generalize for $\SC$-programs by replacing each occurrence of a $\BNL$-program with an $\SC$-program. They likewise generalize for self-feeding circuits by replacing each instance of a $\BNL$-program with a self-feeding circuit.

We say that a $\BNL$-program $\Lambda$ is \textbf{equivalent{} w.r.t. $\sigma_1$} to a general neural network $\cN$ for $\fF$ if for each input $\bi$ for $\cN$ the run sequence induced by $\Lambda$ with input $\sigma_1(\bi)$ is equivalent{} w.r.t. $\sigma_1$ to the run sequence induced by $\cN$ with input~$\bi$.
Respectively, a general neural network $\cN$ for $\fF$ is \textbf{equivalent{} w.r.t. $\sigma_2$} to a $\BNL$-program $\Lambda$ if
for each input $\bj$ for $\Lambda$ the run sequence of $\cN$ with input $\sigma_2(\bj)$ is equivalent{} w.r.t. $\sigma_2$ to the run sequence induced by $\Lambda$ with input $\bj$.
\textbf{Strong equivalence w.r.t. $\sigma_1$ or $\sigma_2$} is defined analogously, i.e., the run sequences must be strongly equivalent w.r.t. $\sigma_1$ or $\sigma_2$ with any inputs that are similar w.r.t. $\sigma_1$ or $\sigma_2$. 
The concepts of \textbf{time delay} and \textbf{precomputation rounds} between programs and neural networks are defined analogously to how they were defined between programs.
We once again also require that the time delay and number of precomputation rounds are the same regardless of input.

The above definitions generalize for two general neural networks as follows. Let $\fF'$ be another floating-point format or normalized float format and let $\sigma$ be a similarity from $\fF$ to $\fF'$. 
A neural network $\cN'$ for $\fF'$ is equivalent{} to a neural network $\cN$ for $\fF$ if for each input $\bi$ for $\cN$ the run sequence induced by $\cN'$ with input $\sigma(\bi)$ is equivalent{} w.r.t. $\sigma$ to the run sequence induced by $\cN$ with input $\bi$.
As with programs, attention functions of equivalent{} objects are trivial to define by using Remark~\ref{remark: outputs and attention}.

Next, we define three important similarities between $\BNL$-programs and neural networks that we will use to obtain our translations. The similarities are canonical and not trivializing. 
\begin{enumerate}
    \item The first canonical similarity is a similarity $\rho_1$ from $\fF$ to $\{0,1\}$ and it is actually the one that we defined previously in Section \ref{one-hot} where it was called an embedding of $\fF$ into binary. Recall that $\rho_1$ thus maps each sequence of floating-point numbers in $\fF$ to a corresponding one-hot representation.
    \item The second canonical similarity is a similarity $\rho_2$ from $\{0,1\}$ to $\fF'$ and it 
    is induced by an arbitrary embedding $\sigma \colon \{0,1\} \to \fF'$ as follows. 
    We define for each bit string $b_1 \cdots b_n \in \{0,1\}^*$ that $\rho_2(b_1 \cdots b_n) = (\sigma(b_1), \dots, \sigma(b_n))$.
    \item Taking the canonical similarity $\rho_1$ from $\fF$ to $\{0,1\}$ and a canonical similarity $\rho_2$ from $\{0,1\}$ to $\fF'$, we obtain a canonical similarity $\rho_2 \circ \rho_1$ from $\fF$ to $\fF'$.
\end{enumerate}
For the rest of the paper, when we say that \textbf{1)} a $\BNL$-program is equivalent{} to a general neural network, \textbf{2)} a general neural network is equivalent{} to a $\BNL$-program, or \textbf{3)} a general neural network is equivalent{} to another general neural network, we mean equivalent{} w.r.t. a corresponding canonical similarity.

\subsection{From NN to BNL}

We provide a translation from general neural networks to Boolean network logic. The proof is based on the results obtained for floating-point arithmetic in Section \ref{fpa}. Note that in the theorem below the obtained program is equivalent{} w.r.t. a canonical similarity as defined in Section \ref{subsection:equivalence_of_NNs}.
\emph{Informally, the following theorem states that we can translate a general neural network to an equivalent program with a small blow-up in size and small computation delay.}

\begin{theorem}\label{NN_to_BNL}
    Let $\cF(p, q, \beta)$ be a floating-point format, let $r = \max\{p,q\}$ and let $\cN$ be a (possibly heterogeneous) general neural network for $\cF$ (or for $\cF_{\mathsf{n}}$) with $N$ nodes, 
    $L$ edges, 
    degree $\Delta$ and piece-size~$P$. Let $\Omega$ be the order of $\cN$
    (setting $\Omega := 2$ if the order is $1$ or $0$).
    We can construct an equivalent{} $\BNL$-program $\Lambda$ such that
    \begin{itemize}
        \item the size of $\Lambda$ is $\ordo((N P \Omega + L) (r^4 + r^3\beta + r^{2}\beta^2 + r\beta^3))$,
        \item there are 
        $\ordo((N P\Omega + L)(r^{2} \beta + r \beta^{2}))$ 
        head predicates in $\Lambda$, and
        \item the computation delay of $\Lambda$ is $\ordo(\log(\Omega)(\log(r) + \log(\beta)) + \log(\Delta))$.
    \end{itemize}
\end{theorem}

\begin{proof}
Intuitively, the $\BNL$-program must, for each node in the neural network, carry out the aggregation of the activation values of its in-neighbours and then simulate the activation function, and repeat this process indefinitely via a while loop. The head predicates are split such that for each node there are separate head predicates which calculate its activation values. Each one-hot representation of a number in $\cF$ contains $\ell = 2 + \beta(p + q)$ bits, and so each node has $\ell$ head predicates that encode its activation value. For input nodes, these are input predicates, and for output nodes, these are print predicates. Additionally, each attention node has an associated attention predicate. Whereas a neural network can calculate new activation values in one step, doing so with $\BNL$ would blow up the size of the program, and so we use the algorithms introduced in Section \ref{fpa}.

First, we handle aggregations. 
Each of the $N$ nodes simultaneously calculates the aggregation of its in-neighbours' activation values, which consists of 
$L$ multiplications and $L$ additions across the whole graph (including the addition of biases).
The size of the aggregation is thus 
$L$
times the sum of the size of multiplication in Lemma \ref{BL multiplication} 
and addition in Lemma \ref{fp-addition}, which totals 
$\ordo(L (r^4 + r^3\beta + r^{2}\beta^2 + r\beta^3))$. 
The multiplications are done all at once in parallel, after which the additions are performed in a parallel 
sequence of length at most $\lceil \log(\Delta) \rceil$.
Thus, the time required is the time for multiplication in Lemma \ref{BL multiplication} plus $\lceil \log(\Delta) \rceil$ times the time for addition in Lemma \ref{fp-addition}, and therefore each aggregation phase requires in total $\ordo(\log(r) + \log(\beta)) + \ordo(\log(\Delta))$ time.

Then, we move on to activation functions. Each of the $N$ nodes simultaneously calculates a piecewise polynomial function over $\cF$; by Theorem \ref{piecewise-polynomial_simulation} this requires size 
$\ordo(N P \Omega (r^4 + r^3\beta + r^{2}\beta^2 + r\beta^3))$
and time 
$\ordo(\log(\Omega)(\log(r) + \log(\beta)))$.
Combining this with the aggregation, we get a program of size 
$\ordo((N P \Omega + L) (r^4 + r^3\beta + r^{2}\beta^2 + r\beta^3))$ 
and computation delay 
$\ordo(\log(\Omega)(\log(r) + \log(\beta)) + \log(\Delta))$.

As for the attention nodes, each of them requires a subprogram that compares the node's activation value to its threshold value, which is stored in normalized or subnormalized form. This subprogram then consists of \textbf{1)} 
normalizing the activation value,
\textbf{2)} comparing the exponents and \textbf{3)} comparing the significands. At worst, we have $N$ attention nodes and thus $N$ instances of shifting and $2N$ comparisons. This does not affect the size and time complexity of the program.

As these steps must be iterated indefinitely, we also include a one-hot counter in the program that tells when to loop. By calculating the amount of time $n$ required to simulate aggregation and activation functions in the worst case scenario (thus, $n$ is linear in the computation delay), we define the one-hot counter $T_{0}, \dots, T_{n}$ where the predicates are used as conditions to ensure that the calculations are carried out synchronously for each node. In particular, $T_{n}$ is used as a condition for attention predicates, which check if the attention threshold is surpassed in any attention node, and $T_{0}$ is used to trigger ``communication'' between nodes. The counter does not add to the space or time complexity of the program.

Lastly, as the program contains $L$ multiplications with $\ordo(r^2 \beta + r \beta^2)$ head predicates each and $N$ piecewise polynomial activation functions with $\ordo(P \Omega (r^2 \beta + r \beta^2))$ head predicates each, the total number of head predicates is $\ordo((N P \Omega + L)(r^2 \beta + r \beta^2))$ (just like with size, this dwarfs the number of head predicates introduced by comparisons, shifting and the one-hot counter).
\end{proof}

\subsection{From BNL to NN}\label{BNL to NN}

When translating $\BNL$-programs to neural networks, the choice of activation function can greatly influence the similarity by which equivalence is obtained. We consider the homogeneous case, i.e., the case where each node has the same activation function, and provide a general axiomatization for activation functions and the corresponding similarity for which we obtain equivalence. 
The axiomatization includes two alternative cases: the second one one utilizes the properties of floating-point arithmetic, rounding and saturation in particular, while the first one does not.

Intuitively, the axiomatization states that given an activation function, there are three points on the graph of the activation function which do not belong to the same line and are at a suitable distance from each other. In more detail, there are three distinct points $(x_1, y_1)$, $(x_2, y_2)$, $(x_3, y_3)$ on the graph of the activation function such that $(x_1, y_1)$ and $(x_2, y_2)$ are on the same horizontal line and $(x_3, y_3)$ is equally far from $(x_2, y_2)$ along the $x$-axis as $(x_1, y_1)$ (but in the opposite direction) and not on the same horizontal line as the two other points, i.e., $x_3 \neq x_1$, $\abs{x_2 - x_3} = \abs{x_2 - x_1}$, $y_1 = y_2$ and $y_1 \neq y_3$.  In order to simulate a $\BNL$-program with a neural network by using arithmetic operations we have to assume that some numbers w.r.t. $x_1, x_2, x_3, y_1, y_2, y_3$ are accurately representable in the floating-point format that is used for the corresponding neural network. 
Moreover, in the case where the activation function increases or decreases fast enough w.r.t. the maximum (or minimum) value of the floating-point format, we obtain even simpler requirements which take advantage of the properties of floating-point formats.

The axiomatization covers a wide class of approximations of activation functions (see \cite{activation_function_survey} for an overview on activation functions), including the following: \emph{linear functions} $y = kx$, the \emph{rectified linear unit} $\mathrm{ReLU}(x) = \max\{0, x\}$ and many variations thereof, the \emph{Heaviside step function} $H(x) = 1$ if $x > 0$ and $H(x) = 0$ otherwise, the \emph{sigmoid} function $s(x) = \frac{1}{1 + e^{-x}}$, the \emph{hyperbolic tangent} $\tanh(x) = \frac{e^{x} - e^{-x}}{e^{x} + e^{-x}}$, and the \emph{Gaussian function} $g(x) = e^{-x^{2}}$.

Now, we begin to formally define the axiomatization.
Let $\fF$ be an arbitrary floating-point format or normalized float format and let $\alpha \colon \fF \to \fF$ be a piecewise polynomial activation function. 
We also define a corresponding canonical similarity $\sigma$ from $\{0,1\}$ to $\fF$ for the axiomatization.
Below in the proof of Theorem \ref{BNL_to_NN}, we identify each floating-point number in $\fF$ with the rational number that it represents.
We also let $F_{\max}$ and $F_{\min}$ respectively denote the greatest and least element in $\fF$.

The axiomatization is given below, where the sum, subtraction, multiplication and division are the precise operations over the field of real numbers. There are two possible cases, where the former places more restrictions on the floating-point format than the latter.
\begin{equation}\label{A}\tag{A}
    \parbox{0.85\textwidth}{%
    There exist $x_0,y_0,d,h \in \R$, $d, h \neq 0$ such that $x_0$ and $x_0+d$ are accurately representable in $\fF$, $\alpha(x_0)$ accurately represents $y_0-h$, $\alpha(x_0+d)$ accurately represents $y_0+h$ and one of the following holds:
    \smallskip
    
    \begin{enumerate}
        \item The following real numbers are also accurately representable in $\fF$: $x_0-d$, $\frac{d}{h}$, $\frac{d}{2h}$, $\frac{d}{h}y_0$, $x_0 + \frac{d}{h}y_0$, $x_0 - \frac{d}{h}y_0$, $\frac{d}{h}y_0 + d$, $\frac{d}{h}y_0 - d$, $\frac{d}{2h}y_0 + \frac{d}{2}$ and $\frac{d}{2h}y_0 - \frac{d}{2}$. 
        Moreover, we require that $\alpha(x_0-d) = \alpha(x_0)$.
        \item $d \in \{F_{\max}, F_{\min}\}$, $y_0 = -h$ and $\abs{y_0} \geq \frac{1}{2}$.
    \end{enumerate}
    }
\end{equation}
The chosen similarity is the canonical similarity induced by $\sigma(0) = \alpha(x_0)$ and $\sigma(1) = \alpha(x_0+d)$ as specified in Section \ref{subsection:equivalence_of_NNs}. 
We say that an activation function $\alpha \colon \fF \to \fF$ is \textbf{regular} if it satisfies Axiomatization \ref{A}.

Case 1 of the axiomatization covers, e.g., the rectified linear unit $\mathrm{ReLU}$ (and biased variations thereof) and the Heaviside step function.
Note that $\mathrm{ReLU}$ and Heaviside satisfy the axioms with any floating-point format $\fF$, as we can simply choose $x_0 = 0$ and $d = 1$, in which case $y_0 = h = \frac{1}{2}$ and 
all of the other values in Case 1 of \ref{A} are in the set $\{-1, 0, 1, 2\}$.
Case 2 of the axiomatization implies that $\abs{\alpha(x_0)} \geq 1$ and $\alpha(x_0+d) = 0$, and it covers e.g., any function of the type $\alpha(x) = kx$ where $k \in \fF$ and $k \neq 0$ as we can choose $x_0 = F_{\max}$ and $d = F_{\min}$ (with a large enough floating-point format $\fF$ to ensure that $\abs{\alpha(F_{\max})} \geq 1$). Case~2 also covers many variations of $\mathrm{ReLU}$ and piecewise polynomial approximations of the Gaussian function (in the case of the Gaussian function we simply choose $x_0 = 0$). Moreover, due to underflow Case 1 of the axiomatization also covers floating-point approximations of some linear functions $y = ax + b$, where $a$ and $b$ are small enough w.r.t. the floating-point format. For example, consider $\cF(1,1,10)$ and let $f(x) = \frac{x}{10} + \frac{2}{10}$, $x_0 = \frac{5}{10}$, $d = -\frac{1}{10}$, $y_0 = \frac{25}{100}$ and $h = -\frac{5}{100}$. Now, we have $f(x_0 - d) = f(x_0) = y_0 - h = \frac{3}{10}$ and $f(x_0 + d) = y_0 + h = \frac{2}{10}$ as wanted and all the numbers required in the axiomatization are in the set $\{0,\frac{2}{10}, \frac{3}{10}, \frac{4}{10}, \frac{5}{10}, \frac{6}{10}, 1, 2\} \subseteq \fF$. Thus, the function $f$ satisfies the axiomatization (note also that $y_0 \notin \fF$, but this is allowed as the axiomatization only requires that $y_0-h, y_0+h \in \fF$).

By utilizing properties of floating-point arithmetic, Case 1 of the axiomatization also covers such piecewise polynomial approximations $\alpha$ of \textbf{1)} the sigmoid function and the Gaussian function 
that output $0$ with a small enough floating-point number and \textbf{2)} the hyperbolic tangent that output $-1$ with a small enough floating-point number. This is a justified assumption because the values of these functions approach $0$ (resp. $-1$) very quickly when moving toward negative infinity, e.g. the sigmoid function $s(x) = \frac{1}{1 + e^{-x}}$ decreases exponentially toward $0$ when $x$ approaches $- \infty$.
This means that we can choose $x_0$ to be any (negative) floating-point number such that $s(x_0)$ rounds to $0$ in $\fF$ and $2x_0 \in \fF$ and then choose $d = -x_0$. This fixes $y_0 = h = \frac{1}{4}$ and we have $\alpha(x_0 - d) = \alpha(2x_0) = \alpha(x_0) = 0$ and $\alpha(x_0+d) = \alpha(0) = \frac{1}{2}$. This idea works also for the Gaussian function $g$ by replacing $s$ with $g$ but notice that $y_0 = h = \frac{1}{2}$ and $\alpha(0) = 1$. The same idea also works for the hyperbolic tangent, but there we choose $x_0$ to be any floating-point number such that $\tanh(x_0)$ rounds to $-1$ in $\fF$, in which case we get $y_0 = -\frac{1}{2}$, $h = \frac{1}{2}$ and $\alpha(0) = 0$.

We are now ready to translate $\BNL$-programs into equivalent neural networks.
In the below theorem, the equivalent neural network for $\fF$ is constructed with respect to the corresponding canonical similarity associated with Axiomatization \ref{A}.
The weights and biases of the network are defined slightly differently depending on which case of \ref{A} is satisfied.
\begin{theorem}\label{BNL_to_NN}
    For any regular activation function $\alpha$, a $\BNL$-program of size $s$ and depth $d$ translates to an equivalent general neural network with $\ordo(s)$ nodes, computation delay $\ordo(d)$ and $\alpha$ at each node.
\end{theorem}
\begin{proof}
    We consider the case where the $\BNL$-program $\Lambda$ includes attention predicates, since the case where the output rounds are given by an attention function is trivial to obtain by Remark \ref{remark: outputs and attention}.
    We first use Theorem \ref{BNL_to_FBNL} to translate $\Lambda$ into an equivalent{} fully-open $\BNL$-program $\Lambda'$ where $\Lambda'$ has size $\ordo(s)$ and computation delay $\ordo(d)$. 
    Next, we can construct for any regular activation function $\alpha \colon \fF \to \fF$ a general neural network $\cN$ over $\fF$ that is strongly equivalent to $\Lambda'$ (w.r.t. the canonical similarity $\sigma$ from $\{0,1\}$ to $\fF$ associated with Axiomatization \ref{A}) as follows. We create a node $v_{X}$ for each head predicate $X$ of $\Lambda'$ (hence we have $\ordo(s)$ nodes), and an edge from node $v_{X}$ to $v_{Y}$ iff $X$ appears in the induction rule of $Y$. Note that the in-degree of each node is thus at most~$2$.
    The input/output/attention nodes of $\cN$ are the nodes $v_{X}$, where $X$ is an input/print/attention predicate of $\Lambda'$.
    The nodes are ordered according to the ordering of the head predicates.
    We define the bias of $v_{X}$ and the weights of its incoming edges depending on the induction rule $\psi$ of $X$.
    Since schemata of the form $\top \land Y$ and $Y \land \top$ are logically equivalent to $Y$ and since $\top \land \top$ is logically equivalent to $\top$, we do not consider schemata of these types separately.
    Recall that $\alpha(x_0-d) = \alpha(x_0) = y_0-h$, $\alpha(x_0+d) = y_0+h$, $\sigma(0) = \alpha(x_0)$ and $\sigma(1) = \alpha(x_0+d)$.
    If $\psi = \top$, then the bias is $x_0 + d$ and if $\psi = \neg \top$, then the bias is~$x_0$. 
    We proceed to define the remaining cases depending on how $\alpha$ satisfies~\ref{A}.

    First, we consider the case where $\alpha$ satisfies Case 1 of \ref{A}. If $\psi = Y$, then the bias is $x_0 - \frac{d}{h}y_0$ and the weight is $\frac{d}{h}$. If $\psi = \neg Y$, then the bias is $x_0 + \frac{d}{h}y_0$ and the weight is $-\frac{d}{h}$. Finally, if $\psi = Y \land Z$, then the bias is $x_0 - \frac{d}{h}y_0$ and the weights are $\frac{d}{2h}$.

    We go over the calculations in the case of negation, and leave the remaining cases to the reader.
    For all $t \in \N$, we calculate the activation value $g_{t+1}(v_X)$ based on previous activation values $g_t$, where $g_t$ is the state of the network at time $t$ as defined in Section \ref{section: general neural networks}.
    If $\psi = \neg Y$, then $g_{t+1}(v_X) = \alpha((x_0 + \frac{d}{h}y_0) + g_{t}(v_Y)(-\frac{d}{h}))$ where
    \begin{itemize}
        \item if $g_{t}(v_Y) = y_0+h = \sigma(1)$, we get $\alpha((x_0 + \frac{d}{h}y_0) + (-\frac{d}{h}y_0 - d)) = \alpha(x_0 - d) = \sigma(0)$, and
        \item if $g_{t}(v_Y) = y_0-h = \sigma(0)$, we get $\alpha((x_0 + \frac{d}{h}y_0) + (-\frac{d}{h}y_0 + d)) = \alpha(x_0 + d) = \sigma(1)$.
    \end{itemize}

    Second, we consider the case where $\alpha$ satisfies Case 2 of \ref{A}.
    We consider the case where $y_0 \geq \frac{1}{2}$, since the case where $y_0 \leq -\frac{1}{2}$ is obtained simply by switching the signs of the weights. If $\psi = Y$, then the bias is $x_0+d$ and the weight is $-d$. If $\psi = \neg Y$, then the bias is $x_0$ and the weight is $d$. 
    Finally, if $\psi = Y \land Z$, then the bias is $x_0 + d$ and the weights are $-d$. 
    We demonstrate floating-point calculations with $d$ in the case where $\psi = Y \land Z$ and both $v_Y$ and $v_Z$ have the activation value $\alpha(x_0) = \sigma(0)$; then the value of $v_X$ becomes $\alpha((x_0+d) + (\alpha(x_0) \cdot (-d) + \alpha(x_0) \cdot (-d))) = \alpha((x_0+d) + (-d -d)) = \alpha((x_0+d) - d) = \alpha(x_0) = \sigma(0)$, where the first equality holds because $\alpha(x_0) = y_0-h = 2y_0 \geq 1$ and the second because $-d - d = -d$.
\end{proof}

\subsection{Fixed parameter results and other consequences}\label{section:remarks}

From Theorems \ref{NN_to_BNL} and \ref{BNL_to_NN}, we obtain the following corollary in the scenario where the parameters of the floating-point format and the activation functions are assumed to be constant (the latter direction is unchanged from Theorem \ref{BNL_to_NN}).
\begin{theorem}\label{theorem:NN_to_BNL_fixed}
    Given a fixed floating-point format $\cF$ and fixed parameters for the activation functions, the following hold.
    \begin{itemize}
        \item A general neural network for $\cF$ (resp. for $\cF_{\mathsf{n}}$) with $N$ nodes, 
        $L$ edges 
        and degree $\Delta$ translates to an equivalent{} $\BNL$-program of size 
        $\ordo(N + L)$ 
        and computation delay $\ordo(\log(\Delta))$.
        \item For any regular activation function $\alpha$ over $\cF$ (resp. over $\cF_{\mathsf{n}}$), a $\BNL$-program of size $s$ and depth $d$ translates to an equivalent general neural network for $\cF$ (resp. for $\cF_{\mathsf{n}}$) with $\ordo(s)$ nodes, computation delay $\ordo(d)$ and $\alpha$ at each~node.
    \end{itemize}
\end{theorem}

We note that it is possible to obtain different space and time complexities for Theorem \ref{NN_to_BNL} depending on the implementation of the arithmetic operations.
For example, in circuit complexity there are various ways to implement integer addition, multiplication and comparison with some size-depth tradeoffs in the literature, e.g., the Kogge-Stone adder \cite{Kogge-stone}.
These circuit complexity results for arithmetic integer operations can be used by taking the corresponding circuit for an arithmetic operation, transforming it into a corresponding self-feeding circuit and then using Theorem \ref{thrm:SF_circuit_to_BNL} to obtain the wanted program. 
Next we present a corollary in the case where the subprograms for integer addition, comparison and multiplication are arbitrary.

Let $\cF(p, q, \beta)$ be a floating-point format and let $r = \max\{p,q\}$.
Assume that we can simulate integer operations in $\cZ(r, \beta)$ in $\BNL$ as follows: \textbf{1)} integer comparison in size $s_C$ and computation time $t_C$ (and $h_C$ head predicates), \textbf{2)} integer addition in space $s_A$ and computation time $t_A$ (and $h_A$ head predicates), and \textbf{3)} integer multiplication in space $s_M$ and computation time $t_M$ (and $h_M$ head predicates). Then we obtain the following corollary.
\begin{corollary}\label{corollary:black_box}
    Let $\cN$ be a general neural network for $\cF$ (or for $\cF_{\mathsf{n}}$) 
    with $N$ nodes, 
    $L$ edges, 
    degree~$\Delta$ and piece-size $P$. Let $\Omega$ be the order of $\cN$
    (setting $\Omega := 2$ if the order is $1$ or $0$).
    We can construct an equivalent{} $\BNL$-program $\Lambda$ such that
    \begin{itemize}
        \item the size of $\Lambda$ is 
        $\ordo((NP\Omega + L)(s_M + s_A + s_C + r^2 \beta))$,
        \item there are 
        $\ordo((N P \Omega + L) (h_M + h_A + h_C))$
        head predicates in $\Lambda$, and
        \item the computation delay of $\Lambda$ is $\ordo(\log(\Omega)(t_M + t_A + t_C) + \log(\Delta)(t_A + t_C))$.
    \end{itemize}
\end{corollary}
\begin{proof}
    We analyze the size and time delay in the $\BNL$-implementations of shifting, normalization, rounding, floating-point addition, floating-point multiplication and piecewise polynomial floating-point functions. The proofs below are essentially the same as the proofs in previous sections, except that we replace the subprograms for integer operations with the subprograms above.
    
    Shifting to the right involves integer addition and auxiliary rules of size $\ordo(r^2 \beta)$, and thus has size $\ordo(s_A + r^2 \beta)$ and computation time $\ordo(t_A)$ (and $\ordo(h_A)$ head predicates). Shifting to the right by one reduces the size to $\ordo(s_A)$ while the computation time (and number of head predicates) stays the same. Shifting to the left is otherwise the same as shifting to the right, but it also involves integer comparison, so it has size $\ordo(s_A + s_C + r^2 \beta)$ and computation time $\ordo(t_A + t_C)$ (and $\ordo(h_A + h_C)$ head predicates). Thus, shifting to the right dwarfs integer addition and shifting to the left dwarfs shifting to the right.

    Normalization involves integer comparison, integer addition and shifting to the left, so it has the same complexities as shifting to the left: size $\ordo(s_A + s_C + r^2 \beta)$ and computation time $\ordo(t_A + t_C)$ (and $\ordo(h_A + h_C)$ head predicates).

    Rounding involves integer addition, integer comparison and shifting to the right by one, so it has size $\ordo(s_A + s_C)$ and computation time $\ordo(t_A + t_C)$ (and $\ordo(h_A + h_C)$ head predicates). Thus, rounding is dwarfed by normalization.

    Floating-point addition involves integer comparison, integer addition, shifting to the right, normalization and rounding. Thus, it has the same complexities as normalization: size $\ordo(s_A + s_C + r^2 \beta)$ and computation time $\ordo(t_A + t_C)$ (and $\ordo(h_A + h_C)$ head predicates).

    Floating-point multiplication involves integer addition, integer multiplication, rounding and normalization. Thus, it has size $\ordo(s_M + s_A + s_C + r^2 \beta)$ and computation time $\ordo(t_M + t_A + t_C)$ (and $\ordo(h_M + h_A + h_C)$ head predicates). Thus, floating-point multiplication dwarfs floating-point addition.

    Simulating a piecewise polynomial floating-point function with order $\Omega$ and $P$ pieces involves $\ordo(P)$ integer comparisons, $\ordo(\Omega)$ floating-point multiplications in a binary tree of height $\ordo(\log(\Omega))$, $\ordo(P \Omega)$ floating-point multiplications in parallel and $\ordo(P\Omega)$ floating-point additions in a binary tree of height $\ordo(\log(\Omega))$. Thus, the size of the program is $\ordo(P \Omega (s_M + s_A + s_C + r^2 \beta))$ and the computation time is $\ordo(\log(\Omega)(t_M + t_A + t_C))$ (and the number of head predicates is $\ordo(P \Omega (h_M + h_A + h_C))$).

    Finally, a single iteration of a general neural network involves 
    $L$
    floating-point multiplications in parallel, 
    $L$
    floating-point additions in a binary tree of height $\log(\Delta)$ and $N$ piecewise polynomial floating-point functions in parallel. Thus, the size is 
    $\ordo((N P \Omega + L) (s_M + s_A + s_C + r^2 \beta))$
    and the delay is $\ordo(\log(\Omega)(t_M + t_A + t_C) + \log(\Delta) (t_A + t_C))$ (and 
    $\ordo((N P \Omega + L) (h_M + h_A + h_C))$
    head predicates).
\end{proof}

While we considered neural networks with saturating floating-point formats, there are several other floating-point formats used in practical applications that could be studied, e.g. fixed-point arithmetic. 
As mentioned before, our floating-point formats do not include special symbols, e.g., $\infty$, $-\infty$ or $\mathrm{NaN}$ (Not a Number). However, we believe our translations in both directions generalize (with some modifications) for formats that do include them even when cases of overflow map to $\infty$ or $- \infty$ although with the apparent exception of activation functions that only satisfy Case 2 of Axiomatization \ref{A} in Theorem \ref{BNL_to_NN}.

As a concluding note for this section, the match between $\BNL$ and neural networks provides a 
concrete demonstration of the obvious fact that---in some relevant 
sense---there is no difference
between symbolic and non-symbolic approaches. Under reasonable background assumptions, non-symbolic approaches can be technically reduced to symbolic ones. More than to the differences between the symbolic and non-symbolic realms, the clear advantages of modern AI methods relate to the difference between systems based on explicit programming and systems that involve an aspect of learning not based on explicit and fully controlled programming steps.

\section{Results for feedforward neural networks}\label{sec: fnn}

While we have so far studied general neural networks, 
in many practical applications the neural networks that are used are feedforward \cite{rumelhart1986learning}.
A feedforward network contains no loops and the nodes are instead split into multiple distinct successive \emph{layers}. Each node calculates its activation value as before, but only based on nodes in previous layers. In this section, we provide definitions for feedforward neural networks and analogous definitions for feedforward $\BNL$-programs. 
We then show that the translations of Section \ref{descsection} apply for the feedforward framework with the same size complexities and delays. We also provide translations for the synchronized feedforward framework, where the connections between layers cannot skip layers, with only an additional quadratic size increase in one direction or the other. 
We cover both cases for the sake of generality, since it is not always clear in the literature whether a neural network framework allows asynchronicity.
The characterization is given in Theorems~\ref{theorem:FFNN_to_FFBNL} and~\ref{theorem:FFBNL_to_FFNN}.

During this section, we do not consider programs or neural networks with attention predicates or attention nodes respectively, but instead use attention functions. 
The reason for this is because feedforward networks in practice only output once, and the output round is simply determined by the number of layers in the network.
For the same reason, instead of computation delay we pay special attention to the number of layers resulting from our translations.

Before we define feedforward neural networks and programs formally, we define the skeleton of a $\BNL$-program or neural network which essentially tells its topology. Let $\Lambda$ be a $\cT$-program of $\BNL$. The \textbf{skeleton} of $\Lambda$ is a graph $(\cT, E_{\Lambda})$, where $E_{\Lambda} = \{\, (X, Y) \in \cT^2 \mid \text{$X$ appears in the induction rule of $Y$} \,\}$.
Next, let $\cN = ((V, E, <^V), \fa, \fb, \fw, \pi)$ be a general neural network. The \textbf{skeleton} of $\cN$ is the graph $(V, E)$.
We say that $\Lambda$ is an \textbf{$n$-layer feedforward $\BNL$-program} (resp. $\cN$ is an \textbf{$n$-layer feedforward neural network}) if 
its skeleton 
is a directed acyclic graph of depth $n-1$. 
Moreover, the input predicates (resp. the input nodes) are the nodes in the skeleton with in-degree $0$, and the print predicates (resp. the output nodes) are the nodes in the skeleton with out-degree $0$. 
Finally, $\Lambda$ (resp., $\cN$) is always associated with an attention function that maps each input to $n-1$.
The \textbf{layer} of each predicate $X \in \cT$ (resp. the layer of each node $v \in V$) is its height in the skeleton. 
We say that $\Lambda$ (resp. $\cN$) is \textbf{synchronous} if its skeleton also satisfies the following conditions:
\begin{enumerate}
    \item the nodes of the skeleton with out-degree zero have height $n-1$ and
    \item for each node $v$ of the skeleton, the in-neighbours of $v$ have height $\mathsf{height}(v)-1$. 
\end{enumerate}

While feedforward programs (resp. feedforward neural networks) are syntactically also programs of $\BNL$ (resp. general neural networks) in the ordinary sense, we provide them with a different semantics that is similar to the previously defined semantics but with the following distinction.
Informally, the truth value of a head predicate in the program (resp. the activation value of a node in the network) stops updating after the round $h$, where $h$ is the height of the predicate (resp. node) in the skeleton. 
These semantics are analogous to the semantics of ordinary Boolean circuits.
More formally, let $X$ be a predicate in the feedforward program (resp. let $v$ be a node in the feedforward neural network) and let $(g_t)_{t \in \N}$ be the sequence of global configuration functions defined in Section \ref{subsection: notions_SC_BNL} (resp. let $(g_t)_{t \in \N}$ be the sequence of states of the network defined in Section \ref{section: general neural networks}); we modify each $g_t$ by defining that $g_{t}(X) = g_h(X)$ (resp. $g_{t}(v) = g_h(v)$) for all $t \geq h$, where $h$ is the height of $X$ (resp. $v$) in the skeleton.

Given a feedforward $\BNL$-program, the \textbf{schema depth} of a given layer $i$ is the maximum depth of an induction rule of a head predicate in that layer.
We say that a $\BNL$-schema $\theta$ is \textbf{balanced} if for each subschema of $\theta$ of the form $\varphi \land \psi$, the schemata $\varphi$ and $\psi$ have the same depth. 
Analogously, we call a feedforward $\BNL$-program \textbf{balanced} if \textbf{1)}~each body of the induction rules is balanced, and \textbf{2)}~for each layer, the induction rules of head predicates in that layer all have the same depth. Lastly, a feedforward $\BNL$-program is \textbf{pure} if the only occurrences of $\top$ are in the base rules of non-input predicates or the induction rules of input predicates.

Given the above definitions, the output of a feedforward neural network is invariant w.r.t. the function $\pi$ giving initial values to non-input nodes as well as the biases and activation functions of input nodes. Likewise, the output of a feedforward program is invariant w.r.t. the base rules of non-input predicates and the induction rules of input predicates.

It is trivial to turn a feedforward program into a pure feedforward program in linear size:
\begin{lemma}\label{lemma:FFBNL_to_pureFFBNL}
    For each feedforward $\BNL$-program $\Lambda$ of size $s$, we can construct a strongly equivalent pure feedforward $\BNL$-program $\Lambda'$ where $\Lambda'$ has size $\ordo(s)$ and the same number of layers as $\Lambda$. If $\Lambda$ is balanced (resp. synchronous), then so is $\Lambda'$.
\end{lemma}
The proof centers around replacing instances of $\top$ in the induction rules with logically equivalent schemata (e.g. $\neg (X \land \neg X)$, where $X$ is a predicate from a previous layer). Synchronicity is trivially preserved. Making the resulting schemata and program balanced involves similar techniques as the proof of the next lemma.

We show that feedforward $\BNL$-programs translate to equivalent balanced synchronous feedforward programs with only a quadratic size blow-up.
\begin{lemma}\label{unbalanced_BNL_to_balanced_BNL}
    Given a pure unbalanced feedforward $\BNL$-program of size $s$ and depth $d$, we can construct an equivalent{} pure balanced synchronous feedforward $\BNL$-program of size $\ordo(s^2)$ and depth $\ordo(d)$ with the same number of layers. 
\end{lemma}
\begin{proof}
    First, we balance the feedforward $\BNL$-program in two steps. We first replace the body of each induction rule in layer $0$ with $\top$. Let $\varphi \land \psi$ be a $\BNL$-schema in the program in the layer $\ell \geq 1$ such that the depth of $\varphi$ is greater than that of $\psi$. Let $Y$ be a head predicate in the layer $\ell-1$, and define 
    $\tau \colonequals \neg((Y \land Y) \land \neg Y)$
    (i.e. $\tau$ is a balanced tautology with depth $3$). 
    Let $d \in \N$ be the difference between the depths of $\psi$ and $\varphi$.
    Let $k$ be the depth of $\psi$; we assume that $k \geq 3$, as we can always modify $\varphi$ and $\psi$ by adding $2$ or $4$ negations in front of them.
    We define a schema $\gamma(\psi, d)$ that is logically equivalent to $\psi$ and adds $d$ to the depth of $\psi$ as follows (also, $\gamma(\psi, d)$ is balanced if $\psi$ is balanced).
    We start by defining some auxiliary notations.
    Given $i \in \N$, we let $(\neg)^i$ denote $i$ nested negations.
    Moreover, given $i \in \N$, $i \geq 3$, we let $\chi_i$ denote a balanced schema $\chi_i$ of depth $i$ that is a tautology:
    let $\chi_i$ denote $(\neg)^{i-3} \tau$ if $i$ is odd and $(\neg)^{i-4} (\tau \land \tau)$ if $i$ is even.
    Now, we can define $\gamma(\psi, d)$. If $d$ is even, then $\gamma(\psi, d) \colonequals (\neg)^d \psi$. If $d$ is odd, then $\gamma(\psi, d) \colonequals (\neg)^{d-1} ( \chi_k \land \psi)$. 
    We replace each subschema $\varphi \land \psi$ by the logically equivalent schema $\varphi \land \gamma(\psi, d)$ (this is done recursively starting with schemata of depth $0$, then $1$, etc.). 
    Next, we balance the depths of the bodies of rules in each layer by using a similar method as follows. Let $d_{\theta}$ be the depth of the body $\theta$ of a rule and let $D$ be the depth of the program. If $D-d_{\theta}$ is non-zero, then $\theta$ is replaced by $\gamma(\theta, D-d_{\theta})$.

    Lastly, we synchronize the program by adding missing head predicates to each layer. Let $X$ be a head predicate which appears in the layer $n$ and let $n' = n + m$ be the maximal layer where $m \in \Z_+$, $m \geq 2$, such that $X$ appears in an induction rule in that layer. 
    In the case where $X$ is a print predicate, we instead set that $n' = L$, where $L$ is the number of layers in the program.
    For each $i \in [m-1]$, let $d_i$ denote the depth of the layer $n+i$; we can assume that $d_i \geq 4$ as we can always modify the induction rules such that this is the case with only a linear size increase. Now, 
    we let $X_{0}$ denote $X$, and
    for each $i \in [m-1]$ we define a head predicate $X_i$ with the base rule $X_i (0) \colonminus \bot$ and the induction rule 
    $X_i \colonminus \gamma((\neg)^4 X_{i-1}, d_i-4)$ 
    (the nested negations $(\neg)^4$ are added as $\gamma(\psi, d)$ requires that the depth of $\psi$ is at least $3$). 
    Finally, for each head predicate $Z$ in a layer $n+i$ where $i \in [m]$, $i \geq 2$, we replace each instance of $X$ with $X_{i-1}$.
    If $X$ was a print predicate, then $X_{L-1-n}$ in layer $L-1$ replaces its status as print predicate.

    Clearly, the obtained program is balanced, synchronous and equivalent{} to the original program.
    We analyze the size of the obtained program. 
    In the first step of balancing, where we balance each subschema of the form $\varphi \land \psi$, there may be $\ordo(s)$ subschemata that have to be balanced. Each instance of balancing increases the size by $\ordo(s)$. Thus, after the first step the size of the program is $\ordo(s^2)$. In the second step, the size of each rule increases by at most $\ordo(s)$ which also adds $\ordo(s^2)$ to the size. 
    When we synchronize the program, there are at most $s$ variables that have to be synchronized, and each instance of synchronization increases the size by $\ordo(s)$.
    Since all the steps are independent from each other, the size of the obtained program is $\ordo(s^2)$. 
    It is easy to see that the depth of the obtained program is $\ordo(d)$.
\end{proof}

We are now ready to translate feedforward neural networks into feedforward $\BNL$-programs.
We also analyze the size differences depending on whether the programs are balanced or not, and whether the feedforward neural network is synchronous or not.
Let $\cN$ be any feedforward neural network $\cN$ for any floating-point format $\cF(p,q,\beta)$ (resp., $\cF_{\mathsf{n}}$) with $N$ nodes,
$L$ edges,
degree $\Delta$ and piece-size $P$. Let $\Omega$ be the order of $\cN$ (setting $\Omega := 2$ if the order is $1$ or $0$) and let $r = \max\{p,q\}$. Let $K_{\cN}$ and $T_{\cN}$ denote the size and delay from Theorem \ref{NN_to_BNL}, i.e. we let
$K_{\cN} \colonequals (N P \Omega + L) (r^4 + r^3\beta + r^{2}\beta^2 + r\beta^3)$ 
and $T_{\cN} \colonequals \log(\Omega)(\log(r) + \log(\beta)) + \log(\Delta)$.
Note that in the theorem below, equivalence{} is defined w.r.t. the canonical similarity given in Section \ref{subsection:equivalence_of_NNs}.
\begin{theorem}\label{theorem:FFNN_to_FFBNL}
    Let $\cN$ be a (possibly heterogeneous) $n$-layer feedforward neural network.
    We can construct the following $\BNL$-programs that are equivalent{} to $\cN$ and have $\ordo(nT_{\cN})$ layers:
    \begin{itemize}
        \item an unbalanced feedforward $\BNL$-program of size $\ordo(K_{\cN})$ that is also synchronous if $\cN$ is synchronous,
        \item a balanced synchronous feedforward $\BNL$-program of size $\ordo(K_{\cN}^2)$.
    \end{itemize}
\end{theorem}
\begin{proof}
    We alter the construction of Theorem \ref{NN_to_BNL} to make the resulting program feedforward.
    First, as the neural network is feedforward, it is enough to include a program for each addition and multiplication only once. Furthermore, subprograms that calculate additions can be made feedforward by simply constructing as many copies of each head predicate as there are iteration steps in the addition; because the number of steps is constant, this does not affect the size complexity. 
    As for multiplications, there are at most 
    $(N P \Omega + L)$ 
    multiplications with $\ordo(r^{2} \beta + r \beta^{2})$ head predicates each, which have to be slowed down by $\ordo(\log(r) + \log(\beta))$ (the computation time of floating-point multiplication); this adds up to 
    $\ordo((N P \Omega + L)(r^{2} \beta + r \beta^{2})(\log(r) + \log(\beta)))$ 
    which is dwarfed by $K_{\cN}$ and thus does not affect size.
    Now if $\cN$ was synchronous, then clearly so is $\Lambda$.
    If we want to obtain a balanced synchronous $\BNL$-program, we can simply purify it with Lemma \ref{lemma:FFBNL_to_pureFFBNL} and then apply Lemma \ref{unbalanced_BNL_to_balanced_BNL}; in that case, the size of the program is $\ordo(K_{\cN}^2)$.
    None of the described steps cause any additional increase to the computation delay $T_{\cN}$, and so the number of layers in $\Lambda$ is obtained by multiplying $T_{\cN}$ by the number of layers in $\cN$.
\end{proof}

Next, we show that each feedforward program of $\BNL$ can be turned into a linear-size fully-open one by increasing the number of layers.

\begin{lemma}\label{FFBNL_to_fully-openFFBNL}
    For each pure $n$-layer feedforward $\BNL$-program $\Lambda$ of size $s$ where the schema depth of each layer $i$ is $d_{i}$, we can construct an equivalent{} pure fully-open feedforward $\BNL$-program $\Lambda'$ with $\ordo(\sum_{i = 1}^{n-1} d_{i})$ layers and size $\ordo(s)$.
    Moreover, if $\Lambda$ is balanced and synchronous, then $\Lambda'$ is synchronous.
\end{lemma}
\begin{proof}
    Let $\Lambda$ be a pure $n$-layer feedforward $\BNL$-program. We construct a fully-open pure feedforward $\BNL$-program $\Lambda'$ as follows.
    We construct a head predicate $X_{\psi}$ for each schema $\psi$ in $\Lambda$ with the exception of subschemata that only appear in the induction rules of input predicates of~$\Lambda$. For each schema $\neg \psi$ we construct the induction rule $X_{\neg \psi} \colonminus \neg X_{\psi}$, for each schema $\psi \land \theta$ we construct the induction rule $X_{\psi \land \theta} \colonminus X_{\psi} \land X_{\theta}$, and for each head predicate $Y$ of $\Lambda$ with the induction rule $Y \colonminus \psi$ we construct the induction rule $X_{Y} \colonminus X_{\psi}$ with the exception of input predicates of $\Lambda$ which receive the rule $X_{Y} \colonminus \top$. If $Y$ is an input/print predicate of $\Lambda$, then $X_Y$ is an input/print predicate of~$\Lambda'$.
    The base rules of non-input predicates can be chosen arbitrarily.
    It is easy to see that $\Lambda'$ is equivalent{} to $\Lambda$. Now if $\Lambda$ is balanced and synchronous, then $\Lambda'$ has $d_i$ layers for each layer $i$ of $\Lambda$, each of which computes subschemata of a specific depth, and the edges do not skip layers because $\Lambda$ is balanced and synchronous.
\end{proof}

Now, we are ready to translate feedforward $\BNL$-programs into equivalent feedforward neural networks. We again analyze how the size is affected by whether the program is balanced and synchronous and whether the network is synchronous.
Let $\fF$ be an arbitrary floating-point format or normalized float format.
Analogously to general neural networks, in the below theorem, the equivalent feedforward neural network for $\fF$ that uses a regular activation function is then constructed with respect to the canonical similarity given by Axiomatization \ref{A}.

\begin{theorem}\label{theorem:FFBNL_to_FFNN}
    Let $\alpha$ be any regular activation function and
    let $\Lambda$ be an $n$-layer feedforward $\BNL$-program of size $s$, where the schema depth of each layer $i$ is $d_{i}$. 
    We can construct the following equivalent{}
    $\ordo(\sum_{i = 1}^{n-1} d_{i})$-layer feedforward neural networks that use $\alpha$ (at each node):
    \begin{itemize}
        \item a feedforward neural network with $\ordo(s)$ nodes that is also synchronous if $\Lambda$ is balanced and synchronous.
        \item a synchronous feedforward neural network with $\ordo(s^{2})$ nodes.
    \end{itemize}
\end{theorem}
\begin{proof}
    We start by purifying $\Lambda$ using Lemma \ref{lemma:FFBNL_to_pureFFBNL} which does not blow up the size.
    To obtain a neural network with $\ordo(s)$ nodes,
    we first apply Lemma \ref{FFBNL_to_fully-openFFBNL} to turn purified $\Lambda$ into a pure fully-open feedforward program whose size is linear in the size of $\Lambda$. Then by using Theorem \ref{BNL_to_NN} we turn the obtained fully-open program into an equivalent{} feedforward neural network (skipping the application of Theorem \ref{BNL_to_FBNL}) whose size is linear in the size of the fully-open program (and by extension linear in the size of $\Lambda$). Notice that if $\Lambda$ is balanced and synchronous, then the resulting fully-open program is also pure and synchronous, and in that case the resulting feedforward neural network is also clearly synchronous.

    To obtain a synchronous neural network with $\ordo(s^2)$ nodes, we simply apply Lemma \ref{unbalanced_BNL_to_balanced_BNL} to make purified $\Lambda$ balanced and synchronous, causing a quadratic size increase, and apply Lemma \ref{FFBNL_to_fully-openFFBNL} and Theorem \ref{BNL_to_NN} as before.
\end{proof}

Analogously to Theorem \ref{theorem:NN_to_BNL_fixed}, we obtain the following fixed-parameter result for feedforward neural networks and feedforward $\BNL$-programs.

\begin{theorem}
    Given a fixed floating-point format $\cF$ and fixed parameters for the activation functions, the following hold.
    \begin{itemize}
        \item An $n$-layer feedforward neural network $\cN$ with $N$ nodes, $L$ edges and degree $\Delta$ translates to the following equivalent feedforward $\BNL$-programs that have $\ordo(n\log(\Delta))$ layers:
        \begin{itemize}
            \item an unbalanced feedforward $\BNL$-program of size $\ordo(N+L)$ that is also synchronous if $\cN$ is synchronous,
            \item a balanced synchronous feedforward $\BNL$-program of size $\ordo((N+L)^2)$.
        \end{itemize}
        \item For each regular activation function $\alpha$, a feedforward $\BNL$-program of size $s$, where the schema depth of each layer $i$ is $d_{i}$, translates into the following equivalent $\ordo(\sum_{i = 1}^{n-1} d_{i})$-layer feedforward neural networks with $\alpha$ at each node:
        \begin{itemize}
            \item a feedforward neural network with $\ordo(s)$ nodes that is also synchronous if $\Lambda$ is balanced and synchronous,
            \item a synchronous feedforward neural network with $\ordo(s^{2})$ nodes.
        \end{itemize}
    \end{itemize}
\end{theorem}

\section{Further characterizations of neural networks}\label{sec: corollaries}

In this section, we present a number of important corollaries. The most interesting of these is Theorem~\ref{theorem:NN_to_NN}, which shows that any (heterogeneous) neural network can be transformed into another equivalent{} (homogeneous) neural network where each node has the same activation function with only polynomial blow-up in size and polylogarithmic delay. We further show in Corollary \ref{corollary:NN_to_simple_NN} that the activation function can be chosen to be very simple.
On the other hand, we show in Corollaries \ref{corollary:NN_to_SC} and \ref{corollary:SC_to_NN} that our characterization for neural networks with $\BNL$ can alternatively be obtained with $\SC$ 
and in Corollaries \ref{corollary:NN_to_SFC} and \ref{corollary: SFC to NN} we show the same for self-feeding circuits. 
The characterization with $\SC$ is interesting because recently its extension $\msc$ was used to characterize message passing circuits \cite{ahvonen_journal} and the graded extension $\GMSC$ of $\msc$ was used to characterize recurrent graph neural networks~\cite{gnn_neurips}. 
We also show translations between $\SC$ and self-feeding circuits in Corollary \ref{corollary:SC_SF_circuit}.

In the corollaries of this section, equivalence is always obtained with respect to a canonical similarity described in Section \ref{subsection:equivalence_of_NNs}. When translating neural networks to programs or self-feeding circuits (Corollaries \ref{corollary:NN_to_SC} and \ref{corollary:NN_to_SFC}), we use a canonical similarity from floating-point strings to bit strings. Conversely, when translating programs or self-feeding circuits to neural networks (Corollaries \ref{corollary:SC_to_NN} and \ref{corollary: SFC to NN}), we use a canonical similarity from bit strings to floating-point strings that aligns with the chosen activation function as specified in Section \ref{BNL to NN}. Finally, when translating neural networks into neural networks (Theorem \ref{theorem:NN_to_NN} and Corollaries \ref{corollary:NN_to_NN_fixed} and \ref{corollary:NN_to_simple_NN}), we use the composition of two canonical similarities of the above form, one from floating-point strings to bit strings and the other from bit strings to floating-point strings.

First, we recall some auxiliary notations from Section \ref{sec: fnn}.
Let $\cN$ be any general neural network for any floating-point format $\cF(p,q,\beta)$ (resp., $\cF_{\mathsf{n}}$) with $N$ nodes, 
$L$ edges,
degree $\Delta$, piece-size $P$.
Let $\Omega$ be the order of $\cN$ (setting $\Omega := 2$ if the order is $1$ or $0$) and let $r = \max\{p,q\}$. Let $K_{\cN}$ and $T_{\cN}$ denote the size and delay from Theorem \ref{NN_to_BNL}, i.e. we define 
$K_{\cN} \colonequals (N P \Omega + L) (r^4 + r^3\beta + r^{2}\beta^2 + r\beta^3)$ 
and $T_{\cN} \colonequals \log(\Omega)(\log(r) + \log(\beta)) + \log(\Delta)$.

Our first corollary states that $\SC$-programs and self-feeding circuits translate to each other. We can translate an $\SC$-program into a self-feeding circuit by first applying Theorem \ref{SC_BNL} and then Theorem~\ref{thrm:BNL_to_SF_circuit}. Note that Theorem \ref{SC_BNL} increases the depth of the program, but only by a constant.
We can also translate a self-feeding circuit into an $\SC$-program by applying Theorem \ref{thrm:SF_circuit_to_BNL} and then Theorem \ref{SC_BNL}.
\begin{corollary}\label{corollary:SC_SF_circuit}
    Given an $\SC$-program of size $s$ and depth $\ell$, we can construct an equivalent{} self-feeding circuit with bounded fan-in and $1$ precomputation round either of size $\ordo(s)$ and depth $\ell + \ordo(1)$ or of size $\ordo(s^2)$ and depth $\ordo(\log s)$.

    Respectively, given a self-feeding circuit $C$ with size $n$, fan-in $k$, depth $d$ and $m$ edges, we can construct an equivalent{} $\SC$-program of size $\ordo(n+m)$ and depth $\ordo(k)$ where the computation delay of the program is $\ordo(d)$. Moreover, if $C$ has bounded fan-in, then the size of the program is $\ordo(n)$ and the depth of the program is $\ordo(1)$.
\end{corollary}

The first corollary on neural networks, where we translate a neural network into an $\SC$-program, is obtained by first applying Theorem \ref{NN_to_BNL} and then Theorem \ref{SC_BNL}. 
\begin{corollary}\label{corollary:NN_to_SC}
    For any general neural network $\cN$, we can construct an equivalent{} $\SC$-program $\Lambda$ of size $\ordo(K_{\cN})$ where the computation delay of $\Lambda$ is $\ordo(T_{\cN})$.
\end{corollary}

The second corollary on neural networks, where we translate an $\SC$-program into a general neural network, is obtained by first applying Theorem \ref{SC_BNL} and then Theorem \ref{BNL_to_NN}. Note that the single precomputation round introduced by Theorem \ref{SC_BNL} becomes $\ordo(d)$ precomputation rounds, since the neural network requires $\ordo(d)$ rounds to simulate one round of the program. 
Recall below that a ``regular activation function'' refers to activation functions $\alpha \colon \fF \to \fF$ that satisfy Axiomatization~\ref{A}, and the equivalent neural network for $\fF$ is then constructed with respect to the corresponding canonical similarity.
\begin{corollary}\label{corollary:SC_to_NN}
    Let $\alpha$ be any regular activation function and
    let $\Lambda$ be an $\SC$-program of size $s$ and depth $d < s$. We can construct an equivalent{} general neural network $\cN_\Lambda$ that uses $\alpha$ (at each node), where the number of nodes in $\cN_{\Lambda}$ is $\ordo(s)$ and the computation delay of $\cN_{\Lambda}$ is $\ordo(d)$ with $\ordo(d)$ precomputation rounds.
\end{corollary}

The first corollary between self-feeding circuits and neural networks, where we translate neural networks into self-feeding circuits, is obtained by applying Theorems \ref{NN_to_BNL} and \ref{thrm:BNL_to_SF_circuit}.
\begin{corollary}\label{corollary:NN_to_SFC}
    For any general neural network $\cN$, we can construct an equivalent{} self-feeding circuit~$C_{\cN}$ with bounded fan-in and computation delay $\ordo(T_{\cN})$ either of size $\ordo(K_{\cN})$ or of size $\ordo(K_{\cN}^2)$. In the latter case, the depth of $C_{\cN}$ is $\ordo(\log (K_{\cN}))$.
\end{corollary}

The second corollary between self-feeding circuits and neural networks, where we translate self-feeding circuits into neural networks, is obtained by applying Theorems \ref{thrm:SF_circuit_to_BNL} and \ref{BNL_to_NN}. 
\begin{corollary}\label{corollary: SFC to NN}
    Let $\alpha$ be any regular activation function and 
    let $C$ be a self-feeding circuit with size $n$, depth $d$, fan-in $k$ and $m$ edges. We can construct an equivalent{} general neural network $\cN_C$ that uses $\alpha$ (at each node), where $\cN_C$ has $\ordo(n+m)$ nodes and computation delay $\ordo(dk)$. Moreover, if $C$ has bounded fan-in then $\cN_C$ has $\ordo(n)$ nodes and computation delay $\ordo(d)$.
\end{corollary}

It is easy to obtain similar corollaries that would extend the results on feedforward neural networks from Section \ref{sec: fnn} to $\SC$-programs and self-feeding circuits. In the case of $\SC$, this would mean defining feedforward $\SC$-programs, but a feedforward self-feeding circuit is simply an ordinary Boolean circuit.

As an interesting result we obtain the following theorem, which states that we can always translate a (possibly heterogeneous) neural network into an equivalent (homogeneous) neural network that uses any regular activation function $\alpha$ (at each node) with polynomial size blow-up.
In the case of feedforward neural networks, we can translate any feedforward neural network into an equivalent (synchronous) feedforward neural network.\footnote{Note that we have not defined the concept of computation delay in terms of feedforward neural networks. 
Naturally, if a given feedforward neural network has $n$ layers, then an equivalent feedforward neural network with computation delay $T$ has exactly $Tn$ layers.}
Again note that the equivalent neural networks are constructed w.r.t. any corresponding canonical similarity. 
Below, the notion of a parameter being \textbf{polynomial} (or resp. \textbf{polylogarithmic}) \textbf{in the size parameters of a neural network $\cN$} means that it is polynomial (resp. polylogarithmic) 
in the following parameters: the number of nodes, degree, number of pieces, order, fraction precision, exponent precision and base. 
\begin{theorem}\label{theorem:NN_to_NN}
    Let $\cN$ be a general neural network. For any regular activation function $\alpha$, we can translate $\cN$ to an equivalent general neural network $\cN'$ which uses $\alpha$ (at each node) and:
    \begin{itemize}
        \item the number of nodes in $\cN'$ is polynomial in the size parameters of $\cN$,
        \item the computation delay of $\cN'$ is polylogarithmic in the size parameters of $\cN$,
        \item if $\cN$ is feedforward, then $\cN'$ is feedforward (and even synchronous).
    \end{itemize}
\end{theorem}
\begin{proof}
    Given a general neural network for a floating-point format or normalized format $\fF$,
    we first translate $\cN$ into an equivalent{} $\BNL$-program $\Lambda$ (w.r.t. a canonical similarity $\sigma_1$ from $\fF$ to $\{0,1\}$),
    where the size of $\Lambda$ is polynomial in the size parameters of $\cN$ and the computation delay of $\Lambda$ is polylogarithmic in the size parameters of $\cN$ via Theorem \ref{NN_to_BNL}. Then, we apply Lemma \ref{lem:BL_to_circuit} (Theorem~4 from \cite{size-depth-tradeoff}) to each rule of $\Lambda$ to obtain a strongly equivalent $\BNL$-program $\Lambda'$, where for each rule of size $n$ in $\Lambda$, the new rule in $\Lambda'$ is of size quadratic in $n$ and depth $\ordo(\log(n))$; the size of the resulting program is thus quadratic in the size of $\Lambda$ (still polynomial in the size parameters of $\cN$) where the depth of each rule is logarithmic in the size of $\Lambda$ (thus polylogarithmic in the size parameters of $\cN$, because of standard laws of logarithms). Then we use Theorem \ref{BNL_to_NN} to obtain a neural network $\cN'$ for a floating-point format or normalized format $\fF'$ that is equivalent{} 
    to $\Lambda'$ (w.r.t. a canonical similarity~$\sigma_2$ from $\{0,1\}$ to $\fF'$), where the number of nodes in $\cN'$ is linear in the size of $\Lambda'$ (polynomial in the size parameters of $\cN$) and the computation delay of $\cN'$ is linear in the depth of $\Lambda'$ (polylogarithmic in the size parameters of $\cN$).
    Thus, $\cN'$ is also equivalent{} to $\cN$ (w.r.t. the canonical similarity $\sigma_2 \circ \sigma_1$ from~$\fF$ to $\fF'$), 
    where the computation delay of $\cN'$ is polylogarithmic in the size parameters of $\cN$.
    The case is similar for feedforward neural networks by applying Theorems \ref{theorem:FFNN_to_FFBNL} and \ref{theorem:FFBNL_to_FFNN}, and Lemma~\ref{lem:BL_to_circuit}.
\end{proof}
In particular, Theorem \ref{theorem:NN_to_NN} applies to linear activation functions of the form $y = kx$ by Axiomatization~\ref{A}. 
A folklore result states that recurrent neural networks with linear activation functions and reals instead of floats are essentially linear dynamical systems. The relationship between recurrent neural networks and linear dynamical systems is studied, for example, in \cite{valente2022probing}.
Furthermore, with respect to feedforward networks, a well-known folklore result \cite{Sharma2020ACTIVATIONFI} states that feedforward neural networks with linear activation functions that carry out computations with real numbers correspond to a linear regression model. 
Theorem \ref{theorem:NN_to_NN} demonstrates, in stark contrast to these folklore results, that at least with floating-point numbers any recurrent (resp. feedforward) neural network with non-linear activation functions can be translated into an equivalent recurrent (resp. feedforward) neural network with linear activation functions. 

Moreover, we obtain the following corollary in the scenario where the parameters of the floating-point format and activation functions are fixed.
\begin{corollary}\label{corollary:NN_to_NN_fixed}
    Let $\cN$ be a general neural network with $N$ nodes,
    $L$ edges
    and degree $\Delta$. Assuming that the parameters of floating-point formats and activation functions are fixed, for any regular activation function $\alpha$, we can translate $\cN$ to a general neural network $\cN'$ which uses $\alpha$ (at each node) and:
    \begin{itemize}
        \item $\cN'$ has 
        $\ordo((N + L)^2)$ 
        nodes (or 
        $\ordo((N + L)^4)$ 
        if $\cN$ is feedforward and we want $\cN'$ to be synchronous),
        \item if $\cN$ is not feedforward, then the computation delay of $\cN'$ is 
        $\ordo(\log(N + L)\log(\Delta))$.
        \item if $\cN$ is feedforward with $n$ layers, then $\cN'$ is feedforward with $\ordo(n \log(N + L)\log(\Delta))$ layers.
    \end{itemize}
\end{corollary}

As noted, restrictions of $\mathrm{ReLU}$ to \emph{any} floating-point format or normalized float format satisfy Case~1 of Axiomatization \ref{A}. 
Similarly, we can note that restrictions of the identity function to \emph{any} float format or normalized float format satisfy Case 2 of Axiomatization \ref{A}.
Thus, we obtain the following corollary.
Note that below when we refer to $\mathrm{ReLU}$ or the identity function, we refer to one which is restricted to the float format (or normalized float format) that is used by the constructed neural network.
\begin{corollary}\label{corollary:NN_to_simple_NN}
    We can translate a general neural network $\cN$ to a general neural network $\cN'$ which either uses $\mathrm{ReLU}$ (at each node) or the identity function (at each node), and:
    \begin{itemize}
        \item the number of nodes in $\cN'$ is polynomial in the size parameters of $\cN$,
        \item the computation delay of $\cN'$ is polylogarithmic in the size parameters of $\cN$,
        \item if $\cN$ is feedforward, then $\cN'$ is feedforward (and even synchronous).
    \end{itemize}
\end{corollary}

\section{Conclusion}

We have shown a strong equivalence between a general class of neural networks and Boolean network logic in terms of discrete time series. The translations are simple in both directions, with reasonable time and size blow-ups. 
Also, similar translations are obtained for feedforward neural networks via a feedforward-style fragment of Boolean network logic.
We receive similar results for the logic $\SC$ and self-feeding circuits. 
The link to self-feeding circuits is novel, since it allows us to apply circuit-based methods to reason about neural networks in the recurrent setting.
Our translations from Boolean network logic to neural networks is given w.r.t. an axiomatization that covers piecewise polynomial floating-point approximations of many generally used activation functions, including linear functions and even the identity function. The result for linear functions is in contrast to the folklore result that this necessarily makes the neural network a linear dynamical system, or in the case of feedforward neural networks, a linear regression model.
Interesting future directions involve investigating extensions with randomization, as well as studying the effects of using alternatives to floating-point numbers, such as, for example, fixed-point arithmetic.

\subsection*{Acknowledgments}
The list of authors on the first page is given based on the alphabetical order. 
Antti Kuusisto was supported by the Academy of Finland project Theory of computational logics, grant numbers 352419, 352420, 353027, 324435, 328987. Damian Heiman was supported by the same project, grant number 353027. Antti Kuusisto was also supported by the Academy of Finland consortium project Explaining AI via Logic (XAILOG), grant number 345612. Damian Heiman was also supported by the Magnus Ehrnrooth Foundation. Veeti Ahvonen was supported by the Vilho, Yrjö and Kalle Väisälä Foundation of the Finnish Academy of Science and Letters.

\bibliography{references}

\end{document}